%% file: bootstrap_after_selection.tex
\documentclass[aos, preprint]{imsart}
\usepackage{amsfonts}
\usepackage{amsmath}
\usepackage{enumerate}
\usepackage{mathabx}
\usepackage{listings}
\usepackage{graphicx}
\usepackage{color}

\usepackage{bm}
\usepackage{verbatim}
\usepackage{enumitem}

\include{preamble}

\RequirePackage{natbib}

\usepackage{hyperref}
\hypersetup{colorlinks, 
			citecolor=blue, 
			linkcolor=blue, 
			urlcolor=black}

\newcommand{\bb}[1]{\bm{#1}}
\usepackage[toc,page]{appendix}
\newtheorem{Example}{Example}[section]
\RequirePackage[OT1]{fontenc}

\usepackage[space]{grffile}
\usepackage{subfig}

\begin{document}

\title{Bootstrap inference after using multiple queries for model selection}

\begin{aug}
\author{\fnms{Jelena} \snm{Markovic}\corref{}\thanksref{t1}\ead[label=e1]{jelenam@stanford.edu},  
\fnms{Jonathan} \snm{Taylor}\thanksref{t2}\ead[label=e2]{jonathan.taylor@stanford.edu}
}
\runauthor{Markovic and Taylor}
\affiliation{Stanford University}
\address{Department of Statistics\\ Stanford University 
\\ Sequoia Hall \\ Stanford, CA 94305, USA \\ \printead{e1} \\
\printead*{e2} } 
\end{aug}

\thankstext{t1}{Supported by Stanford Graduate Fellowship.}
\thankstext{t2}{Supported in part by National Science Foundation grant DMS-1208857 and Air Force Office of Sponsored Research grant 113039.}

\begin{abstract}

In this work, we provide a refinement of the selective CLT result of \cite{tian2015selective}, which allows for selective inference in non-parametric settings by adjusting for the asymptotic Gaussian limit for selection.
Under some regularity assumptions on the density of the randomization, including heavier tails than Gaussian satisfied by e.g.~logistic distribution, we prove the selective CLT holds without any assumptions on the underlying parameter, allowing for rare selection events. We also show a selective CLT result for Gaussian randomization, though the quantitative results are qualitatively different for the Gaussian randomization as compared to the heavier tailed results.

Furthermore, we propose a bootstrap version of this test statistic, which is provably asymptotically pivotal uniformly across a family of non-parametric distributions. This result can be interpreted as resolving the impossibility results of \cite{leeb2006can}. We describe several sampling methods involving the projected Langevin Monte Carlo to compute the bootstrapped test statistic and the corresponding confidence intervals valid after selection.

The applications of our work include valid inferential and sampling tools  after running various model selection algorithms including their combinations into multiple views/queries framework. We also present a way to do \textit{data carving}, providing more powerful tests than classical data splitting by reusing the information in the data from the first stage.
	
\end{abstract}

\maketitle
\setattribute{journal}{name}{}


\section{Introduction}

This work continues a line of research in inference after model selection beginning with \cite{berk_valid_2013, lee2013exact, sequential_post_selection, fithian2014optimal, tian2015selective}. 
This train of thought leads to the concept of a {\em selective model} in which each distribution $\mathbb{F}_n$ in some statistical model $\mathcal{F}_n$ is conditioned on the output $M$ of some model selection procedure, a canonical example being the choice of variables and their signs by the LASSO \citep{tibs_lasso}. Let us denote by $\widehat{M}(\bm S,\bm\omega)$ some possibly randomized selection procedure which we might think of as a {\em query}, i.e.~a function we evaluate on data $\bb S\sim\mathbb{F}_n$ along with possible independent randomization $\bm\omega\sim\mathbb{G}$. Also, let us denote with $\mathbb{F}_n^*$ the post-selection distribution of the data, i.e.~distribution of $\bm S$ conditional on selection $\widehat{M}(\bm S,\bm\omega)=M$. The selective model $\mathcal{F}_n^*$ is determined by the outcome $M$ of the procedure $\widehat{M}$ and the corresponding selective likelihood ratio
\begin{equation*}
	\ell_{\mathbb{F}_n}^M(\bb S) = \frac{d\mathbb{F}^*_n}{d\mathbb{F}_n}({\bb S}) = \frac{(\mathbb{F}_n\times\mathbb{G})\{\widehat{M}(\bm S,\bm\omega)=M \:\bigl \vert\: {\bb S}\}}{(\mathbb{F}_n\times\mathbb{G})\{\widehat{M}(\bm S,\bm\omega)=M\}}.
\end{equation*}
(We remove $M$ from the superscript in later sections for simplicity.) In the numerator above we marginalize only over the possible randomization while the numerator becomes an indicator function in the case where there is no added randomization. The behavior of the random variables $(\ell_{\mathbb{F}_n}^M)_{\mathbb{F}_n\in\mathcal{F}_n}$ can be used to determine properties of the selective model 
\begin{equation}  \label{eq:intro:selective:model}
	\mathcal{F}^*_n = \left\{\mathbb{F}_n^*: \frac{d\mathbb{F}^*_n}{d\mathbb{F}_n}(\cdot) = \ell_{\mathbb{F}_n}^M(\cdot), \:\mathbb{F}_n \in \mathcal{F}_n \right\},
\end{equation}
consisting of all distributions $\mathbb{F}_n^*$ we get by conditioning $\mathbb{F}_n$ on the selection event for given $M$. For example, \cite{tian2015asymptotics} show that consistency under a sequence of {\em unselective} (or pre-selective, treating $M$ as fixed) models $\mathcal{F}_n$ can be transferred to a corresponding sequence of selective models $\mathcal{F}^*_n$ under conditions on the appropriate selective likelihood ratios. Similarly, questions of weak convergence for sequences of selective models in $\mathcal{F}_n^*$ can be related to weak convergence under the unselective models $\mathcal{F}_n$.
It is these questions of weak convergence that are the focus of this work. In particular, we construct selective bootstrap procedures that produce asymptotically pivotal quantities which converge weakly in a uniform sense over a large class of models.

Before detailing our contributions here, we review other related work in selective inference. In \cite{lee2013exact}, the authors use the LASSO to select variables and provide valid tests and confidence intervals for the parameters chosen based on looking at the active set of LASSO. To achieve that, the authors base inference using the conditional distribution of the data, where conditioning is  on the result of the selection by the LASSO. Assuming Gaussian errors with known variance (the assumptions on $\mathbb{F}_n$), the truncated Gaussian test statistic constructed in \cite{lee2013exact, sequential_post_selection} is an exact pivot valid after selection. It is constructed for saturated models, meaning that no relationship is assumed between the response vector and the predictors before or after selection  \citep{lee2013exact, loftus2015selective, tibshirani2015uniform}. In the saturated model framework, the selection is done to adjust for the choice of the parameter to report and there are no additional assumption on $\mathbb{F}_n$ coming after observing $M$.

Followup work, including \cite{fithian2014optimal, tian2015selective} among others, point out that this model may not be realistic in many situations. However, the principle of conditioning on the result of the selection by the LASSO is applicable in other statistical models besides the saturated model. 
In the selected model framework, we assume some relationship between the response and the covariates after looking at the outcome of the model selection algorithm \citep{fithian2014optimal}. In other words, the selected model framework allows for adding assumptions on the underlying data generating distribution $\mathbb{F}_n$ after looking at the selection outcome $M$ as long as the inference stage is done using the conditional distribution $\mathbb{F}_n^*$. In this framework inference is done by conditioning on less than in saturated framework, hence increasing power. In this scenario, the conditioning on the selection region adjusts for both choosing a model for $\mathbb{F}_n$ after selection and choosing the parameters to test.

Some recent work in selective inference has been focusing on removing the Gaussian assumption on the data and getting asymptotic results for truncated Gaussian test statistic in non-parametric settings \citep{tian2015asymptotics, tibshirani2015uniform}. \cite{tibshirani2015uniform} also propose a bootstrap version of this test that is asymptotically conservative in the unconditional sense, i.e.~under the pre-selection distribution of the data $\mathbb{F}_n$. All the works mentioned have been focusing on the non-randomized selective inference, meaning that given the data there is no additional randomness coming in either the model selection procedure or the inferential procedures.

\cite{tian2015selective} propose doing selection on a randomized version of the response vector or adding randomization directly into the objective function of a model selection algorithm. They notice a significant increase in the power of a test with added randomization. By adding randomization, the test statistic they construct becomes a smoother function of the data compared to the truncated Gaussian test statistic. 
Their test statistic, we call \textit{plugin Gaussian or plugin CLT pivot}, is constructed by adjusting for selection a test statistic that satisfies a pre-selection CLT, treating $M$ as fixed. \cite{tian2015selective} prove a selective CLT result, saying the plugin Gaussian pivot is asymptotically $\textnormal{Unif}[0,1]$ under the conditional distribution $\mathbb{F}_n^*$ of the data in the uniform sense across a non-parametric family of distributions $\mathcal{F}_n^*$. Since the conditional statements (under $\mathbb{F}_n^*$) imply the unconditional ones (under $\mathbb{F}_n$), the guarantees under the conditional distributions are stronger. 

In this work, we build on the main ideas of randomized selective inference, allowing for both selected and saturated model setting. We extend the non-parametric results of \cite{tian2015selective} and also introduce one more construction of the test statistic using bootstrap.

\subsection{Contributions of this paper}

From theoretical perspective our contributions are as follows.
\begin{enumerate}[leftmargin=*, label=(T\arabic*)]

\item \textbf{Selective CLT with Lipschitz randomization.} Under some regularity conditions on the randomization, we first prove selective CLT without assuming local alternatives (defined in Section \ref{sec:selective:clt}), relaxing the result of \cite{tian2015selective}. This allows for the parameter of interest to be arbitrarily far from the selection region, allowing for selection events that are relatively rare under the true data generating mechanism. This result requires the randomization to have tails heavier than Gaussian such as logistic distribution.

\item \textbf{Selective CLT with Gaussian randomization.} Our second theoretical result is the proof that the selective CLT holds with Gaussian randomization under the local alternatives. This means if the distance of the parameter vector to the selection region is not growing too fast, the constructed pivot is asymptotically uniform hence valid for inference.

\item \textbf{Selective bootstrap.} Building on the non-parametric results, we also propose a bootstrap version of the asymptotically pivotal test statistic constructed in \cite{tian2015selective}. We prove that our bootstrap test statistic is asymptotically pivotal after randomized model selection procedures with affine constraints under local alternatives. The results are also under the conditional distribution $\mathbb{F}_n^*$ of the data (conditioning is on the selection region) and in the uniform sense across a family $\mathcal{F}_n^*$ of non-parametric distributions. We refer to the constructed bootstrap test statistic as the \textit{bootstrap pivot}.

\end{enumerate}

We develop two novel samplers for computing the bootstrap pivot.

\begin{enumerate}[leftmargin=*, label=(C\arabic*)]
	\item \textbf{Wild bootstrap sampler.} In addition to the challenges related to sampling from a post-selection distribution, computing bootstrap pivot requires using the bootstrapped samples for the data. Under the selective density, not all bootstrapped samples are equally likely so the standard resampling with replacement techniques do not obviously work. To provide an efficient way of computing the bootstrap pivot, we use the wild bootstrap of \cite{wu1986jackknife, mammen1993bootstrap} coupled with the selective sampler of \cite{tian2016magic, selective_sampler}. This allows us to sample the weights along with the optimization variables from a continuous density with constraints.
	
	\item \textbf{Weighted optimization sampler.} Since the above sampler requires running the sampling chain separately for each test, we devise an efficient way of constructing both pivots and confidence intervals for multiple tests, e.g.~all selected coefficients in a regression. It relies on reusing the optimization samples across tests thus the sampling has to be done only once.
	
\end{enumerate}

Important applications of this work are as follows.

\begin{enumerate}[leftmargin=*, label=(A\arabic*)]
\item \textbf{One view / query on the data.} We develop an efficient way to report valid inference for the selected coefficients based on observing the outcome of a randomized model selection procedure such as LASSO, GLMs with $\ell_1$ penalty, marginal screening etc. 

\item \textbf{Multiple views / queries of the data.} 
As most data analysts will want to try various model selection algorithms when choosing a model, we present a way to construct confidence intervals after multiple views/queries of the data. 
Any of the procedures above can be combined into multiple views of the data framework, where we choose the selected coefficients based on the outcomes of several model selection procedures. This allows a statistician to do inference after GLMs with $\ell_1$-penalties, forward-stepwise, marginal screening etc.~or any of their combinations.

\item \textbf{Data carving.} We introduce a novel way to do data splitting through \textit{data carving} \citep{fithian2014optimal}. Classical data splitting uses a part of the data to select a model (stage one) and the leftover data for inference (stage two). Hence, the classical data splitting conditions on the whole first stage data used for selection. In this work, we select the model using the stage one data as well, but by conditioning only on the model selected, we use the whole data for inference. Conditioning on less, the data carving procedure provides an increase in power compared to the traditional data splitting while preserving the right coverage for the selected coefficients.

\item \textbf{Multiple splits.} Finally, the multiple views framework is combined with the data splitting framework into the \textit{multiple splits}, a way to do inference after looking at the models selected from several splits. 

\end{enumerate}

\subsection{Outline}

The outline of the rest of this paper is as follows.
We illustrate the main methodological and computational aspects of our proposal through two examples of one view on the data. The first one describes how to do inference for the mean after randomization selection on the response without covariates (simple example, Section \ref{sec:simple:example}). The second one is about inference after looking at the outcome of the randomized LASSO procedure (Section \ref{sec:lasso}). The examples are followed by the general setup of randomized selective inference and computational methods \ref{sec:gen:setup}.


We provide the selective CLT results for plugin Gaussian pivots for both heavy-tailed randomization and the Gaussian randomization in Section \ref{sec:selective:clt}. In section \ref{sec:consistency}, we show that under heavy-tailed randomization, the estimators consistent before selection (under $\mathbb{F}_n$, treating $M$ as fixed) are also consistent post-selection (under $\mathbb{F}_n^*$, treating $M$ as random).
In Section \ref{sec:bootstrap}, we propose a bootstrap statistic and show it is asymptotically pivotal.

Further, in Section \ref{sec:mv:general}, we present a general way of doing inference after multiple views of the data, illustrated via two examples: forward-stepwise and running several $\ell_1$-penalized logistic regressions. In Section \ref{sec:data:carving}, we describe the novel computational methods for data carving. Combining this technique with the multiple views of the data, we present the tools to do inference after multiple splits (Section \ref{sec:multiple:splits}). We advocate the projected Langevin method, which enjoys theoretical guarantees described in \cite{bubeck2015sampling}, for sampling from the log-concave density with simple constraints. The sampling details are in Section \ref{app:sampling} in the supplement.


\section{Simple example: inference for the mean after randomized selection} \label{sec:simple:example}

We start with an example where we do inference for the mean of i.i.d.~random variables where the data vector is selected such that the randomized $t$-statistic is above some given threshold.
The data vector $\bb y=(y_1, \ldots, y_n)$ consists of $y_i\overset{i.i.d.}{\sim}\mathbb{F}_n$, $i=1,\ldots,n$, pre-selection. The subscript $n$ in $\mathbb{F}_n$ indicates that the data generating distribution might vary with $n$. $\mathbb{F}_n^n$ denotes the joint distribution of $n$ i.i.d.~samples from $\mathbb{F}_n$. For simplicity, let us assume the variance of $\mathbb{F}_n$ is 1 and denote the mean of $\mathbb{F}_n$ as $\mu(\mathbb{F}_n)$. We want a test for the mean $H_0:\mu(\mathbb{F}_n)=\mu$ after selection of the following form
\begin{equation*}
	\sqrt{n}\bar{y}+\omega > \mathfrak{t}, \;\; (\bb y, \omega)\sim\mathbb{F}_n^n\times\mathbb{G},
\end{equation*}
where $\bar{y}=\frac{1}{n}\sum_{i=1}^n y_i$, $\mathfrak{t}$ is some fixed (and known) threshold. $\omega\sim\mathbb{G}$ is added randomization independent of the data and the distribution $\mathbb{G}$ on $\mathbb{R}$ with the density denoted as $g$ is known. Denote with $\widebar{G}$, the survival function of $\omega$, i.e.~$\widebar{G}(x)=\mathbb{G}\{\omega>x\}$. Let the post-selection distribution of the data vector $\bb y$ be $\mathbb{F}_n^*$. In other words, $\mathbb{F}_n^*$ is the marginal distribution of the data vector $\bm y$, where we marginalize over the randomization $\omega$, in the joint selective distribution of $(\bm y,\omega)\sim\mathbb{F}_n^n\times\mathbb{G}$ given $\sqrt{n}\bar{y}+\omega>\mathfrak{t}$. Let us also denote as $F_{\mu}(t)$ the CDF of $\sqrt{n}(\bar{y}-\mu)$, where $\bb y\sim\mathbb{F}_n^*$.

\cite{leeb2006can} show that for the normal distribution $\mathbb{F}_n=\Phi_{(\mu, 1)}$ one cannot estimate $F_{\mu}(t)$ uniformly consistently, i.e.~for a fixed $t$ any estimate $\widehat{F}(t) $ that is consistent at $\mu=0$,
\begin{equation} \label{eq:simple:example:consistency}
	\forall \delta>0, \;\;\underset{n\rightarrow\infty}{\lim}\mathbb{F}_n^*\left\{|\widehat{F}(t) - F_{0}(t)|>\delta\right\}=0,
\end{equation}
necessarily fails to be consistent uniformly in the neighborhood around zero. In other words, $\widehat{F}(t)$ that satisfies \eqref{eq:simple:example:consistency} also satisfies
\begin{equation} \label{eq:simple:example:uniformity:fail}
	\underset{n\rightarrow\infty}{\textnormal{lim inf}} \underset{|\mu|<\frac{h}{\sqrt{n}}}{\sup} \mathbb{F}_n^*\left\{|\widehat{F}(t)-F_{\mu}(t)|>\delta\right\}\geq c>0,
\end{equation}
for some constants $h, \delta, c>0$. The local alternatives in \eqref{eq:simple:example:uniformity:fail} can be replaced by $|\mu|<R$ for some constant $R$ as well. The authors show the result for non-randomized selection events but the argument holds with randomized selection events as well. Hence, no uniformly consistent estimator of the conditional distribution $F_{\mu}(t)$ exists, not even locally.

Despite their result, we can still get a uniformly consistent confidence interval for the mean $\mu$ by inverting a test. If we construct a test statistic $\widehat{F}_{\mu}(t)$ such that
\begin{equation} \label{eq:simple:ex:uniform:test}
	\forall\delta>0,\;\; \underset{n\rightarrow\infty}{\lim}\:\underset{|\mu|<\frac{h}{\sqrt{n}}}{\sup} \mathbb{F}_n^*\left\{\underset{t\in\mathbb{R}}{\sup}\left|\widehat{F}_{\mu}(t)-F_{\mu}(t)\right|>\delta \right\} = 0,
\end{equation}
then we can construct a confidence interval for $\mu$ by inverting this test statistic. Using the fact that $F_{\mu}(\sqrt{n}(\bar{y}-\mu))$ is exactly distributed as $\textnormal{Unif}[0,1]$, \eqref{eq:simple:ex:uniform:test} implies that  under $\mathbb{F}_n^*$ the test statistic $\widehat{F}_{\mu}(\sqrt{n}(\bar{y}-\mu))$ is asymptotically uniformly $\textnormal{Unif}[0,1]$, i.e.
\begin{equation} \label{eq:simple:examaple:selective:CLT}
	\underset{n\rightarrow\infty}{\lim}\:\underset{|\mu|<\frac{h}{\sqrt{n}}}{\sup} \:\underset{t\in[0,1]}{\sup}\left|\mathbb{F}_n^* \left\{\widehat{F}_{\mu}(\sqrt{n}(\bar{y}-\mu))\leq t\right\}-t\right| = 0.
\end{equation}
We refer to the statement in \eqref{eq:simple:examaple:selective:CLT} as the selective CLT.  
This implies the constructed $(1-\alpha)$-confidence interval	
$\mathcal{I}_{\alpha}(\bb y) = \left\{\mu: \widehat{F}_{\mu}\left(\sqrt{n}(\bar{y}-\mu)\right)\in\left[\alpha/2, 1-\alpha/2\right]\right\}$
is uniformly valid under the conditional distribution of the data, i.e.
\begin{equation} \label{eq:simple:ex:uniformly:consistent:intervals}
	\underset{n\rightarrow\infty}{\lim} \:\underset{|\mu|<\frac{h}{\sqrt{n}}}{\sup} \:\left|\mathbb{F}_n^*\{\mu\in \mathcal{I}_{\alpha}(\bb y)\}-(1-\alpha)\right|=0,
\end{equation}
for a given level $\alpha$ (follows from Lemma A.1 in \cite{romano2012uniform}).
Our goal is to construct $\widehat{F}_{\mu}(t)$ that satisfies \eqref{eq:simple:ex:uniform:test} using bootstrap.


\begin{remark1}
A consistent $(1-\alpha)$-confidence interval $\mathcal{I}_{\alpha}(\bm y)$ for $\mu$ is defined as 
	\begin{equation*}
		\underset{n\rightarrow\infty}{\lim}\left|\mathbb{F}_n^*
		\{\mu\in \mathcal{I}_{\alpha}(\bm y)\}-(1-\alpha)\right|=0
	\end{equation*}
	across a single sequence of distributions $\{\mathbb{F}_n\}_{n=1}^{\infty}$.
Having uniformly consistent confidence intervals and not just consistent is an important inferential goal since it makes a difference in coverage guarantees for finite $n$ \citep{romano2012uniform, tibshirani2015uniform}. Given $\epsilon>0$, there exists $n(\epsilon)$ such that for every $n\geq n(\epsilon)$ the coverage of uniformly consistent intervals is guaranteed at level $1-\alpha-\epsilon$ no matter what the underlying $\mu$ is. For consistent confidence intervals, this might not hold since the required sample size will depend on $\mu$ as well.
\end{remark1}

\begin{remark1}	
	The non-regularity problems of this kind (parameter at the boundary) have been also considered in \cite{andrews2000inconsistency, laber2014dynamic, mckeague2015adaptive}. \cite{mckeague2015adaptive} provide a bootstrap test for testing the global null in the non-randomized version of the model selection problem. We consider a similar problem of performing inference after marginal screening in Section \ref{sec:marginal:screening} in the supplement. As pointed out in \cite{leeb_comment}, their test is for a point null, hence the non-regularity problem is not an issue  anymore \citep{leeb2006can, leeb2006performance}. We provide a similar fix, i.e.~circumventing the non-regularity issue by providing a test, instead of doing inference based on a cumulative distribution function of an estimator (that is known to provide non-uniform limiting behavior \citep{leeb2006performance}). However, our test is not just for testing a global null but can be used more generally for testing any linear combination of the underlying parameter vector. As for the construction of confidence intervals, the test of \cite{mckeague2015adaptive} is hard to invert, while we construct confidence intervals as well.
	
\end{remark1}

\subsection{Plugin Gaussian pivot}

Before describing the bootstrap versions of $\widehat{F}_{\mu}(t)$, let us first construct a Gaussian version for $\widehat{F}_{\mu}(t)$ denoted as
\begin{equation*}
	\widehat{F}_{\mu}^G(t)=(\Phi_{(0,1)}\times\mathbb{G})\left\{Z\leq t\:|\:Z+\sqrt{n}\mu+\omega>\mathfrak{t}\right\},
\end{equation*}
where the probability on the RHS is under $Z\times\omega\sim\Phi_{(0,1)}\times \mathbb{G}$. $\widehat{F}_{\mu}^G(t)$ represents the conditional CDF of the data if the data was normal pre-selection, i.e.~in the case $\mathbb{F}_n=\Phi_{(\mu,1)}$. We refer to $\widehat{F}^G_{\mu}(\sqrt{n}(\bar{y}-\mu))$, as the plugin Gaussian (or plugin CLT) pivot. \cite{tian2015selective} show that under heavy-tailed distribution $\mathbb{G}$, the selective CLT from \eqref{eq:simple:examaple:selective:CLT} holds with the local alternatives assumption on $\mu$ as stated $|\mu|<h/\sqrt{n}$ for a constant $h$. Our results in Section \ref{sec:selective:clt} extend their results to show that in cases when $\mathbb{G}$ is heavy-tailed we can remove the local alternatives assumption so that the uniformity statement is across $\mu\in\mathbb{R}$. In the case when $\mathbb{G}$ is Gaussian, we keep the local alternatives to prove selective CLT. Thus, the plugin Gaussian pivot leads to an asymptotically uniformly valid test and a confidence interval for $\mu$ post-selection across a family of non-parametric distributions $\mathbb{F}_n$.

Computationally, we present two  sampling ways of computing the plugin Gaussian pivot. The first approach is a more natural method for the particular setting of the simple example, also used to illustrate the selective sampler. The second approach is more efficient in more complicated settings so we present it here to convey the idea.

\begin{itemize}[leftmargin=*]

\item (Selective sampler) It suffices to sample $Z$ and $\omega$ from the selective density proportional to
\begin{equation*}
	\phi_{(0,1)}(Z)\cdot g(\omega)\cdot\mathbb{I}_{\{Z+\sqrt{n}\mu+\omega>\mathfrak{t}\}},
\end{equation*}
 where $\phi_{(0,1)}(\cdot)$ is the density of the standard normal distribution. Using a simple change of variables technique, this sampling density can be simplified so that the constraint set does not involve both $Z$ and $\omega$. The change of variables here is a special example of the pull back trick of \cite{selective_sampler} that we refer to as the \textit{selective sampler}. Denoting $v=Z+\omega+\sqrt{n}\mu$, the selective density of $(Z,v)$ becomes proportional to
\begin{equation} \label{eq:simple:example:selective:density}
	\phi_{(0,1)}(Z) \cdot g(v-Z-\sqrt{n}\mu)\cdot\mathbb{I}_{\{v>\mathfrak{t}\}}.
\end{equation}
Sampling from the density in \eqref{eq:simple:example:selective:density} becomes easier since the density does not involve any restrictions on the data variable $Z$. Having the $Z$ samples from the density above, computing  the plugin Gaussian pivot $\widehat{F}^G(\sqrt{n}(\bar{y}-\mu))$ is straightforward. To compute the confidence interval for $\mu$ based on this test, we run the sampler once at a reference value for $\mu$, usually taken to be MLE pre-selection ($\bar{y}$ in this example). Then we tilt the $Z+\sqrt{n}\mu$ samples at the reference to get the post-selection $Z+\sqrt{n}\mu$ samples at other $\mu$ values.

\item (Weighted optimization sampler) We write the selective density of $(Z,v)$ as proportional to
	\begin{equation*}
		\phi_{(0,1)}(Z) \cdot g(v-\sqrt{n}\bar{y}^{obs})\cdot\mathbb{I}_{\{v>\mathfrak{t}\}} \cdot\frac{g(v-Z-\sqrt{n}\mu)}{g(v-\sqrt{n}\bar{y}^{obs})},	
	\end{equation*}
	where $\sqrt{n}\bar{y}^{obs}$ is the observed value. We sample $Z\sim\Phi_{(0,1)}$ and the optimization variable $v$ from a density proportional to
\begin{equation} \label{eq:simple:ex:optimization:density}
	g(v-\sqrt{n}\bar{y}^{obs})\cdot\mathbb{I}_{\{v>\mathfrak{t}\}}
\end{equation}
and independently from $Z$ samples.
By tilting the independent $(Z,v)$ samples by the ratio $w(Z,v)=\frac{g(v-Z-\sqrt{n}\mu)}{g(v-\sqrt{n}\bar{y}^{obs})}$ we get the samples from the selective density.
Precisely, to compute the pivot for a given $\mu$ we do the following.
\begin{enumerate}[leftmargin=*]
	\item (Optimization sampler) Get the optimization samples $v^s$, $s=1,\ldots, S$ from the density in \eqref{eq:simple:ex:optimization:density}. We use projected Langevin for this step.
	\item (Gaussian sample) Sample $Z^s$, $s=1,\ldots, S$, from a normal distribution $\mathcal{N}(0,1)$.
	\item (Tilting / Weighting) We weight the combined samples $(Z^s,v^s)$ from the first two steps to compute the plugin Gaussian pivot as
	\begin{equation*}
		\sum_{s=1}^S \mathbb{I}_{\{Z^s+\sqrt{n}\mu\leq\sqrt{n}\bar{y}^{obs}\}} \cdot \frac{w(Z^s,v^s)}{\sum_{s'=1}^Sw(Z^{s'}, v^{s'})}.	
	\end{equation*}
	To compute the confidence intervals, we need to repeat only the third step above for $\mu$ values over a grid.
\end{enumerate}

\end{itemize}

\subsection{Bootstrap pivot}

Using bootstrap, we approximate $F_{\mu}(t)$ by
\begin{equation*}
	\widehat{F}^B_{\mu}(t)=(\widehat{\mathbb F}_n^n\times\mathbb{G})\left\{\sqrt{n}(\bar{y}^*-\bar{y})\leq t\:\big|\:\sqrt{n}(\bar{y}^*-\bar{y})+\sqrt{n}\mu+\omega>\mathfrak{t} \right\},
\end{equation*}
where $\widehat{\mathbb{F}}_n$ is the empirical distribution of our sample $\bb y$. $\bm y^*=(y_1^*,\ldots, y_n^*)\sim\widehat{\mathbb{F}}_n^n$ denotes a bootstrap sample and $\bar{y}^*$ is its sample mean. The asymptotically pivotal test statistic in this case becomes $\widehat{F}^B_{\mu}(\sqrt{n}(\bar{y}-\mu))$ and we refer to it as the \textit{bootstrap pivot}. Our result in Section \ref{sec:selective:clt} proves that for more general affine selection events the constructed bootstrap pivots are asymptotically uniformly $\textnormal{Unif}[0,1]$ after selection under some assumptions, thus can be used for inference post-selection.

We present three ways of computing the bootstrap pivot. The first method is the most natural for the simple example; however, since the method is hard to generalize we also provide two bootstrap samplers that are applicable in more complicated regression examples.

\begin{itemize}[leftmargin=*]

\item (Bootstrap samples adjusted for selection probabilities) Using the standard sampling with replacement, we bootstrap the data vector $\bm y$ and compute the bootstrapped mean samples as $\bar{y}^{*b}$, $b=1, \ldots, B$. Then we compute $\widehat{F}_{\mu}^B(t)$ as
\begin{equation} \label{eq:simple:ex:nonparam:boot:test}
	\frac{\sum_{b=1}^B\widebar{G}\left(\mathfrak{t}-\sqrt{n}\mu-\sqrt{n}(\bar{y}^{*b}-\bar{y})\right)\cdot \mathbb{I}_{\{\sqrt{n}(\bar{y}^{*b}-\bar{y})\leq t\}}}{\sum_{b=1}^B\widebar{G}\left(\mathfrak{t}-\sqrt{n}\mu-\sqrt{n}(\bar{y}^{*b}-\bar{y})\right)}.
\end{equation}
We can think of computing the above quantity as weighting the bootstrap samples with the corresponding selection probabilities  $\widebar{G}(\cdot)$.
To perform a test for the mean we use $\widehat{F}^B_{\mu}(\sqrt{n}(\bar{y}-\mu))$ as the test statistic and its asymptotically uniform distribution as the reference distribution under the null. 
The confidence intervals can be constructed by inverting this test. This involves a grid search for $\mu$ such that $\widehat{F}_{\mu}(\sqrt{n}(\bar{y}-\mu))\in[\alpha/2,1-\alpha/2]$. This is essentially the grid bootstrap of \cite{hansen1999grid}; however, in our example we do not have to resample the data (to compute the test statistics) for each $\mu$ on the grid, but we can sample once and reuse samples in computing the test statistic for each $\mu$. This makes the construction of confidence intervals faster. 
Since in more complicated examples computing the selection probabilities efficiently is hard, we devise two alternative ways to do bootstrap.

\item  (Wild bootstrap sampler) Without selection, the wild bootstrap approximates the distribution of $\sqrt{n}(\bar{y}-\mu)$ with the distribution of $\frac{1}{n}\sum_{i=1}^n(y_i-\bar{y})\alpha_i$, where $\bb\alpha=(\alpha_1,\ldots, \alpha_n)$ is a vector of bootstrap weights. We take $\alpha_i\overset{i.i.d.}{\sim}\mathbb{H}_{\alpha}$, $i=1,\ldots, n$, with density $h_{\alpha}$. The wild bootstrap approximation of $F_{\mu}(t)$ becomes
\begin{align}
	&(\mathbb{H}_{\alpha}^n\times\mathbb{G})\left\{\frac{1}{\sqrt{n}}\sum_{i=1}^n(y_i-\bar{y})\alpha_i\leq t\:\Big|\:\frac{1}{\sqrt{n}}\sum_{i=1}^n(y_i-\bar{y})\alpha_i+\sqrt{n}\mu+\omega>\mathfrak{t}\right\} \nonumber \\ 
	&=\frac{\mathbb{E}_{\bb\alpha\sim\mathbb{H}_{\alpha}^n}\left[\widebar{G}\left(\mathfrak{t}-\sqrt{n}\mu-\frac{1}{\sqrt{n}}\sum_{i=1}^n(y_i-\bar{y})\alpha_i \right) \mathbb{I}_{\left\{\frac{1}{\sqrt{n}}\sum_{i=1}^n(y_i-\bar{y})\alpha_i\leq t\right\}}\right]}{\mathbb{E}_{\bb\alpha\sim\mathbb{H}_{\alpha}^n}\left[\widebar{G}\left(\mathfrak{t}-\sqrt{n}\mu-\frac{1}{\sqrt{n}}\sum_{i=1}^n(y_i-\bar{y})\alpha_i\right)\right]}. \label{eq:simple:ex:wild:boot:cdf}
\end{align}
In this example we can just resample $\bb\alpha$ from $\mathbb{H}_{\alpha}^n$ or numerically integrate over the numerator and denominator in \eqref{eq:simple:ex:wild:boot:cdf}. However, since our approach to computing the bootstrap pivot in more complicated selection events in higher dimensions involves MC techniques, we illustrate that approach here as well.
A sampling approach to computing the wild bootstrap test statistic involves sampling weights $\bb\alpha\in\mathbb{R}^n$ from the density proportional to
\begin{equation} \label{eq:simple:example:wild:bootstrap:density}
	\left(\prod_{i=1}^n h_{\alpha}(\alpha_i)\right)\cdot \widebar{G}\left(\mathfrak{t}-\sqrt{n}\mu-\frac{1}{\sqrt{n}}\sum_{i=1}^n(y_i-\bar{y})\alpha_i\right).
\end{equation}
We then compute the bootstrap pivot by computing the quantile of $\sqrt{n}(\bar{y}-\mu)$ with respect to the empirical samples $\frac{1}{\sqrt{n}}\sum_{i=1}^n(Y_i-\bar{y})\alpha_i$. To compute the confidence interval for $\mu$, we do the sampling once at a reference value for $\mu$ and then we do Gaussian tilt of the samples $\frac{1}{\sqrt{n}}\sum_{i=1}^n(Y_i-\bar{y})\alpha_i+\sqrt{n}\mu$ to get the pivots at other $\mu$ values.
In more complicated examples later the sampling density will also involve constrained optimization variables.

\item (Weighted optimization sampler) We reuse the weighted optimization sampler from computing the plugin Gaussian pivot while changing the $Z$ samples in the second step there there with the bootstrap samples $\sqrt{n}(\bar{y}^*-\bar{y})$.

\end{itemize}

\section{LASSO with random design}  \label{sec:lasso}

\subsection{Notation}

In this and the following examples we assume the correlation model where the design matrix $\bm X\in\mathbb{R}^{n\times p}$, with rows $\bm x_i^T$, $i=1,\ldots, n$, is random. We denote the response vector as $\bb y=(y_1,\ldots,y_n)\in\mathbb{R}^n$. Using a randomized procedure on $(\bm X, \bb y)$, we select a subset $E\subset\{1,\ldots, p\}$ of predictors. We denote the pre-selection distribution of the data as $(\bm x_i, y_i)\overset{i.i.d.}{\sim}\mathbb{F}_n$, $i=1,\ldots, n$, with the density denoted as $f_n$. Let the entries of matrix $\bb X$ be scaled by $1/\sqrt{n}$.

We provide inference for the population quantities corresponding to the selected predictors $E$. Note that when we talk about the selected model in general we use the notation $M$ to denote the model. In concrete examples presented, our model is determined by the selected predictors which we denote as $E$. Let us denote as $\bm X_E\in\mathbb{R}^{n\times |E|}$ the submatrix of $\bm X$ consisting of the selected columns from $E$ only.
In the case of Gaussian loss, the parameter of interest equals $\bm\beta_E^*=\bm\beta_E^*(\mathbb{F}_n)=(\mathbb{E}_{\mathbb{F}_n}\left[\bm X^T_E\bm X_E\right])^{-1}\mathbb{E}_{\mathbb{F}_n}\left[\bm X^T_E\bm y\right]$ and in the case of logistic loss $\bm\beta_E^*$ satisfies $\mathbb{E}_{\mathbb{F}_n}[\bm X_E^T(\bm y-\pi_E(\bm\beta_E^*))]=0$, where $\pi_E(\bm\beta)=\frac{\exp(\bm X_E\bm\beta)}{1+\exp(\bm X_E\bm\beta)}$ for any $\bm\beta\in\mathbb{R}^{|E|}$. Precisely, our hypothesis is $H_0:\bm A^T\bm\beta_E^*=\bm\theta$ for a given matrix $\bm A\in\mathbb{R}^{a\times |E|}$. The goal is to do valid inference for this linear combination of $\bm\beta_E^*$ taking into account that $E$ is computed based on data.

To do inference for $\bm\beta_E^*$ pre-selection (treating $E$ is fixed), we can use $\bar{\bm\beta}_E$, the MLE of the model $\bm y\sim \bm X_E$. In the case of Gaussian likelihood, the MLE becomes $\bar{\bm\beta}_E=(\bm X_E^T\bm X_E)^{-1}\bm X_E^T\bm y$ and in the case of logistic likelihood, it satisfies $\bm X_E^T(\bm y-\pi_E(\bar{\bm\beta}_E)) = \bm 0$. Pre-selection, the MLE is asymptotically Gaussian so we use its normal distribution as the reference distribution. It is worth repeating that we use the terms pre-selection, before selection or in the original distribution $\mathbb{F}_n$ to say that we treat the selected model $E$ as fixed (non-random). 

\begin{remark1}
We should emphasize that in order to do valid inference we do not assume the selected model is true data generating mechanism but we construct confidence intervals for the population parameters corresponding to the selected model. Also, upon looking at the outcomes of the model selection procedures, an analyst decides on her own model for which to report inference.
\end{remark1}

To adjust for selection, the inference is done under the conditional distribution of $\bar{\bm\beta}_E$, where the conditioning is on the event that the randomized procedure applied to $(\bm X,\bb y)$ chooses model $E$. The conditional distribution of the data is denoted as $\mathbb{F}_n^*$. The pivots constructed based on $\bar{\bm\beta}_E$ and the selection event have guarantees under $\mathbb{F}_n^*$, meaning that in selective inference results we account for the fact that $E$ is random.

Before describing the LASSO selection in detail, let us introduce more notation. Usually, the selection event imposes constraints on not only $\bar{\bm\beta}_E$ but also on additional data vectors.
In the LASSO selection, the selection event can be written in terms of $\bb D=\begin{pmatrix} \bm D_E \\ \bm D_{-E} \end{pmatrix}=\begin{pmatrix} \bar{\bm\beta}_E \\ \bm X_{-E}^T(\bm y-\bm X_E\bar{\bm\beta}_E)\end{pmatrix}$ in the case of Gaussian loss and in terms of $\bb D=\begin{pmatrix} \bar{\bm\beta}_E \\ \bm X_{-E}^T(\bm y-\pi _E(\bar{\bm\beta}_E)\end{pmatrix}$ in the case of logistic loss.

\subsection{Problem setup}

The randomized LASSO \citep{tian2015selective, selective_sampler} applied to $(\bm X ,\bb y)$ solves
\begin{equation} \label{eq:lasso:objective}
	\hat{\bb \beta}(\bb X,\bb y, \bb\omega)=\underset{\bm{\beta}\in\mathbb{R}^p}{\textnormal{argmin}}\:\frac{1}{2}\left\|\bb y-\bm X\bm\beta \right\|_2^2+\lambda\|\bm\beta\|_1+\frac{\epsilon}{2}\|\bb\beta\|_2^2-\bb\omega^T\bm\beta, \;\;\; ((\bb X,\bb y),\bb\omega)\sim\mathbb{F}_n^n\times\mathbb{G},
\end{equation}
where $\bb\omega\sim \mathbb{G}$ represents the noise from a known distribution $\mathbb{G}$ on $\mathbb{R}^p$ with density $g$. The term $\frac{\epsilon}{2}\|\bb\beta\|_2^2$ for some constant $\epsilon>0$ is added in the objective to make sure the solution exists (see \cite{selective_sampler}). Taking $\epsilon$ to be of order $O(1/\sqrt{n})$ corresponds to performing randomized LASSO and taking $\epsilon$ to be $O(1)$ corresponds to a randomized elastic net \citep{elastic_net}. Assume the randomized LASSO selects $(E,\bm s_E)$, where $E\subset\{1,\ldots,p\}$ is the candidate set of variables and $\bm s_E\in\{\pm 1\}^{|E|}$ are their signs. We write the solution of \eqref{eq:lasso:objective} as $\hat{\bb \beta}=\hat{\bb \beta}(\bb X, \bb y,\bb \omega)=\left(\begin{matrix}
	\hat{\bb\beta}_E \\ \bb 0\end{matrix}\right)$.

In order to get the distribution of a chosen test statistic for testing a linear combination $\bm\beta_E^*$ under the null, we need to sample data form the distribution conditional on LASSO in \eqref{eq:lasso:objective} selecting model $E$. 
 The selection event consists of all such data and randomization pairs such that solving the randomized LASSO \eqref{eq:lasso:objective} for that pair gives the same $(E, \bb s_E)$: 
$$
\mathcal{S}_{(E,\bm s_E)}=\left\{(\bb X',\bb y',\bb\omega'): \hat{\bb\beta}_{-E}(\bb X',\bb y', \bb \omega')=0,\; \textnormal{sign}(\hat{\bb \beta}_E(\bb X',\bb y',\bb \omega')) = \bb s_E \right\}.
$$
Sampling the data and the randomization from this region is hard due to the complicated constraints.

Using the pull-back measure trick of \cite{selective_sampler}, we re-parametrize a complicated constraints set of data and randomization $\mathcal{S}_{(E,\bm s_E)}$ using the so called optimization variables, naturally arising random variables in the problems of interest, that along with the data describe the selection region.
	Now instead of sampling data and the randomization variables from a much more complicated region, we sample data and the optimization variables from a much simpler region with the constraints on the optimization variables only.
In the randomized LASSO problem presented the optimization variables are $(\bb \beta_E, \bm u_{-E})\in\mathbb{R}^{|E|}\times\mathbb{R}^{p-|E|}$, where $\bm\beta_E$ corresponds to the active part of the minimizer $\hat{\bm\beta}$ and $\bm u_{-E}$ corresponds to the inactive part of the sub-gradient of the penalty. 
The optimization variables are chosen such that we can recover $\bb\omega$ by the sub-gradient equation of \eqref{eq:lasso:objective} as $\bb\omega = \omega(\bb X,\bm y,\bb \beta_E,\bm u_{-E})$ for some function $\omega$. Instead of sampling data and the randomization variables from a complicated set of joint constraints, we sample data and the optimization variables with a simpler set of constraints only on the optimization variables.
The sampling density of $(\bb X,\bb y,\bb \beta_E,\bb u_{-E})$ is then proportional to
\begin{equation} \label{eq:lasso:density:Xy}
	\left(\prod_{i=1}^n f_n(\bb x_i,y_i)\right)\cdot g\left(-\bb X^T\bb y+
	\begin{pmatrix}
	\bb X_E^T\bb X_E+\epsilon \bb I_{|E|} \\ \bb X_{-E}^T\bb X_E	 \end{pmatrix}\bb \beta_E 
	+\lambda \begin{pmatrix}
		\bb s_E \\ \bb u_{-E} \end{pmatrix} \right)
\end{equation}
supported on $\mathbb{R}^{n\times p}\times \mathbb{R}^n\times\mathbb{R}_{\bb s_E}^{|E|}\times [-1,1]^{p-|E|}$, where $\mathbb{R}_{\bb s_{E}}^{|E|}$ denotes the orthant in $\mathbb{R}^{|E|}$ corresponding to signs $\bb s_E$. Usually we do not know $f_n$ explicitly but for doing the inference on $\bb\beta_E^*$ we sample from a simpler space using the pre-selection asymptotic distributions as follows.

The only randomness coming from the data in \eqref{eq:lasso:objective} is in the gradient of the loss, which can be expressed in terms of the vector $\bb D$.
Pre-selection (treating $E$ as fixed), this vector is asymptotically normally distributed $\bb D\rightarrow \mathcal{N}_p\left(\bm\mu_{\bm D}, \bb \Sigma_{\bb D} \right)$,
as $n\rightarrow\infty$, under $(\bm X,\bm y)\sim\mathbb{F}_n$. In terms of $\bb D$, the sub-gradient equation from solving \eqref{eq:lasso:objective} becomes
\begin{equation*}
	\omega(\bb D, \bb \beta_E, \bb u_{-E}) 
	=\bm M\bm D+\bm B\bm\beta_E+\bm U\bm u_{-E}+\bm L,
\end{equation*}
with $\textnormal{sign}(\bb \beta_E)=\bb s_E$ and $\|\bb u_{-E}\|_{\infty}\leq 1$, where
\begin{equation*}
	\bm M=-\begin{pmatrix} \bb X_E^T\bm X_E & \bm 0 \\ \bm X_{-E}^T\bm X_E & \bb I_{p-|E|} \end{pmatrix},\; \bm B= \begin{pmatrix}
\bb X_E^T\bb X_E+\epsilon \bb I_{|E|} \\ \bb X_{-E}^T\bb X_E	 \end{pmatrix}, \; \bm U=\begin{pmatrix}
	\bm 0 \\ \lambda\bm I_{p-|E|}
\end{pmatrix}, \; \bm L = \begin{pmatrix}
	\lambda\bm s_E \\ \bm 0
\end{pmatrix}
\end{equation*}
with $\bm I$ denoting the identity matrix of the dimension in the subscript. Note that with the abuse of notation we use $\omega(\cdot)$ to denote the randomization reconstruction map in terms of both $(\bm X,\bm y,\bm\beta_E,\bm u_{-E})$ and $(\bm D,\bm\beta_E,\bm u_{-E})$.
 Hence, instead of sampling from the density in \eqref{eq:lasso:density:Xy}, we alternatively sample $(\bb D,\bb\beta_E, \bb u_E)$ using the asymptotic normality of $\bb D$. The selective density of $(\bm D, \bm\beta_E,\bm u_{-E})$ becomes proportional to
\begin{equation*}
	\phi_{(\bm\mu_{\bm D},\bm\Sigma_{\bm D})}(\bm D)\cdot g(\omega(\bm D,\bm\beta_E,\bm u_{-E}))\cdot\mathbb{I}_{\{\textnormal{sign}(\bm\beta_E)=\bm s_E\}}\cdot\mathbb{I}_{\{\|\bm u_{-E}\|_{\infty}\leq 1\}}.
\end{equation*}

\begin{remark1}
We can further reduce the dimension of optimization variables from $\dim(\bm\beta_E)+\dim(\bm u_{-E})=p$ to $|E|$ by marginalizing $\bm u_{-E}$ from the selective density above. That is doable explicitly given that $\mathbb{G}$ consists of independent coordinates. For further details see \cite{selective_sampler}.
\end{remark1}

\subsection{Linear decomposition}

The change of variables technique above addresses the difficulties in sampling from the selective density of the data conditional on the randomized selection region. Another computational issue is sampling the relevant part of the data vector corresponding to the selected parameter of interest, while conditioning on the part of the data vector corresponding to the nuisance parameters. This is done via \textit{linear decomposition}.

When testing the hypothesis $H_0:\bm A^T\bb\beta_E^*=\bm\theta$, we use $\|\bm T-\bm\theta\|_2^2$ as a test statistic, where $\bm T=\bm A^T\bar{\bm\beta}_E$ denotes the so called \textit{target statistic}. To have this test be valid pre-selection (treating $E$ as non-random), we use asymptotic normality of $\bm T$ to determine a reference distribution. Post-selection, however, we need to base our inference under the post-selection distribution of $\bm T$.
 Note that we can do inference for any parameter $\theta(\mathbb{F}_n, E)=\bm\theta$, assuming the target statistic $\bm T$ for $\bm\theta$ and the data vector $\bm D$ satisfy the pre-selection CLT:
\begin{equation*}
	\begin{pmatrix} \bm T \\ \bm D\end{pmatrix}\rightarrow\mathcal{N}_{a+p}\left(\begin{pmatrix} \bm\theta \\ \bm\mu_{\bm D} \end{pmatrix}, \begin{pmatrix}\bm\Sigma_{\bm T} & \bm\Sigma_{\bm T,\bm D}\\ \bm\Sigma_{\bm D,\bm T} &\bm\Sigma_{\bm D} \end{pmatrix}\right)
\end{equation*}
as $n\rightarrow\infty$. Denote $\bm F = \bm D-\widehat{\bm\Sigma}_{\bm D, \bm T}\widehat{\bm\Sigma}_{\bm T}^{-1}\bm T$, where $\widehat{\bm\Sigma}_{\bm D,\bm T}$ and $\widehat{\bm\Sigma}_{\bm T}$ are the estimates of the respective covariances. By decomposing $\bm D=\bm F+\bm T$ and by conditioning on $\bm F$, it suffices to sample $(\bb T, \bb\beta_E, \bb u_{-E})$ in order to get the post-selection distribution of the target statistic $\bm T$ under the null.
The plugin CLT sampling density of $(\bm T,\bm\beta_E,\bm u_{-E})$ is proportional to
\begin{equation} \label{eq:lasso:density:plugin:clt}
	\phi_{(\bm\theta,\widehat{\bm\Sigma}_{\bm T})}(\bm T) \cdot g\left(\widetilde{\bm M}\bm T +\bm B \bb\beta_E+\bm U\bb u_{-E} +\widetilde{\bm L}\right),
\end{equation}
supported on $\mathbb{R}^p\times \mathbb{R}^{|E|}_{\bb s_E}\times [-1,1]^{p-|E|}$, where $\widetilde{\bm M} = \bm M\widehat{\bm\Sigma}_{\bm D,\bm T}\widehat{\bm\Sigma}_{\bm T}^{-1}, \; \widetilde{\bm L} = \bm L+\bm M\bm F.$
Here $\phi_{(\cdot,\cdot)}$ is the density of multivariate normal with the respective mean and covariance matrix in the subscript. We can estimate $\bb\Sigma_{\bb D,\bb T}$ and $\bm\Sigma_{\bm T}$ either parametrically or non-parametrically using pairs bootstrap \citep{freedman1981bootstrapping, buja2014conspiracy}.

\begin{remark1}
Note that we do not assume the selected linear model $\bm y\sim\bm X_E$ is true, i.e.~we do not assume $\mathbb{E}[\bm\varepsilon|\bm X]=0$, where $\bm\varepsilon=\bm y-\bm X_E\bm\beta_E^*$ are the true residuals. In other words, observing $E$ does not impose additional assumptions about the underlying $\mathbb{F}_n$ (saturated model framework). If we further assume the selected linear model is true (selected model framework), we could condition on less information. In the case the linear model is true, the asymptotic mean of $\bm D$ is $\begin{pmatrix} \bm\beta_E^* \\ \bm 0 \end{pmatrix}$, so we condition on $\bm D_E-\bm\Sigma_{\bm D_E,\bm T}\bm\Sigma_{\bm T}^{-1}\bm T$ and marginalize over the null statistic $\bm D_{-E}$. We focus on the saturated model in this work although the inferential tools are the same in the selected model approach as well.
\end{remark1}

\subsection{Challenges in computing the bootstrap pivot}

Let us mention the challenges in using the standard non-parametric bootstrap of \cite{efron_bootstrap, freedman1981bootstrapping} in this setting. Pairs bootstrap resamples $(\bb X_i^*, y_i^*)$, $i=1,\ldots, n$, from the empirical distribution of the data. This is equivalent to sampling $b_i$ from a multinomial with equal probabilities $1/n$ at each $(1,\ldots,n)$ and setting $(\bb X_i^*, y_i^*)=(\bb X_{b_i}, y_{b_i})$ for each $i=1,\ldots, n$. Denote with $\bb X^*(\bb b)$ the matrix with rows $\bb X_i^*=\bb X_{b_i}$, $i=1,\ldots, n$, and with $\bb y^*(\bb b)=(y_1^*,\ldots, y_n^*)$.
To do bootstrap after selection, i.e.~after conditioning on the observed model, we would need to sample $(\bb b, \bb\beta_E, \bb u_{-E})\in \{1,\ldots, n\}^n\times \mathbb{R}^{|E|}\times\mathbb{R}^{p-|E|}$ such that the model selection algorithm for the bootstrapped data $(\bb X^*({\bb b}), \bb y^*({\bb b}))$ and $\omega(\bb X^*(\bb b),\bb y^*(\bb b), \bb \beta_E, \bb u_{-E})$ gives the same selected variables $E$. We assume that the bootstrap samples lie in the same space as the original data, which is true in this case. This couples $\bb b$ and $(\bb \beta_E,\bb u_{-E})$ so that their joint density is proportional to 
\begin{equation*}
	g\left(\omega(\bb X^*(\bb b),\bb y^*(\bb b), \bb \beta_E, \bb u_{-E})\right),
\end{equation*} 
and supported on $(\bm b,\bm \beta_E, \bm u_{-E})\in\{1,\ldots, n\}^n \times \mathbb{R}^{|E|}_{\bb s_E}\times[-1,1]^{p-|E|}$. 
Since sampling from this density is computationally hard, we devise more efficient ways of computing the pivot using bootstrap samples.

\begin{remark1}
There are some possible modifications of the bootstrap with replacement that sample $b_i$, $i=1,\ldots, n$, from the Poisson distribution instead of multinomial \citep{poission_boot}; however, using this bootstrap after selection still requires sampling from a partly discrete distribution with constraints which is computationally hard.
\end{remark1}

\begin{remark1}
	Computing the bootstrap pivot by weighting bootstrap samples with selection probabilities (the first approach to computing the bootstrap pivot in the simple example) is hard since we do not have an efficient way of computing the selection probabilities exactly. Using the techniques of \citep{selective_bayesian} the selection probabilities can be approximated, providing an alternative way of computing the pivot. We pursue this direction in future work.
\end{remark1}

\subsection{Wild bootstrap sampler}

One approach to solving the issue of using the bootstrap samples is to use a continuous version of bootstrap, e.g.~wild bootstrap instead of Efron's bootstrap. Assume the parameter of interest is $\bm\beta_E^*$ and the corresponding target statistic $\bm T=\bar{\bm\beta}_E$.
The wild bootstrap approximates the pre-selection distribution of $\bar{\bb\beta}_E-\bb\beta_E^*=(\bb X_E^T\bb X_E)^{-1}\bb X_E^T\bb \varepsilon$ with $(\bb X_E^T\bb X_E)^{-1}\bb X_E^T\textnormal{diag}(\hat{\bb \varepsilon})\bb\alpha$, where $\hat{\bb\varepsilon}=\bb y-\bb X_E\bar{\bb\beta}_E$ are the observed residuals from fitting OLS with response $\bb y$ and the predictors $\bb X_E$ and $\bb \alpha\in\mathbb{R}^n$ are bootstrap weights. We take $\alpha_i\overset{i.i.d.}{\sim} \mathbb{H}_{\alpha}$ with density $h_{\alpha}$ and support $\textnormal{supp}(\mathbb{H}_{\alpha})$. We replace $\bm T$ with $\bm T(\bm\alpha)+\bm\theta$ in the sampling density \eqref{eq:lasso:density:plugin:clt}, where $\bm T(\bm\alpha) = \bm A^T(\bm X_E^T\bm X_E)^{-1}\bm X_E^T\textnormal{diag}(\hat{\bm\varepsilon})\bm\alpha$.
The randomization reconstruction map in the bootstrap case becomes
\begin{equation} \label{eq:lasso:random:reconstruction:boot}
\begin{aligned}
	\omega^B(\bb \alpha,\bb \beta_E, \bb u_{-E}) =\widetilde{\bm M}\bm T(\bm\alpha)+\bm B\bm\beta_E+\bm U\bm u_{-E}+\widetilde{\bm L}+\widetilde{\bm M}\bm\theta.
\end{aligned}
\end{equation}
In this case the bootstrap density on $(\bb\alpha,\bb\beta_E,\bb u_{-E})$ is proportional to
\begin{equation} \label{eq:lasso:density:bootstrap}
	\left(\prod_{i=1}^n h_{\alpha}(\alpha_i)\right) \cdot g\left(\omega^B(\bb\alpha,\bb\beta_E, \bb u_{-E})\right)
\end{equation}
and supported on $(\textnormal{supp}(\mathbb{H}_{\alpha}))^n\times \mathbb{R}^{|E|}_{\bb s_E}\times[-1,1]^{p-|E|}$. We now use some of the standard Monte Carlo techniques to sample from the density above with constraints (see Section \ref{app:sampling} for sampling details including the projected Langevin updates).

\begin{remark1} \label{rmk:choosing:boot:weights}
For the wild bootstrap without selection to be consistent, it suffices to have the mean of $\mathbb{H}_{\alpha}$ to be 0 and the variance to be 1. The condition that the skewness of $\mathbb{H}_{\alpha}$ is also 1 has been introduced by \cite{liu1988bootstrap} to improve the rate of convergence of the bootstrap distribution. In practice, we use $\mathbb{H}_{\alpha}$ to be standard normal. 
\end{remark1}

\begin{remark1}
Besides the wild bootstrap, there are other possible continuous bootstrap versions such as the Bayesian bootstrap of \cite{rubin1981bayesian}, or more generally, the weighted bootstrap of \cite{weighted_bootstrap} that can be used here but we do not pursue them in practice.
\end{remark1}

\subsection{Efficient inference via weighted optimization sampler}

Although we can efficiently compute the pivot for testing $\theta(\mathbb{F}_n,E)=\bm\theta$ by sampling $(\bm T,\bm \beta_E, \bm u_{-E})$ from the density in \eqref{eq:lasso:density:plugin:clt} or by sampling $(\bm\alpha,\bm\beta_E,\bm u_{-E})$ from the density in \eqref{eq:lasso:density:bootstrap}, in cases when we want to perform multiple tests at once we need to run any of these samplers separately for each test. For example, in order to provide the selective confidence intervals for all the selected coefficients ${\bm\beta}_{E,j}^*$, $j\in E$, we need to run $|E|$ samplers setting the target $\bm T$ to be each of $\bar{\bm\beta}_{E,j}$, $j\in E$.
To make this more efficient, we use the weighted optimization sampler already introduced in the simple example. By sampling the optimization variables from the selective density that fixes the data at its observed value, we can reuse the same optimization samples across different tests. This allows us to run the sampler only once while providing inference for multiple tests at once.

Computing the plugin Gaussian pivots includes the following steps.

\begin{enumerate}[leftmargin=*]
\item (Selective sampler) Given $\bm D=\bm D^{obs}$, we sample the optimization variables $(\bm \beta_E, \bm u_{-E})$ given the observed data vector come from the density proportional to
	\begin{equation}
		g(\omega(\bm D^{obs}, \bm\beta_E, \bm u_{-E}))
	\end{equation}
	with the constraints on $(\bm\beta_E,\bm u_{-E})$. Denote the samples as $(\bm\beta_E^s, \bm u_{-E}^s)$, $s=1,\ldots,S$, where $S$ is the sample size.

\item (Gaussian samples) We sample target from its pre-selection normal distribution to get samples $\bm T^s\sim\mathcal{N}(\bm 0,\widehat{\bm\Sigma}_{\bm T})$, $s=1,\ldots, S$.

\item (Importance weighting) We tilt the combined samples $(\bm T^s+\bm\theta,\bm\beta_E^s,\bm u_{-E}^s)$, $s=1,\ldots, S$, from the first and the second step using the importance sampling by weighting each of the triples $(\bm T^s+\bm\theta,\bm\beta_E^s,\bm u_{-E}^s)$ with the ratio 
	\begin{equation*}
		w(\bm T^s,\bm\beta^s_E, \bm u_{-E}^s)=\frac{g(\widetilde{\bm M}(\bm T^s+\bm\theta)+\bm B\bm\beta_E^s+\bm U\bm u_{-E}^s+\bm L)}{g(\bm M\bm D^{obs}+\bm B\bm \beta_E^s+\bm U\bm u_{-E}^s+\bm L)}
	\end{equation*}
	to compute the plugin Gaussian pivot as
	\begin{equation*}
		\sum_{s=1}^S\mathbb{I}_{\{\|\bm T^s\|_2\leq \|\bm T^{obs}-\bm\theta\|_2\}}\cdot \frac{w(\bm T^s,\bm\beta^s_E, \bm u_{-E}^s)}{\sum_{s'=1}^{S'}w(\bm T^{s'},\bm\beta^{s'}_E, \bm u_{-E}^{s'})}.
	\end{equation*}
\end{enumerate}

In case when $\bm T$ is 1-dimensional and we are interested in computing the confidence interval for $\theta(\mathbb{F}_n, E)$ as well, we need to repeat the third step across $\bm\theta\in\mathbb{R}$ values to invert the pivot. In case when we have multiple tests we need to repeat the second and the third step above. 

Computing the bootstrap pivot includes changing the second step above to use bootstrap samples $\bm T^*-\bm T$ instead of Gaussian ones. In case when $\bm T=\bar{\bm\beta}_E$, the bootstrap version becomes $\bm T^*=(\bm X_E^{*T}\bm X_E^*)^{-1}\bm X_E^{*T}y^*$, where $(\bm X^*,\bm y^*)$ consists of $n$ rows resampled with replacement from $(\bm x_i,y_i)$, $i=1,\ldots,n$.


\section{General setup of randomized selective inference} \label{sec:gen:setup}

We describe a general framework of randomized selective inference that is the foundation for all our seemingly more complicated  examples. We do selection by solving a standard optimization problem involving penalized loss with added linear randomization term. We then compute the bootstrapped test statistic while adjusting for selection, hence accounting for/conditioning on the fact that we looked at the selected model to choose the coefficients for which we report $p$-values and confidence intervals.

Given data $\bb S\in\mathbb{R}^{n\times m}$ with rows $\bb S_i\overset{i.i.d}{\sim}\mathbb{F}_n$, $i=1,\ldots, n$, where the distribution $\mathbb{F}_n$ has a density $f_n$, we solve the randomized model selection algorithm, introduced in  \cite{tian2015selective}, \cite{selective_sampler}, of the following form
\begin{equation} \label{eq:gen:setup:objective}
	\hat{\bm\beta}=\hat{\bm\beta}(\bm S,\bm\omega)=\underset{\bb\beta\in\mathbb{R}^p}{\textnormal{argmin}}\:\ell(\bb\beta;\bb S)+\mathcal{P}(\bb\beta)-\bb\omega^T\bb\beta+\frac{\epsilon}{2}\|\bb \beta\|_2^2, \;\;\; \bb S\times\bb\omega \sim \mathbb{F}_n\times \mathbb{G},
\end{equation}
where $\ell(\bb \beta;\bb S)$ is the loss function and $\mathcal{P}(\bb \beta)$ is a penalty term. $\bb\omega$ is the added randomization, a random variable drawn from a known distribution $\mathbb{G}$ on $\mathbb{R}^p$ with density $g$. The solution of \eqref{eq:gen:setup:objective} is $\hat{\bm\beta}=\textnormal{prox}_{\frac{1}{\epsilon}(\ell(\cdot;\bm S)+\mathcal{P}(\cdot))}\left(\frac{\bm\omega}{\epsilon}\right)$. Assuming $\ell(\cdot;\bm S)+\mathcal{P}(\cdot)$ is proper and closed convex function its proximal map exists and it is unique for all arguments.

As in the LASSO example above, before running the objective we decide how to choose the selected model $\widehat{M}(\bb S,\bb\omega)=M$ based on the solution $\hat{\bm\beta}$.
After looking at the outcome $M$, the parameter of interest $\bm \theta=\theta(\mathbb{F}_n, M)$ is chosen based on $M$. In order to have valid inference on $\bm\theta$ after selection, we need to base our inference using the post-selection distribution of the data.

The selective density on $\bm S$ and $\bm \omega$ is proportional to $\left(\prod_{i=1}^nf_n(\bm S_i)\right)\cdot g(\bm\omega)\cdot\mathbb{I}_{\{(\bm S,\bm\omega)\in \mathcal{S}_M\}}$, where $\mathcal{S}_M=\{(\bm S',\omega'):\widehat{M}(\bm S',\omega')=M\}$ is the selection event. 
Using the selective sampler of \cite{selective_sampler}, we re-parametrize $\mathcal{S}_M$ in terms of data $\bb S$ and the induced optimization variables $\bb v\in\mathbb{R}^q$, which are problem dependent.
The sub-gradient equation of \eqref{eq:gen:setup:objective} can be written as $\bb\omega = \omega(\bb S,\bb v)$ for some function $\omega$. Also, the selected model usually can be described only through the optimization variables $\bb v$, hence $M=\widetilde{M}(\bb v)$. This implies that selecting model $M$ is equivalent to having the optimization variables $\bb v$ constrained to $\mathcal{V}_M=\{\bb v'\in\mathbb{R}^q: \widetilde{M}(\bb v')=M\}\subset\mathbb{R}^q$. In selective inference problems, this constraint set generally becomes simpler than $\mathcal{S}_M$, thus the re-parametrization of the data and randomization in terms of the data and the optimization variables is extremely useful for sampling purposes.
The selective density of $(\bb S,\bb v)$ is proportional to
\begin{equation*}
	\left(\prod_{i=1}^n f_n(\bb S_i)\right) \cdot g(\omega(\bb S, \bb v))\cdot |J(\bb S,\bb v)|,
\end{equation*}
supported on $\mathbb{R}^{n\times m}\times \mathcal{V}_M$, where $J(\bb S, \bb v)$ is the Jacobian coming from the change of density. In most of the problems we consider the Jacobian is a constant hence we do not need to compute it but there are selective inference problems with non-trivial Jacobian e.g.~the group LASSO \citep{selective_sampler}. The sampling from the density above requires knowing the distribution of the data $\mathbb{F}_n$ exactly or using some asymptotic distribution. 

Assume that we have an asymptotically Gaussian test statistic $\bm T=T(\bm S)$ used for testing $\theta(\mathbb{F}_n,M)$ pre-selection, treating $M$ as non-random.
To do post-selection inference for $\theta(\mathbb{F}_n,M)$ using a chosen test-statistic $\bm T=T(\bm S)$, it suffices to have only the post-selection distribution of $\bm T$ and not of the whole $\bm S$. As seen in the LASSO example, we might be able to sample $\bm T$ directly while conditioning in the sampler on nuisance statistics. Then we can reuse either the selective sampler, the wild bootstrap sampler or the weighted optimization sampler to compute the selective pivots.




\section{Selective CLT} \label{sec:selective:clt}

We describe the content of the selective CLT, the fundamental tool used to construct valid tests after selection. After selection, we decide on testing a linear combination of the functional $\bm\mu(\mathbb{F}_n, M)=\bm\mu_n$ that depends on the unconditional data generating distribution $\mathbb{F}_n$ and the selected model $M$. In the scenario $M$ was fixed and not chosen based on data, we assume we use $\bm D_n$ as a test statistic for testing $\bm\mu_n$ based on the pre-selection asymptotic Gaussianity of $\bm D_n-\bm\mu_n$.
In the LASSO example, $\bm\mu_n=\begin{pmatrix} \bm\beta_E^* \\ \mathbb{E}_{\mathbb{F}_n}[\bm X_{-E}^T(\bm y-\bm X_E^T\bar{\bm\beta}_E^*)]\end{pmatrix}$, $\bm D_n=\bm D$ (defined in Section \ref{sec:lasso}) and thus $\bm D_n-\bm\mu_n$ is asymptotically Gaussian, treating $E$ as fixed and not chosen based on the data.
In order to construct a test statistic based on $\bm\eta^T\bm D_n$ to test the functional $\bm\eta^T\bm\mu_n$ after we select the model $M$ and choose $\bm\eta$, we use a pivot, constructed in \cite{tian2015selective}. Based on their selective CLT result, this pivot is asymptotically $\textnormal{Unif}[0,1]$ under the conditional distribution of the data, treating $M$ as random.

\subsection{Pre-selection asymptotic linearity of the chosen test statistic}

Let us now provide the precise setup where we apply the selective CLT. Assume our data at step $n$ consists of i.i.d.~random vectors $\bm S_i, i=1,\ldots, n$, whose distribution is non-parametric distribution denoted as $\mathbb{F}_n$. For each $n$, we take $\mathbb{F}_n\in\mathcal{F}_n$, where $\{\mathcal{F}_n:n\geq 1\}$ is a sequence of families of probability distributions (the restrictions on $\mathcal{F}_n$ will be made later). Denote the data matrix consisting of rows $\bm S_i^T$, $i=1,\ldots, n$, as $\bm S$.

After selecting model $M$, we are interested in testing the functional $H_0:\bm\mu(\mathbb{F}_n,M)=\bm\mu_n\in\mathbb{R}^p$, where $p$ fixed and does not depend on $n$ throughout. Assume the test statistics  $\bm D_n=\bm D_n(\bm S)=(D_1(\bm S),\ldots, D_p(\bm S))\in\mathbb{R}^p$, is an asymptotically linear test statistic pre-selection (treating $M$ as non-random) defined as 
\begin{equation} \label{assumption:asymptotically:linear} \tag{AL}
	\bb D_n =\frac{1}{n}\sum_{i=1}^n\bb \xi_i+o_{\mathbb{F}_n}\left(\frac{1}{\sqrt{n}}\right),
\end{equation}
where $\bb\xi_i$ are measurable with respect to $\bb S_i$, i.e.~$\bm\xi_i=\xi(\bm S_i)$ for some function $\xi$. We assume
$\mathbb{E}_{\mathbb{F}_n}[\bb\xi_i]=\bb\mu_n$, $\textnormal{Var}_{\mathbb{F}_n}(\bm\xi_i)=\bm\Sigma_n$ and the third moment of $\bm\xi_i$ is bounded uniformly for all $i=1,\ldots, n$. As a part of the \eqref{assumption:asymptotically:linear} condition, we assume that $\frac{1}{n}\sum_{i=1}^n(\bm \xi_i-\bm\mu_n)$ satisfies a uniform CLT across $\mathbb{F}_n\in\mathcal{F}_n$, i.e.
\begin{equation} \label{eq:uniform:clt:xi}
	\underset{n\rightarrow\infty}{\lim}\:\underset{\mathbb{F}_n\in\mathcal{F}_n}{\sup}\:\underset{\bm t\in\mathbb{R}^p}{\sup}\left|\mathbb{F}_n\left\{\frac{1}{\sqrt{n}}\sum_{i=1}^n(\bm\xi_i-\bm\mu_n)\leq\bm t\right\}-\mathbb{P}_{\bm G\sim\mathcal{N}_p(\bm 0,\bm\Sigma)}\{\bm G\leq \bm t\}\right|=0,
\end{equation}
for some covariance matrix $\bm\Sigma$. We keep here the subscript $n$ in $\bm D_n$ and $\bm\mu_n$ to emphasize they can change with $n$; we drop the subscript in later sections for simplicity.
The notation $o_{\mathbb{F}_n}(\cdot)$ means that the sequence converges to zero in probability at the respective rate uniformly over $\mathbb{F}_n\in\mathcal{F}_n$ treating $M$ as fixed. 
The \eqref{assumption:asymptotically:linear} assumption implies $\bb D_n$ satisfies the following uniform CLT pre-selection:
\begin{equation*} \label{assumption:clt} \tag{CLT}
	\underset{n\rightarrow\infty}{\lim}\:\underset{\mathbb{F}_n\in\mathcal{F}_n}{\sup}\:\underset{\bb t\in\mathbb{R}^p}{\sup} \left|\mathbb{F}_n\left\{\sqrt{n}\left(\bb D_n-\bm\mu(\mathbb{F}_n, M)\right)\leq \bb t\right\} -\mathbb{P}_{\bb G\sim\mathcal{N}_p(\bm 0,\bm\Sigma)}\{\bm G\leq \bb t\}\right|=0.
\end{equation*}
We prove this implication in the supplement (Lemma \ref{lemma:al:implies:clt}).
This is a CLT statement uniform across $\bb t\in\mathbb{R}^p$ and uniform across the class of distributions $\mathbb{F}_n\in\mathcal{F}_n$.


\begin{remark1}
We assume the stronger condition above holds, i.e.~\eqref{assumption:asymptotically:linear} instead of just \eqref{assumption:clt}. This allows us to use the method of \cite{chatterjee2005simple}, providing the rate of convergence as well as the asymptotic convergence. 

Assuming the sequence of parameters converges (a slightly more strict version of local alternatives), we prove a selective CLT using only the condition \eqref{assumption:clt}, convergence of the sequence of selection regions and the continuity assumption on the randomization distribution (Section \ref{app:selective:clt:different} in the supplement). This setup was used in \cite{tibshirani2015uniform} for proving the truncated Gaussian, which is non-randomized pivot, is asymptotically pivotal.
 However, this version of selective CLT does not allow for rare events and does not provide the rate of convergence. Thus we focus here on showing a selective CLT version allowing for underlying parameter $\bm\mu_n$ to be far from the selection region. 
\end{remark1}

\subsection{Linear decomposition} \label{sec:linear:decomposition}

Since we are interested in testing several linear combinations of $\bb \mu$ at once, let our vector-valued parameter of interest be $\bb A^T\bm\mu\in\mathbb{R}^a$ for some matrix $\bb A\in\mathbb{R}^{p\times a}$. Denote with $P_{\bm A}$ the projection matrix onto the column space of $\bm A$ and let $P_{\bm A}^{\perp}=\bm I_p-\bm P_{\bm A}$.
Denoting also $\bb\Sigma_{\bb A}=\bb A^T\bb\Sigma\bb A$, $\bb C=\bb\Sigma \bb A\bb\Sigma_{\bb A}^{-1}$, $\bm D_{\bb A}=\left(\bb I_p-\bb C\bb A^T\right)\bb D$,
we have the decomposition $\bm D=\bm D_{\bb A}+\bm C\bb A^T\bm D.$
Note that in the case $\bb D$ is normally distributed $\mathcal{N}_p\left(\bm\mu,\bm\Sigma/n\right)$ with known $\bm\Sigma$ and unknown $\bm\mu$ and if $\bb A^T\bb\mu$ is the parameter of interest the sufficient statistic for the nuisance parameters $P_{\bm A}^{\perp} \bm \mu$ would be $\bm D_{\bb A}$. 

We assume the model selection event is based on $\bb D$ and satisfies the affine constrains. Hence the selection region at step $n$ is denoted as
\begin{equation*}
\begin{aligned}
	\textnormal{selection}(\bm D,\bb\omega)
	&:=\left\{\sqrt{n}\bm A_M\bm D+\bb\omega\in\bm H_M \right\}
\end{aligned}
\end{equation*}
where $\bm A_M\in\mathbb{R}^{d\times p}$, $\bb\omega\in\mathbb{R}^d$ is a random variable representing the added randomization and $\bm H_M\subset\mathbb{R}^d$. Assume $\bm\omega\sim\mathbb{G}$, where $\mathbb{G}$ is a known distribution in $\mathbb{R}^d$ specified by the user. Further assume the randomization $\bm\omega$ is independent of the data $\bm S$. Also, $d$ is assumed to be fixed and does not depend on $n$. Denoting $\bm\Delta=\bm\Delta(\mathbb{F}_n)=\sqrt{n}\bm\mu$ and $\bb Z=\sqrt{n}\left(\bb D-\bb\mu\right)$, we write the selection region also as
\begin{equation*}
\begin{aligned}
	\textnormal{selection}(\bm D,\bb\omega)
	&=\left\{\bm A_M(\bm Z+\bm{\Delta})+\bb\omega\in\bm H_M\right\}\\
	&=\left\{\bm A_M\bm Z_{\bb A}+\bm A_M \bm C\bm A^T\bm Z+\bm A_M\bm\Delta+\bb\omega\in\bm H_M\right\},
\end{aligned}
\end{equation*}
where $\bm Z_{\bb A}=\left(\bb I_p-\bb C\bb A^T\right)\bm Z$.

To introduce the pivot from \cite{tian2015selective}, let us first define a survival function
\begin{equation*}
\begin{aligned}
	&\mathcal{P}^G\left(t,\bm G_{\bb A},\bb A^T\bm\mu\right)
	:=\Phi_n\times\mathbb{G}\left\{\|\bm\Sigma_{\bm A}^{-1/2}\bb A^T(\sqrt{n}(\bm G-\bm\mu))\|_2\geq t\:\Big|\:\textnormal{selection}(\bm G,\bb\omega), \bm G_{\bb A}\right\} \\
	&=\frac{\int\limits_{\|\bm\Sigma_{\bm A}^{-1/2}\bm A^T\bm\tau\|_2\geq t}\mathbb{G}\{\bm\omega:\sqrt{n}\bm A_M\bm G_{\bm A}+\bm A_M\bm C\bm\tau+\sqrt{n}\bm A_M\bm C\bm A^T\bm\mu+\bm\omega\in\bm H_M\}\phi_{(\bm 0,\bm\Sigma_{\bm A})}(\bm\tau)d\bm\tau}{\int\limits_{\mathbb{R}^a}\mathbb{G}\{\bm\omega:\sqrt{n}\bm A_M\bm G_{\bm A}+\bm A_M\bm C\bm\tau+\sqrt{n}\bm A_M\bm C\bm A^T\bm\mu+\bm\omega\in\bm H_M\}\phi_{(\bm 0;\bm\Sigma_{\bm A})}(\bm\tau)d\bm\tau},
\end{aligned}
\end{equation*}
where the probability on the RHS is under $\bm G\sim\mathcal{N}_p\left(\bb\mu,\bb\Sigma/n\right):=\Phi_n$ and under $\bb\omega\sim\mathbb{G}$. The conditioning in the definition of the pivot is on the selection event for $(\bm G, \bb\omega)$ and the value of $\bm G_{\bb A}=\left(\bb I_p-\bb C\bb A^T\right)\bm{G}$, which is the sufficient statistic for the nuisance parameters assuming our parameter of interest is $\bb A^T\bm{\mu}$. Note that computing this pivot requires knowing $\bb \Sigma$. We will assume the variance $\bb\Sigma$ is known. However, having uniformly consistent estimate of variance would suffice (see \cite{tian2015selective}).

Turning to the distribution of the constructed pivot, let us denote the selective distribution of the data $\bm S$ as $\mathbb{F}_n^*$, which is the marginal distribution of the data $\bm S$ where the joint distribution $(\bm S,\bm \omega)\sim\mathbb{F}_n^n\times\mathbb{G}$ is conditional on $\textnormal{selection}(\bm D,\bb\omega)$. Denote the class of the selective distributions of the data given the selection event corresponding to model $M$ as $\mathcal{F}_n^*$ (defined also in \eqref{eq:intro:selective:model}). All the guarantees given will be under $\mathbb{F}_n^*$, uniformly over the class $\mathcal{F}_n^*$.
Further, denote the conditional distribution of $\bm G$ as $\Phi_n^*$, which is the marginal distribution of $\bm G$ conditional on the event $\textnormal{selection}(\bm G,\bb\omega)$. It is not hard to see that $\mathcal{P}^G\left(\sqrt{n}\bb A^T(\bm G-\bm\mu),\bm G_{\bb A},\bb A^T\bm\mu\right)$ is exactly distributed as $\textnormal{Unif}[0,1]$ under this conditional distribution of $\bm G$ (Lemma 7 in \cite{tian2015selective}). Hence, under normality assumptions, the constructed pivot is exact and in what follows we state the non-parametric results.

\begin{remark1}
This construction of the pivot in exponential families is a standard construction in selective inference \citep{fithian2014optimal}. Under a CLT, the limiting family pre-selection is a Gaussian exponential family, hence its selective counterpart is a Gaussian exponential family subject to selection. This family is used to construct the pivotal quantity. Indeed, the main feature of selective inference is that under many interesting scenarios the limiting model is {\em not} Gaussian but a Gaussian model subject to selection. This observation suggests natural test statistics derived under the limiting Gaussian selective model. Much work is often known about the Gaussian model pre-selection, and this generally transfers to the selective model. 
\end{remark1}

\begin{remark1}
We elaborate what the assumptions above become in the LASSO example in Section \ref{sec:lasso:additional} in the supplement. Most notably, we show that the selection event of the LASSO is asymptotically affine and the data vector $\bm D$ is asymptotically linear pre-selection.
\end{remark1}

\subsection{Selective CLT under Lipschitz randomization}

We show the constructed pivot $\mathcal{P}^G$ is asymptotically $\textnormal{Unif}[0,1]$. An important result leading to that goal is that the non-parametric and Gaussian selective likelihood ratios are close.
The non-parametric selective likelihood ratio is defined as
\begin{equation*}
\ell_{\mathbb{F}_n}(\bm D)=\frac{d\mathbb{F}_n^*}{d\mathbb{F}_n}(\bm D):=\frac{\mathbb{G}\{\textnormal{selection}(\bm D,\bb\omega)\}}{(\mathbb{F}_n\times\mathbb{G})\{\textnormal{selection}(\bm D',\bb\omega)\}}
\end{equation*}
and similarly the Gaussian one $\ell_{\Phi_n}(\bm D)$.
In order to show the results, we need the following assumptions.


\begin{itemize}[leftmargin=*]
\item Uniformly bounded MGF of $\mathbb{F}_n$ in some neighborhood of zero: Precisely, assume there exists $\tau>0$ and a constant $C_{\tau}$ such that 
\begin{equation} \label{assumption:mgf} \tag{MGF}
\underset{n\geq 1}{\sup} \underset{\mathbb{F}_n\in\mathcal{F}_n}{\sup}\mathbb{E}_{\bm s\sim\mathbb{F}_n}\left[e^{\tau\|\xi(\bm s)-\bm\mu\|_1}\right]\leq C_{\tau}.
\end{equation}

\item The norm of matrix $\bb A_M$ does not grow with $n$, i.e.~there exists a finite constant $C_M$ such that 
	 \begin{equation} \label{assumption:norm:A_M} \tag{NA}
	   \underset{n\rightarrow\infty}{\limsup}\|\bm A_M\|_{h,2}\leq C_M<\infty,
	 \end{equation} 
	 where the matrix norm $\|\cdot\|_{h,2}$ is defined as $\|\bm A_M\|_{h,2}=\underset{\bm u\in\mathbb{R}^p}{\sup}\frac{\left\|\bm A_M\bm u\right\|_h}{\|\bm u\|_2}$.
	 
\item We assume the randomization $\bm\omega\in\mathbb{R}^d$ comes from a distribution $\mathbb{G}$ on $\mathbb{R}^d$ with the density $g$ satisfying $g(\bm\omega)=\exp(-\tilde{g}(\bm\omega))/C_g$ with $\tilde{g}$ having bounded derivatives up to order at least three (call the bound on all these derivatives $K_g$), where $C_g$ is normalization constant.
More precisely, for any multi-index $\bm\alpha\in\mathbb{N}_0^d$ with $\|\bm\alpha\|_1\leq 3$ we have
\begin{equation} \label{assumption:g:smoothness} \tag{S}
 \left|\partial_{\bm\omega}^{\bm\alpha}\tilde{g}(\bm\omega)\right|\leq K_g
\end{equation}
for some constant $K_g$. This implies the Lipschitz property of $\tilde{g}$,
\begin{equation} \label{assumption:g:lip} \tag{Lip}
	|\tilde{g}(\bm x)-\tilde{g}(\bm y)|\leq K_g \|\bm x-\bm y\|_h, \;\forall \bm x,\bm y \in\mathbb{R}^d, 
\end{equation}
where $\|\cdot\|_h$ is some norm in $\mathbb{R}^d$. Note that in particular that the logarithm of Gaussian density fails to satisfy the above assumptions as its  gradient is unbounded over $\mathbb{R}^d$. Hence, one of the main properties we require is the the logarithm of the density is globally Lipschitz. Bounded second and third derivatives are required in our proof though these assumptions could likely
be relaxed.
\end{itemize}
Under the assumptions above, including \eqref{assumption:asymptotically:linear}, we show the pivot $\mathcal{P}^G$ (constructed for estimating a linear functional $\bb\eta^T\bb\mu$) is asymptotically $\textnormal{Unif}[0,1]$ in the non-parametric setting under $\bm S\sim\mathbb{F}_n^*$. Before stating the result about the plugin Gaussian pivot, we state a result showing that Gaussian and non-parametric selective likelihood ratios, denoted as are asymptotically close.
The non-asymptotic (finite $n$ versions) of both results showing  $O(1/\sqrt{n})$ rate of convergence are in Lemma \ref{lemma:lip:randomization:lik:diff} and Theorem \ref{thm:selective:clt:finite:n:lipschitz} in the supplement with all the proofs (Section \ref{sec:proofs:selective:clt:lip:randomization}).

\begin{corollary}[Selective likelihood ratios] \label{corr:lip:randomization:lik:diff}
 Assuming \eqref{assumption:asymptotically:linear},  \eqref{assumption:mgf}, \eqref{assumption:norm:A_M} and \eqref{assumption:g:smoothness},
Gaussian and non-parametric selective likelihood ratios are asymptotically close: 
\begin{equation*}
	\underset{n\rightarrow\infty}{\lim}\:\underset{\mathbb{F}_n\in\mathcal{F}_n}{\sup}\:\mathbb{E}_{\mathbb{F}_n}\left[\left|\ell_{\mathbb{F}_n}(\bm D)-\ell_{\Phi_n}(\bm D)\right|\right]=0.
\end{equation*}
\end{corollary}

\begin{corollary}[Lipschitz randomization: Selective CLT] \label{corr:selective:clt:lipschitz} Assuming \eqref{assumption:asymptotically:linear},  \eqref{assumption:mgf}, \eqref{assumption:norm:A_M} and \eqref{assumption:g:smoothness}, we have
\begin{equation} \label{eq:selective:clt:lipschitz} \tag{P}
	\underset{n\rightarrow\infty}{\lim} \: \underset{\mathbb{F}_n^*\in\mathcal{F}_n^*}{\sup}\:\underset{t\in[0,1]}{\sup} \left|\mathbb{F}_n^*\left\{\mathcal{P}^G\left(\sqrt{n}(\bb A^T\bm D-\bb A^T\bm\mu),\bm D_{\bb A},\bb A^T\bm\mu\right)\right\} - t\right|=0.
\end{equation}
\end{corollary}

Note that the result above holds without any assumptions on $\bm \mu$, making the inferential procedures robust to rare selection events. The rare selection events are the ones for which $\bm\mu$ can be far from the selection region meaning that $d_h(\bm 0,\bm H_M-\sqrt{n}\bm A_M\bm\mu(\mathbb{F}_n))$, for some norm $\|\cdot\|_h$ and the corresponding distance $d_h(\cdot,\cdot)$ induced by this norm, can be arbitrarily large and might possibly change with $n$.
\cite{tian2015selective} show the result in \eqref{eq:selective:clt:lipschitz} for only 1-dimensional parameters, i.e $\bm A \in \mathbb{R}^{p \times 1}$ under more assumptions. They have our result of the selective likelihoods being close (Lemma \ref{lemma:lip:randomization:lik:diff} in the supplement) as an assumption. Further, they assume local alternatives, meaning 
\begin{equation} \label{assumption:local:alt} \tag{LA}
	d_h\left(\bm 0,\bm H_M-\sqrt{n}\bm A_M\bm\mu(\mathbb{F}_n)\right)\leq B,
\end{equation}
 where $B$ is a constant. This condition does not allow for the underlying parameter to be far from the selection region, hence the selection events cannot have low probability (cannot be rare).


\subsection{Selective CLT with Gaussian randomization} \label{sec:selective:clt:gaussian:randomization}

Gaussian randomization occurs naturally in some problems of interest, i.e.~data carving (Section \ref{sec:data:carving}), a more powerful version of data splitting. In this case, we use the data and the random split to construct the randomization that is asymptotically normal. The data carving problem becomes a randomized inference problem with Gaussian randomization. We present a result showing the validity of the corresponding pivot, i.e.~we state the selective CLT result for the Gaussian randomization. 
Assuming a slightly weaker condition than the local alternatives, that the parameter value is at the distance at most growing like $o(\log n)$ from the selection region, and sub-Gaussian data we obtain a selective CLT result for the Gaussian randomization that we now describe in more detail.

\begin{remark1} 
As our results for Gaussian randomization require some form of local alternatives, they are less robust to rare selection events than corresponding heavier tailed randomizations. While we have tried to remove such assumptions for Gaussian randomization, our current method of proof requires such assumptions. However, in the non-randomized setting (which can be thought of as randomization with a degenerate distribution) local alternatives are necessary. See \cite{tian2015selective} for a univariate example in which the pivotal quantity fails to converge weakly. Even though there seem to be differences between Gaussian and logistic randomization in the assumptions needed for selective CLT to hold, empirically, as we will see later in the implementation results, both work fine in practice.
\end{remark1}

We take the randomization distribution to be Gaussian in $d$ dimensions
\begin{equation} \label{assumption:gaussian:randomization} \tag{G}
	\mathbb{G}=\mathcal{N}_d(\bm 0, c\bm I_d),
\end{equation}
for some constant variance $c$. 
Assuming the matrix $\bm A$ is of size $d\times 1$ and the variance of $\bm A^T\bm Z=:T$ is 1, the pivot becomes
\begin{equation*}
	\frac{\int_{\bm H_M}\bar{\Phi}\left(\sqrt{\frac{\|\bm v\|_2^2+c}{c}}T-\frac{\bm v^T\bm w}{\sqrt{c(\|\bm v\|_2^2+c)}}\right)\cdot\exp\left(-\frac{1}{2}\bm w^T(\bm v\bm v^T+c\bm I)^{-1}\bm w\right)d\bm \omega'}{\int_{\bm H_M}\exp\left(-\frac{1}{2}\bm w^T(\bm v\bm v^T+c\bm I)^{-1}\bm w\right)d\bm\omega'},
\end{equation*}
where $\bm v=\bm A_M\bm C$ and $\bm w=\bm\omega'-\bm A_M\bm z_{\bm A}-\bm A_M\bm \Delta$ (derived in Section \ref{sec:smoothness:lemmas:gaussian:randomization} in the supplement).
In order to prove the above test statistic is uniform $[0,1]$ asymptotically, we make the following additional assumptions here.

\begin{itemize}[leftmargin=*]

\item The selection region is rectangular $\bm H_M=\prod_{i=1}^n[b_i,\infty)$ for some vector $\bm b\in\mathbb{R}^d$.

\item The distance of the parameter to the selection region is not growing too fast with the sample size:
	\begin{equation} \label{assumption:LA:weaker} \tag{LA-w}
		d(\bm 0, \bm H_M-\bm A_M\bm\Delta)=o(\log n).
	\end{equation}

\item The random variables coming from $\mathbb{F}_n$ are sub-Gaussian, i.e.~there exists $b>0$ such that 
	\begin{equation} \label{assumption:subG} \tag{subG}
	\underset{n\geq 1}{\sup} \underset{\mathbb{F}_n\in\mathcal{F}_n}{\sup}\mathbb{E}_{\bm s\sim\mathbb{F}_n}\left[e^{b\|\xi(\bm s)-\bm\mu\|_2^2}\right]\leq C_b.
	\end{equation}

\end{itemize}


We are ready to state the selective CLT for the Gaussian randomization, which is a corollary of the Theorem \ref{thm:selective:clt:gaussian:randomization} proved in Section \ref{sec:proofs:selective:clt:gaussian:randomization} in the supplement. The result relies on two important lemmas, Lemma \ref{lemma:lik:smoothness:gaussian:randomization} and Lemma \ref{lemma:pivot:smoothness:gaussian:randomization} in Section \ref{sec:proofs:selective:clt:gaussian:randomization} in the supplement, establishing the ``smoothness'' of the likelihood and the pivot.


\begin{corollary}[Gaussian randomization: Selective CLT] \label{thm:selective:clt:gaussian} Assuming  \eqref{assumption:asymptotically:linear},  \eqref{assumption:LA:weaker}, \eqref{assumption:gaussian:randomization}, \eqref{assumption:norm:A_M} and \eqref{assumption:subG}, we have
\begin{equation} \label{eq:selective:clt:gaussian}
	\underset{n\rightarrow\infty}{\lim} \: \underset{\mathbb{F}_n^*\in\mathcal{F}_n^*}{\sup}\:\underset{t\in[0,1]}{\sup} \left|\mathbb{F}_n^*\left\{\mathcal{P}^G\left(\sqrt{n}(\bb A^T\bm D-\bb A^T\bm\mu),\bm D_{\bb A},\bb A^T\bm\mu\right)\right\} - t\right|=0.
\end{equation}
\end{corollary}


\section{Consistency after selection} \label{sec:consistency}


The main result of this section is that assuming a sequence of  estimators is consistent for zero pre-selection (under $\mathbb{F}_n$ and treating $M$ as fixed), we also have that the sequence is consistent for zero post-selection (under $\mathbb{F}_n^*$). This result, important by itself, is also used later in showing the consistency of the bootstrap pivot.

Using the notation introduced so far, the selective likelihood ratio in our problem becomes
\begin{align*}
\ell_{\mathbb{F}_n}(\bm t)
	=\frac{\mathbb{G}\left\{\bm H_M-\bm A_M\bm z-\bm A_M\bm\Delta\right\}}{ \mathbb{E}_{\bm Z\sim\mathbb{F}_n}\left[\mathbb{G}\left\{\bm H_M-\bm A_M\bm\Delta-\bm A_M\bm Z\:\big|\:\bm Z\right\}\right] }, \;\; \bm z=\sqrt{n}(\bm t-\bm\mu).
\end{align*}
The following lemma is providing an integrable upper bounded on the selective likelihood ratio and is used to show the consistency result.
\begin{lemma}[Upper bound on the selective likelihood ratio] \label{lemma:upper:bound}
Using the \eqref{assumption:g:lip} assumption, we have
$\ell_{\mathbb{F}_n}(\bm t)\leq\frac{e^{K_g\|\bm A_M\bm z\|_h}}{\mathbb{E}_{\mathbb{F}_n}\left[e^{-K_g\|\bm A_M\bm Z\|_h}\right]}$.
\end{lemma}
The following result shows that if a sequence of estimators is consistent for zero under the original distribution of the data ($\bm S\sim\mathbb{F}_n$) then that sequence is also consistent for zero under the conditional distribution ($\bm S\sim\mathbb{F}_n^*$). 
The lemma is used later in proving the consistency of the bootstrap pivot under the conditional distribution. The proofs of the results in this section are in Section \ref{app:consistency} in the supplement.

\begin{lemma} [Consistency]  \label{lemma:consistency} 
Assume \eqref{assumption:clt}, \eqref{assumption:g:lip} and \eqref{assumption:norm:A_M} hold. Let $f_n(\bm D)$ be a sequence of estimators such that 
\begin{equation*}
	f_n(\bm D) \overset{\mathbb{F}_{n}}{\rightarrow} 0
\end{equation*}
as $n\rightarrow\infty$ uniformly across $\mathbb{F}_n\in\mathcal{F}_n$, i.e.~for every $\delta>0$, $\underset{n\rightarrow\infty}{\lim}\underset{\mathbb{F}_n\in\mathcal{F}_n}{\sup}\mathbb{F}_n\left\{f_n(\bm D)>\delta\right\}=0.$
Then we also have
\begin{equation*}
	f_n(\bm D) \overset{\mathbb{F}_n^*}{\rightarrow} 0
\end{equation*}
as $n\rightarrow\infty$ uniformly across $\mathbb{F}_n\in\mathcal{F}_n$, i.e.~for every $\delta>0$, $\underset{n\rightarrow\infty}{\lim}\underset{\mathbb{F}_n^*\in\mathcal{F}_n^*}{\sup}\mathbb{F}_{n}^*\left\{f_n(\bm D)>\delta\right\}=0$.
\end{lemma}

Note that this result is about the post-selection consistency of a sequence of estimators for a parameter (we can easily go from estimating zero to estimating a parameter in Lemma \ref{lemma:consistency}). In this context, the result of \cite{leeb2006can} is about impossibility of estimating the post-selection distribution of $f_n(\bm D)$. In the simple example (Section \ref{sec:simple:example}), the above lemma says that since $\bar{y}\rightarrow\mu$ as $n\rightarrow\infty$ pre-selection by SLLN, we have that $\bar{y}\rightarrow\mu$ as $n\rightarrow\infty$ post-selection as well. The result of \cite{leeb2006can} says that we cannot estimate the conditional CDF of $\bar{y}$ post-selection. As we discussed in Section \ref{sec:simple:example}, we can go around this issue and still construct valid confidence intervals for $\mu$.


\section{Bootstrap after selection} \label{sec:bootstrap}

We introduce the bootstrap pivot, $\mathcal{P}^B$, for testing $\bm{\bb A}^T\bm\mu$ after selection and prove it is asymptotically $\textnormal{Unif}[0,1]$, thus valid post-selection and can be used for inference. In order to prove the consistency of the bootstrap pivot, we assume the asymptotic uniformity of the plugin Gaussian pivot, i.e.~the result of the selective CLT holds for plugin Gaussian pivot, and the local alternatives. Additionally, we assume the consistency of the bootstrap pre-selection (treating the model as fixed and not chosen based on data) and having a consistent estimate of the variance. 

Before introducing the bootstrap pivot, let us denote as $\widehat{\mathbb{F}}_n$ the bootstrapped distribution of $\bm S$ and let $\widehat{\mathbb{E}}_n$ denote the corresponding expectation with respect to $\widehat{\mathbb{F}}_n$. $\widehat{\mathbb{F}}_n$ can be empirical distribution of $\bb S$ but not necessarily since we can also use the wild bootstrap, weighted, etc. Denote a bootstrap sample as $\bm S^*=(\bm S_1^*,\ldots, \bm S_n^*)$ with $\bm S_i^*\overset{i.i.d}{\sim}\widehat{\mathbb{F}}_n$, $i=1,\ldots,n$. For short, denote $\bm D^*=\bm D(\bm S^*)$, the test statistic $\bm D$ computed on the bootstrapped data. Define the bootstrapped selection region as
\begin{equation*}
\begin{aligned}
	\textnormal{selection}^*&(\bm D^*,\bm D,\bb \omega) \\
	& := \left\{\sqrt{n}\bm A_M \widehat{\bm D}_{\bb A}+\bm A_M\widehat{\bm C}\left(\sqrt{n}(\bm A^T\bm D^*-\bm A^T\bm D)\right)+\sqrt{n}\bm A_M\widehat{\bm C}\bm A^T\bm\mu+\bb\omega\in\bm H_M \right\}\\
	&= \left\{\sqrt{n}\bm A_M \widehat{\bm D}_{\bb A}+\bm A_M\widehat{\bm C}\bm A^T\bm Z^*+\sqrt{n}\bm A_M\widehat{\bm C}\bm A^T\bm\mu+\bb\omega\in\bm H_M \right\},
\end{aligned}
\end{equation*}
where $\widehat{\bm D}_{\bb A}=(\bb I_p-\widehat{\bb C}\bb A^T)\bb D$ and $\widehat{\bm C}=\widehat{\bb\Sigma}\bb A(\bb A^T\widehat{\bb\Sigma}\bb A)^{-1}$ for some estimate $\widehat{\bm\Sigma}$ of $\bm\Sigma$. Given the bootstrapped selection event, the bootstrapped version of the pivot $\mathcal{P}^G$ is defined as
\begin{equation} \label{eq:bootstrapped:stat}
	\mathcal{P}^B=\mathcal{P}^B\left(\sqrt{n}\left\|\bb A^T\bm D-\bb A^T\bm\mu\right\|_2,\bm D,\bm A^T\bm\mu\right),
\end{equation}
where the function $\mathcal{P}^B(t,\bb D,\bb A^T\bb\mu)$, $t\in \mathbb{R}$, is defined as a survival function of $\left\|\sqrt{n}\bb A^T(\bb D^*-\bb D)\right\|_2$ after bootstrapped selection event and conditional on the data $\bb D$:
\begin{equation*}
\begin{aligned}
	&\mathcal{P}^B\left(t,\bm D,\bb A^T\bm\mu\right):=(\widehat{\mathbb{F}}_n\times\mathbb{G})\left\{ \left\|\bb A^T(\sqrt{n}(\bm D^*-\bm D))\right\|_2\geq t\:\big|\:\textnormal{selection}^*(\bm D^*,\bm D,\bb\omega), \bm D\right\} \\
	&=\frac{\left(\widehat{\mathbb{F}}_n\times\mathbb{G}\right)\left.\left\{\left\|\bm A^T\bm Z^*\right\|_2 \geq t, \sqrt{n}\bm A_M\widehat{\bm D}_{\bb A}+\bm A_M\widehat{\bm C}\bm A^T\bm Z^*+\sqrt{n}\bm A_M\widehat{\bm C}\bm A^T\bm\mu+\bm\omega\in\bm H_M\right|\bm D\right\} }{ \left(\widehat{\mathbb{F}}_n\times\mathbb{G}\right)\left\{\sqrt{n}\bm A_M\widehat{\bm D}_{\bb A}+\bm A_M\widehat{\bm C}\bb A^T\bm Z^*+\sqrt{n}\bm A_M\widehat{\bm C}\bm A^T\bm\mu+\bb\omega \in \bm H_M\:\Big|\:\bm D\right\}}.
\end{aligned}
\end{equation*}
Conditioning on $\bb A^T\bm D^*$, which is equivalent to conditioning on $\bb A^T\bm Z^*$, we get
\begin{equation*}
\begin{aligned}
	&\mathcal{P}^B\left(t,\bm D,\bb A^T\bm\mu\right)\\	
	&=\frac{\widehat{\mathbb{E}}_n\left[\mathbb{I}_{\left\{\left\|\bb A^T\bm Z^*\right\|_2\geq t\right\}}\mathbb{G}\left\{\bm H_M-\sqrt{n}\bm A_M\widehat{\bm D}_{\bb A}-\bm A_M\widehat{\bm C}\bb A^T\bm Z^*-\sqrt{n}\bm A_M\widehat{\bm C}\bb A^T\bm\mu\:\Big|\: \bm D,\bb A^T\bm D^*\right\}\right] }{ \widehat{\mathbb{E}}_n \left[\mathbb{G}\left\{\bm H_M-\sqrt{n}\bm A_M\widehat{\bm D}_{\bb A}-\bm A_M\widehat{\bm C}\bb A^T\bm Z^*-\sqrt{n}\bm A_M\widehat{\bm C}\bb A^T\bm\mu\:\Big|\:\bm D,\bb A^T\bm D^*\right\}\right]}.
\end{aligned}
\end{equation*}

In the simple example (Section \ref{sec:simple:example}) we have described a way to compute the bootstrap pivot using standard sampling with replacement bootstrap.
Assume we had a (fast and accurate) way of computing probabilities of selection for a given data vector, i.e. $\mathbb{G}\left\{\textnormal{selection}(\bb D,\bb\omega)|\bb D\right\}$, as a function of $\bb D$. Then we could compute the bootstrap pivot in a usual way using the standard sampling with replacement bootstrap, where each bootstrap sample is weighted with these selection probabilities. To be precise about this computation, let us denote the bootstrapped samples $\bb D^{*b}$, $b=1, \ldots, B$. Then the proposed bootstrapped statistic would be computed as the following ratio
\begin{equation*}
	\frac{\frac{1}{B}\sum_{b=1}^B p^{*}(\bb D,\bb D^{*b})\mathbb{I}_{\left\{\sqrt{n}\left\|\bm A^T(\bm D^{*b}-\bm D)\right\|_2\geq \sqrt{n}\left\|\bb A^T(\bm D-\bm\mu)\right\|_2 \right\}} }{\frac{1}{B}\sum_{b=1}^B p^{*}(\bb D,\bb D^{*b})},
\end{equation*}
where
\begin{equation*}
\begin{aligned}
	&p^{*}(\bb D,\bb D^*) = \mathbb{G}\left\{\textnormal{selection}^*(\bb D^*, \bb D,\bb\omega)\big|\bb D^*,\bb D\right\} \\
	&=\mathbb{G}\left\{\bm{H}_M-\sqrt{n}\bm{A}_{M}\widehat{\bm D}_{\bb A}-\bm A_M\widehat{\bm C}\left(\sqrt{n}\bb A^T(\bm D^{*}-\bm D)\right)-\sqrt{n}\bm A_M\widehat{\bm C}\bb A^T\bm\mu\:\Big|\:\bm D,\bm D^{*}\right\}.
\end{aligned}
\end{equation*}
In the simple example in Section \ref{sec:simple:example}, we had an explicit form for $p^*$. However, in more complicated examples it is computationally expensive to get good approximations of $p^*$, thus we turn to the wild bootstrap and MC techniques to compute this pivot. 



In addition to the assumptions on the randomization and the local alternatives assumption, we need the asymptotic result of selective CLT to hold (the result \eqref{eq:selective:clt:lipschitz}), the consistency of estimated variance and the consistency of bootstrap pre-selection as stated.

\begin{itemize}[leftmargin=*]
\item Uniform consistency of the variance estimate: for any $\delta>0$, the variance estimate $\widehat{\bm\Sigma}$ of $\bm\Sigma$ satisfies
	\begin{equation} \label{assumption:variance} \tag{Var}
	   	\underset{n\rightarrow\infty}{\lim}\underset{\mathbb{F}_n\in\mathcal{F}_n}{\sup}\mathbb{F}_n\left\{\|\widehat{\bm\Sigma}-\bm\Sigma\|_2>\delta \right\}=0.
	\end{equation}
	This implies that for any $\delta>0$
	\begin{equation*}
		\underset{n\rightarrow\infty}{\lim}\underset{\mathbb{F}_n\in\mathcal{F}_n}{\sup}\mathbb{F}_n\left\{\|\widehat{\bb\Sigma}_{\bb A}-\bb \Sigma_{\bb A}\|_2>\delta \right\}=0 \;\;\textnormal{ and }\;\;
		\underset{n\rightarrow\infty}{\lim}\underset{\mathbb{F}_n\in\mathcal{F}_n}{\sup}\mathbb{F}_n\left\{\|\widehat{\bm C}-\bm C\|_2>\delta \right\}=0,
	\end{equation*}
	where $\widehat{\bb \Sigma}_{\bb A}$ and $\widehat{\bm C}$ are the corresponding quantities computed based on $\widehat{\bm\Sigma}$.
	

\item Uniform consistency of bootstrap before selection: for any $\delta>0$, we have
	\begin{equation} \label{assumption:boot:consistency} \tag{B}
		\underset{n\rightarrow\infty}{\lim}\:\underset{\mathbb{F}_n\in\mathcal{F}_n}{\sup}\mathbb{F}_n\left\{ \underset{\bb t\in\mathbb{R}^a}{\sup}\left|\widehat{\mathbb{F}}_n\left\{\sqrt{n}\bb A^T(\bm D^*-\bm D)\leq\bb t \right\} -\mathbb{P}_{\bb G}\left\{\widehat{\bb \Sigma}_{\bb A}^{1/2}\bb G\leq \bb t \:\big|\:\bm D\right\}\right|>\delta \right\}=0,
	\end{equation}
	where $\bb G\sim\mathcal{N}_a\left(\bb 0,\bb I_a\right)$ given $\bm D$, i.e.
	\begin{equation*}
		\underset{\bm t\in\mathbb{R}^a}{\sup}\left|\widehat{\mathbb{F}}_n\left\{\sqrt{n}\bb A^T(\bm D^*-\bm D)\leq \bb t \right\} -\mathbb{P}_{\bb G}\left\{\widehat{\bb\Sigma}_{\bb A}^{1/2} \bb G\leq \bb t |\bm D\right\}\right|=o_{\mathbb{F}_n}(1)
	\end{equation*}
	uniformly over $\mathbb{F}_n\in\mathcal{F}_n$. In the LASSO example, this assumption means that treating $E$ as fixed we have that the bootstrap for $\bar{\bm\beta}_E-\bm\beta_E$ is consistent. This holds under moment conditions on the data $(\bm X_E,\bm y_E)$ \citep{freedman1981bootstrapping}.
\end{itemize}


The following theorem proves the bootstrap statistic $\mathcal{P}^B$ is asymptotically distributed as $\textnormal{Unif}[0,1]$ under the conditional distribution $\mathbb{F}_n^*$ of the data, uniformly across the class of distributions $\mathcal{F}_n^*$. The proof of this result is in Section \ref{app:proofs:boot} in the supplement.

\begin{theorem}[Bootstrap after selection] \label{thm:boot}
Under the assumptions \eqref{assumption:boot:consistency}, \eqref{assumption:local:alt}, \eqref{assumption:g:lip}, \eqref{assumption:norm:A_M}, \eqref{eq:selective:clt:lipschitz} and \eqref{assumption:variance}  the following holds
\begin{equation*}
	\underset{n\rightarrow\infty}{\lim}\:\underset{\mathbb{F}_n^*\in\mathcal{F}^*}{\sup}\:\underset{t\in[0,1]}{\sup}\left|\mathbb{F}_n^*\left\{\mathcal{P}^B\left(\sqrt{n}\bb A^T(\bm D-\bm\mu), \bm D,\bb A^T\bm\mu\right)\leq t\right\} -t\right|=0.
\end{equation*}	
\end{theorem}


\section{Multiple views of the data} \label{sec:mv:general}

We present a way to do inference after running several randomized model selection procedures and choosing a parameter of interest upon looking at the outcomes of all of them.
In this setting, an analyst runs several model selection procedures on the same data and choses a parameter of interest upon seeing the outcomes of all of the model selection procedures. We note that this target of interest need not agree exactly with the results of any model selection procedure -- the data analyst can use their own expertise to choose a final parameter of interest but is allowed access to the results of the model selection procedure before choosing their parameter of interest.
We present a general sampling framework, which we then illustrate with two examples: randomized forward-stepwise and multiple runs of randomized $\ell_1$-penalized logistic regression.

Given data $\bm S\sim\mathbb{F}_n$ at step $n$, assume we solve $K$ randomized model selection procedures (views) of the following form
\begin{equation} \label{eq:mv:general:objective:k}
	\underset{\bm\beta}{\textnormal{minimize}}\; F_k(\bm\beta, \bm S)-\bm\omega^T_k\bm\beta, \;\; (\bm S,\bm\omega_k)\sim \mathbb{F}_n\times\mathbb{G}_k,
\end{equation}
where $F_k$ is some function, e.g.~$F_k(\bm \beta, \bm S)=\ell_k(\bm\beta;\bm S)+\mathcal{P}_k(\bm\beta)+\frac{\epsilon_k}{2}\|\bm\beta\|_2^2$, for some loss function $\ell_k$, penalty $\mathcal{P}_k$ and $\epsilon_k>0$ is added to make sure the solution exists. $\bm\omega_k$ is the added randomization variable distributed from $\mathbb{G}_k$ with density $g_k$. 
Each view defines a set of KKT conditions from solving \eqref{eq:mv:general:objective:k}, giving the randomization reconstruction equation $\bm\omega_k = \partial_{\bm\beta}F_k(\bm\beta, \bm S)$. Based on the solution of the objective in \eqref{eq:mv:general:objective:k} we choose a model $M_k$ for which we want to report inference.
Then, we write the randomization map write in terms of the optimization variables, denoted as $\bm V_k\in\mathbb{R}^{n_k}$ and a variable $\bm D_k=D_k(\bm S)\in\mathbb{R}^{p_k}$ that is only a function of the data $\bm S$. Precisely, we write the KKT conditions for $k$-th view as 
 \begin{equation*}
 	\bm\omega_k = \omega_k(\bm D_k,\bm V_k)
 \end{equation*}
 with the constraint $\bm V_k\in\mathcal{V}_k\subset\mathbb{R}^{n_k}$.
 
In order to do inference post-selection, we need the marginal distribution of the data $\bm S$ conditional on the event that running all the $K$ procedures for independent $\bm S\sim\mathbb{F}_n$ and $(\bm\omega_1,\ldots,\bm\omega_K)\sim \prod_{k=1}^K\mathbb{G}_k$ gives the same sequence of models $(M_1,\ldots, M_K)$. To get that distribution, we sample $\bm S$ and the optimization variables $(\bm V_1,\ldots,\bm V_K)$ from the density
 \begin{equation} \label{eq:mv:density:whole:data}
 	f_n(\bm S)\cdot\prod_{k=1}^K \left(g_k\left(\omega_k(D_k(\bm S), \bm V_k)\right)\cdot  J_k(\bm V_k)\right)
 \end{equation}
 restricted to $\bm V_k\in\mathcal{V}_k$, $k=1,\ldots K$, where $f_n$ is the density of $\mathbb{F}_n$ and $J_k$ is the Jacobian coming from the change of density at $k$-th view. Depending on the inferential goal, we do not necessarily sample the whole data $\bm S$ but a function of it.

Given the models selected at each of the $K$ views, an analyst decides to test a parameter $H_0: \theta\left(\mathbb{F}_n, \{M_k\}_{k=1}^K\right)=\bm\theta\in\mathbb{R}^a$ using a target test statistic $\bm T=T(\bm S)$. We now describe how to simplify the density in \eqref{eq:mv:density:whole:data} so that we only sample $\bm T$ instead the whole $\bm S$. Assume that there is a CLT holding pre-selection
\begin{equation*}
	\begin{pmatrix}
		\bm T \\ \bm D_k \end{pmatrix} \rightarrow\mathcal{N}_{a+p_k}\left(\begin{pmatrix}\bm\theta \\ \bm\mu_k \end{pmatrix},\begin{pmatrix}
			\bm\Sigma_{\bm T} & \bm\Sigma_{\bm T,\bm D_k} \\ \bm\Sigma_{\bm D_k,\bm T} & \bm\Sigma_{\bm D_k}
		\end{pmatrix} \right)
\end{equation*}
as $n\rightarrow\infty$ for all $k=1,\ldots, K$. Denoting the respective covariance estimates as $\widehat{\bm\Sigma}_{\bm D_k,\bm T}$ and $\widehat{\bm\Sigma}_{\bm T}$, we treat $\bm F_k=\bm D_k - \widehat{\bm\Sigma}_{\bm D_k,\bm T}\widehat{\bm\Sigma}_{\bm T}^{-1}\bm T$ 
as fixed in the sampler. Hence, we write the randomization reconstruction at step $k$ only as a function of $\bm T$ and $(\bm V_1,\ldots,\bm V_K)$ and the plugin CLT sampling density of these variables becomes proportional to
\begin{equation*}
	\phi_{(\bm\theta,\widehat{\bm\Sigma}_{\bm T})}(\bm T)\cdot\prod_{i=1}^K \left(g_k(\omega_k(\bm F_k+\widehat{\bm\Sigma}_{\bm D_k, \bm T}\widehat{\bm\Sigma}_{\bm T}^{-1}\bm T,\bm V_k))\cdot J_k(\bm V_k)\right)
\end{equation*}
with the constrains $(\bm V_1, \ldots,\bm V_K)\in\prod_{k=1}^K\mathcal{V}_k$.


The computational approaches for the multiple views of the data represent generalizations of the sampling methods presented for the LASSO (Section \ref{sec:lasso}).
\begin{itemize}[leftmargin=*]
\item (Weighted optimization sampler) We sample the optimization variables $(\bm V_1,\ldots, \bm V_K)$ from the selective density above fixing the data at its observed value and proceed similarly as in the simple example and the LASSO example.

\item (Wild bootstrap sampler) Using the wild bootstrap, $\bm T$ gets replaced by $T(\bm\alpha)+\bm\theta$, where $T(\bm\alpha)$ is consistent for $\mathcal{N}(\bm 0,\bm\Sigma_{\bm T})$. The sampling density of the bootstrapped weights $\bm\alpha\in\mathbb{R}^n$ and the optimization variables $\bm V_1,\ldots,\bm V_K$ becomes 
\begin{equation*}
		\left(\prod_{i=1}^n h_{\alpha}(\alpha_i)\right)\times\prod_{k=1}^K\left(g_k\left(\omega_k(\bm F_k+\widehat{\bm\Sigma}_{\bm D_k, \bm T}\widehat{\bm\Sigma}_{\bm T}^{-1}(\bm T(\bm\alpha)+\bm\theta),\bm V_k)\right)\cdot J_k(\bm V_k)\right).
\end{equation*}
with the constraints on the bootstrap weights $\bm\alpha\in\textnormal{supp}(\mathbb{H}_{\alpha})^n$ and the optimization variables as before $(\bm V_1,\ldots, \bm V_K)\in\prod_{k=1}^K\mathcal{V}_k$.

\end{itemize}


\subsection{Forward stepwise} \label{sec:forward:stepwise}

We revise the randomized forward stepwise example from \cite{selective_sampler}, which is a special case of a Kac-Rice test \citep{kac_rice}, and describe how the sampling here works once we specify the target of inference.


The data generating mechanism on $(\bm X,\bm y)$ is as in the above examples. In the $K$ steps of forward stepwise, the selection event is characterized by a sequence of indices $\bb j=(j_1,\ldots, j_K)$ with their corresponding signs $\bb s=(s_1,\ldots, s_K)$ that enter the model in that particular order, forming an active set at step $K$. Denote the active set at step $k$ as $E_k=\{j_1,j_2,\ldots, j_k\}$ for all $k=1,\ldots, K$.
At the $k$-th step the randomized forward stepwise solves the following program
\begin{equation} \label{eq:fs:objective:original}
	\underset{\bb \eta\in \mathcal{B}_k}{\textnormal{maximize}}\:\bb\eta^T \left(\bb X_{-E_{k-1}}^TP_{E_{k-1}}^{\perp} \bb y+\bb \omega_k\right), \;\; (\bb X,\bb y)\times\bb\omega_k\sim\mathbb{F}_n^n\times\mathbb{G}_k,
\end{equation}
where $\mathcal{B}_k=\{\bb \eta \in\mathbb{R}^{p-k+1}: \|\bb \eta\|_{1}\leq 1\},$ and $P_{E_{k-1}}^{\perp}\bm y$ is the residual left after projecting $\bm y$ onto $\bm X_{E_{k-1}}$.

Denoting the solution of \eqref{eq:fs:objective:original} at the $k$-th step as $\hat{\bb\eta}_k$, the selection event of interest is given by conditioning on the sign and the index on the nonzero coordinate of $\hat{\bb\eta}_k$ for each $k=1,\ldots, K$. We want to sample from the density of the data and the randomization conditional on this selection event. The randomization reconstruction map for the $k$-th step, from the sub-gradient equation is given by
\begin{equation*}
	\omega_k(\bb y,\bb z_k)=-\bb X_{-E_{k-1}}^TP_{E_{k-1}}^{\perp} \bb y+\bb z_k,
\end{equation*}
where, subdifferential $\bb z_k\in \mathbb{R}^{p-k+1}$ from the $k$-th step is restricted to the normal cone
$\bb z_k\in \partial I_{\mathcal{B}_k}(\hat{\bb \eta}_k)$ (see \cite{selective_sampler}).
In order to sample $((\bb X,\bb y),\bb\omega_1,\ldots, \bb\omega_K)$ from the selective density, we sample $((\bb X,\bb y), \bb z_1,\ldots, \bb z_K)$ from the density proportional to
\begin{equation} \label{eq:fs:density}
	\left(\prod_{i=1}^n f_n (\bb x_i,y_i)\right) \cdot \prod_{k=1}^K g_k\left(\bb z_k-\bb X_{-E_{k-1}}^TP_{E_{k-1}}^{\perp} \bb y\right),
\end{equation}
supported on $\mathbb{R}^{n\times p}\times\mathbb{R}^n \times \prod_{k=1}^K \partial I_{\mathcal{B}_k}(\hat{\bb\eta}_k)$.

After doing $K$ steps of forward stepwise, an analyst looks at the sequence $\{E_k\}_{k=1}^K$ and chooses model $E$ in whichever way she wants. The goal is to inference for the population OLS parameters $\bb\beta_E^*$. 
As in the LASSO example, we simplify the sampling above since we are interested in testing a particular parameter. First note that
$\bb X_{-E_{k-1}}^T P_{E_{k-1}}^{\perp}\bb y$ can be expressed as $\bb Q_k\bb X^T\bb y$, where
$\bb Q_k= \left[\begin{matrix}
\bb X_{-E_{k-1}}^T\bb X_{E_{k-1}}(\bb X_{E_{k-1}}^T \bb X_{E_{k-1}})^{-1} & -\bb I_{p-(k-1)}	\end{matrix} \right].$
Using the asymptotic normality of $\bb D=\begin{pmatrix} \bar{\bb \beta}_E \\ \bb X_{-E}^T(\bb y-\bb X_E\bar{\bb\beta}_E) \end{pmatrix}$, the sampling density of $(\bb D, \bb z_1,\ldots, \bb z_K)$ is proportional to
\begin{equation*}
	\phi_{(\bm\mu_{\bm D},\bm\Sigma_{\bm D})}(\bm D)\cdot \prod_{k=1}^K g_k \left(\bb z_k +\bm M_k \bm D
		\right),
\end{equation*}
and supported on $\mathbb{R}^p\times\prod_{k=1}^K \partial I_{\mathcal{B}_k}(\hat{\bb\eta}_k)$, where $\bm M_k = -\bm Q_k\begin{pmatrix} \bm X_E^T\bm X_E &\bm 0\\ \bb X_{-E}^T\bm X_E& \bm I_{p-|E|}\end{pmatrix}$.

If we want to test $H_0:\bm A^T\bm\beta_E^*$, we use the decomposition as in the previous examples to get that the sampling density of $\bm T=\bm A^T\bar{\bm\beta}_E$ along with optimization variables is proportional to 
\begin{equation} \label{eq:fs:density:plugin:clt}
	\phi_{(\bm\theta,\bm\Sigma_{\bm T})}(\bm T) \cdot \prod_{k=1}^K g_k \left(\bb z_k+\bm M_k\bm F+\widetilde{\bm M}_k\bm T\right)
\end{equation}
and supported on $\mathbb{R}^K\times \prod_{k=1}^K \partial I_{\mathcal{B}_k}(\hat{\bb\eta}_k)$, where $\widetilde{\bm M}_k = \bb M_k \widehat{\bm\Sigma}_{\bm D,\bm T}\widehat{\bm\Sigma}_{\bm T}^{-1}.$ To do inference we can reuse any of the samplers above; the sampling details including Langevin Monte Carlo updates are in Section \ref{app:sampling} in the supplement.


\subsection{Multiple views for GLM loss and group LASSO penalty} \label{sec:mv:glm}

We do a total of $K$ model selection procedures (or views / queries on the data), each solving the following optimization problem
\begin{equation} \label{eq:mv:glm:objective:kth:view}
	\hat{\bm\beta}_k = \underset{\bm\beta\in\mathbb{R}^p}{\textnormal{arg min}}\: \ell_k(\bm\beta;(\bm X,\bm y))-\bm\omega^T_k\bm\beta +\mathcal{P}_k(\bm\beta) +\frac{\epsilon_k}{2}\|\bm\beta\|_2^2, \;\;\; (\bm X,\bm y)\times\bm\omega_k\sim \mathbb{F}^n_n\times\mathbb{G}_k,
\end{equation}
where $\ell_k$ is a loss function and $\mathcal{P}_k$ a penalty term for the $k$-th view, $k=1,\ldots, K$. For the logistic loss function the above loss becomes
\begin{equation*}
	\ell_k(\bm\beta;(\bm X,\bm y))=-\sum_{i=1}^n \left(y_i\log\pi(\bm x_i^T\bm\beta)+(1-y_i)\log(1-\pi(\bm x_i^T\bm\beta))\right),
\end{equation*}
where $\pi(x)=\frac{e^x}{1+e^x}$, and the group LASSO penalty above becomes $\mathcal{P}_k(\bm\beta) = \sum_{g\in G_k}\lambda_{k,g}\|\bm\beta_g\|_2,$
where $G_k$ defines the partition $\cup_{g\in G_k}g =\{1,\ldots, p\}$ and $\lambda_{k,g}$ are the group weights. We use $g$ to denote both a group and its elements. This example has been done in detail in \cite{selective_sampler} for Gaussian loss and one view of the data. We write the selective plugin CLT density before presenting the bootstrap in this case.

Let us introduce some more notation before writing KKT conditions. Denote with $A_k$ the set of active groups selected by solving the above procedure and with $E_k$ the set of active variables selected, i.e.~$E_k=\cup_{g\in A_k}g$. Hence, we write $\hat{\bm\beta}_k = \begin{pmatrix}
	\hat{\bm\beta}_{k,E_k} \\ \bm 0 \end{pmatrix}$, where 
	in the interchangeable notation $(\hat{\bm\beta}_{k,g})_{g\in A_k}=\hat{\bm\beta}_{k,E_k}$. Now let $\bar{\bm\beta}_{E_k}\in\mathbb{R}^{|E_k|}$ denote the MLE including only the variables in $E_k$ ($\bm X_{E_k}$ with rows $\bm x_{i,E_k}$, $i=1, \ldots, n$). Define the data vector $\bm D_k=\begin{pmatrix} \bar{\bm\beta}_{E_k} \\ \bm X_{-E_k}^T(\bm y-\pi_{E_k}(\bar{\bm\beta}_{E_k})) \end{pmatrix}$ and the following quantities
\begin{equation*}
\begin{aligned}
	\pi_{E_k}(\bm\beta) &= \frac{\exp(\bm X_{E_k}\bm\beta)}{1+\exp(\bm X_{E_k}\bm\beta)}, \;\;
	W_{E_k}(\bm\beta) = \textnormal{diag}\left(\pi_{E_k}(\bm\beta)(1-\pi_{E_k}(\bm\beta)\right) \\
	Q_{E_k}(\bm\beta) &= \bm X_{E_k}^T W_{E_k}(\bm\beta)\bm X_{E_k},\;\;
	C_{E_k}(\bm\beta) = \bm X_{-E_k}^T W_{E_k}(\bm\beta)\bm X_{E_k},\;\;
	I_{E_k}(\bm\beta) = C_{E_k}(\bm\beta)Q_{E_k}^{-1}(\bm\beta)
\end{aligned}
\end{equation*}
for any $\bm\beta\in\mathbb{R}^{|E_k|}$.

Conditioning on the active directions, denoted as
$\bm u_{k,g}=\frac{\hat{\bm\beta}_{k,g}}{\|\hat{\bm\beta}_{k,g}\|_2}$, $g\in A_k$, we have $\hat{\bm\beta}_{k,g} = \gamma_{k,g}\bm u_{k,g}$, where $\gamma_{k,g}=\|\hat{\bm\beta}_{k,g}\|_2$. The sub-gradient equation at the $k$-th view is an affine function in terms of the data vector $\bm D_k$ and the optimization variables $(\gamma_{k,g})_{g\in A_k}$, $\gamma_{k,g}\in\mathbb{R}$, and $(\bm z_{k,h})_{h\in -E_k}$, $\bm z_{k,h}\in\mathbb{R}^{|h|}$, restricted to
\begin{equation} \label{eq:constrains:inactive:and:norm}
	\gamma_{k,g}> 0,\;\; \forall g\in A_k,\;\; \textnormal{ and }\;\;  \|\bm z_{k,h}\|_2\leq \lambda_{k,h},\;\; \forall h\in -A_k.
\end{equation}
In a compact form, we write the randomization reconstruction as
\begin{equation*}
	\bm\omega_k=\bm M_k\bm D_k + \bm\Gamma_k(\gamma_{k,g})_{g\in A_k} + \bm Z_k(\bm z_{k,h})_{h\in -E_k}+\bm L_k,
\end{equation*}
where $\bm M_k\in\mathbb{R}^{p\times p}$, $B_k\in\mathbb{R}^{p\times|E_k|}$, $\bm Z_k\in\mathbb{R}^{p\times(p-|E_k|)}$, $\bm L_k\in\mathbb{R}^p$ denote
\begin{equation*}
\begin{aligned}
	\bm M_k &= -\begin{pmatrix}
	Q_{E_k}(\bar{\bm\beta}_{E_k}) & \bm 0 \\ C_{E_k}(\bar{\bm\beta}_{E_k}) & \bm I_{p-|E_k|}
	\end{pmatrix},\;\;\bm\Gamma_k = \begin{pmatrix} Q_{E_k}(\bar{\bm\beta}_{E_k})+\epsilon \bm I_{|E_k|} \\ C_{E_k}(\bar{\bm\beta}_{E_k}) \end{pmatrix} \textnormal{diag}((\bm u_{k,g})_{g\in A_k}), \\
	\bm Z_k &= \begin{pmatrix} \bm 0 \\ \bm I_{p-|E_k|} \end{pmatrix},\;\;\bm L_k = \begin{pmatrix}
		(\lambda_{k,g}\bm u_{k,g})_{g\in A_k} \\ \bm 0
	\end{pmatrix}.
\end{aligned}
\end{equation*}

After looking at the outcomes of all $K$ views, i.e.~the sets $E_1,\ldots, E_K$, an analyst decides on a model $E$. It could be that the analyst decides on $E=\cup_{k=1}^K E_k$ the union of all active variables across the views but not necessarily. If we are interested in testing $H_0:\bm A^T\bm\beta_E^*=\bm\theta$, we use the same target statistic and decomposition as in the previous examples. Focusing on the data part pre-selection, there is a CLT for all $k=1,\ldots, K$:
\begin{equation*}
	\begin{pmatrix} \bm T \\ \bm D_k \end{pmatrix} \overset{d}{\rightarrow} \mathcal{N}_{a+p}\left(\begin{pmatrix} \bm\theta \\ \bm\mu_k \end{pmatrix}, \begin{pmatrix} \bm\Sigma_{\bm T} & \bm\Sigma_{\bm T,\bm D_k} \\ \bm\Sigma_{\bm D_k,\bm T} & \bm\Sigma_{\bm D_k} \end{pmatrix}\right)
\end{equation*}
as $n\rightarrow\infty$. Decomposing $\bm D_k$ in terms of $\bm T$, we fix the quantities
$\bm F_k = \bm D_k - \widehat{\bm\Sigma}_{\bm D_k,\bm T}\widehat{\bm\Sigma}_{\bm T}^{-1}\bm T$ in the sampler.
The randomization reconstruction at the $k$-th view becomes
\begin{equation} \label{eq:randomization:reconstruction}
	\bm\omega_k=\widetilde{\bm M}_k \bm T+\bm\Gamma_k(\gamma_{k,g})_{g\in A_k}+\bm Z_k(\bm z_{k,h})_{h\in -A_k}+\widetilde{\bm L}_k,
\end{equation}
where $\widetilde{\bm M}_k = \bm M_k\widehat{\bm\Sigma}_{\bm D_k,\bm T}\widehat{\bm \Sigma}_{\bm T}^{-1}$, $\widetilde{\bm L}_k = \bm L_k+\bm M_k\bm F_k$ and the optimization variables $(\gamma_{k,g})_{g\in A_k}$ and $(\bm z_{k,h})_{h\in -A_k}$ are restricted as in \eqref{eq:constrains:inactive:and:norm}.
The plugin CLT selective density in terms of the target $\bm T$ and the optimization variables $(\gamma_{k,g})_{g\in A_k}$ and $(\bm z_{k,h})_{h\in -A_k}$, $k=1,\ldots, K$, is proportional to 
\begin{equation} \label{eq:mv:glm:target:density}
	\phi_{(\bm\theta,\widehat{\bm\Sigma}_{\bm T})}(\bm T)\cdot\prod_{k=1}^K \left(g_k(\widetilde{\bm M}_k\bm T+\bm\Gamma_k(\gamma_{k,g})_{g\in A_k}+\bm Z_k(\bm z_{k,h})_{h\in -A_k}+\tilde{\bm L}_k) \cdot J_k\right)
\end{equation}
and restricted to \eqref{eq:constrains:inactive:and:norm}, where 
$J_k$ is the Jacobian coming from the change of density at view $k$. The sampling details are in Section \ref{app:sampling} in the supplement.


\section{Data splitting revisited} \label{sec:data:carving}

Classical data splitting uses a random subsample of the data to choose the model and the leftover of the data to do inference for the selected coefficients. Hence, a part of the data (called the first stage data) is used only for model selection.
In this work, we also select a model based on a part of the data but we use the whole data for inference by conditioning only on the model selected in the first stage. While the classical data splitting conditions on all of the first stage data, we condition only on the model selected using the first stage data, conducting the inference using the leftover data together with the leftover information in the first stage data (after using it to select the model). This idea, called \textit{data carving}, of reusing the leftover information from the first stage has been introduced in \cite{fithian2014optimal}. We provide the computational tools to do data carving efficiently in non-parametric settings.

Let us now describe the setup and the procedure in detail. Given data matrix $\bm S\in\mathbb{R}^{n\times m}$, we resample $n_1$ of its rows randomly without replacement and denote this part of the data as $\bm S_1\in\mathbb{R}^{n_1\times (p+1)}$, where $n_1\leq n$ is given. Assume pre-selection $\bm S\sim\mathbb{F}^n_n$, i.e.~the rows are sampled i.i.d.~from $\mathbb{F}_n$.
Denote the full loss as $\ell(\bm\beta;\bm S)$ and the loss of the subsampled data as $\ell_1(\bm\beta;\bm S_1)$. The standard data splitting uses only the data $\bm S_1$ to select the model 
by solving an optimization problem of the form
\begin{equation} \label{eq:objective:split:data}
	\underset{\bm\beta}{\textnormal{minimize}}\:\frac{1}{\rho}\ell_1(\bm\beta;\bm S_1)+\mathcal{P}(\bm\beta),
\end{equation}
where $\rho=\frac{n_1}{n}$ and $\mathcal{P}(\bm\beta)$ is a  penalty. Denote the selected model as $M$. In data splitting, the inference is done for the coefficients corresponding to model $M$ using the leftover data  $\bm S_2\in\mathbb{R}^{(n-n_1)\times m}$, where $\bm S_2$ consists of the rows of $\bm S$ not used to form $\bm S_1$. 
Denoting $\delta(\bm\beta;\bm S) = \rho\ell(\bm\beta;\bm S)-\ell_1(\bm\beta;\bm S_1)$,
\eqref{eq:objective:split:data} is equivalent to solving
\begin{equation} \label{eq:data:splitting:objective}
	\underset{\bm\beta}{\textnormal{minimize}}\:\ell(\bm\beta;\bm S)-\frac{1}{\rho}\delta(\bm\beta;\bm S)+\mathcal{P}(\bm\beta).
\end{equation}
Interpreting the term $\frac{1}{\rho}\delta(\bm\beta;\bm S)$ as the randomization term, the problem above can be seen as an example of a randomized procedure of the type we have considered so far. In other words, we rewrite the model selection problem using the first stage data as a randomized model selection problem using the full data with the randomization being a function of both the first stage and the second stage data.

Denoting the minimizer of the above objective as $\hat{\bm\beta}$, the sub-gradient equation of \eqref{eq:data:splitting:objective} becomes
\begin{equation*}
	\bm\omega :=\frac{1}{\rho}\nabla\delta(\hat{\bm\beta};\bm S) = \nabla\ell(\hat{\bm\beta};\bm S)+\partial\mathcal{P}(\hat{\bm\beta}) =\omega(\bm D, \bm V),
\end{equation*}
where $\bm D=D(\bm S)$ is a data dependent vector and $\bm V\in\mathcal{V}$ are naturally arising optimization variables restricted to a set $\mathcal{V}$. Choosing the model $M$ is equivalent to having $\bm V\in\mathcal{V}$, hence to condition on the model it suffices that the optimization variables satisfy this constraint. Thus, to sample the data $\bm D$ post-selection, i.e.~conditional on \eqref{eq:data:splitting:objective} choosing $M$, it suffices to sample $\bm D$ and $\bm V$ from the density as follows.
Assuming that pre-selection $\bm\omega=\frac{1}{\rho}\nabla\delta(\hat{\bm\beta};\bm S)$ and $\bm D$ satisfy a CLT
\begin{equation*}
	\begin{pmatrix}\bm \omega \\ \bm D\end{pmatrix} \overset{d}{\rightarrow}\mathcal{N}_{\textnormal{dim}(\bm\omega)+\textnormal{dim}(\bm D)}\left(\bm\mu,\bm\Sigma\right)
\end{equation*}
as $n\rightarrow\infty$, we sample $(\bm D, \bm V)$ from the density proportional to the multivariate Gaussian density
\begin{equation} \label{eq:data:spliting:gaussian:denisty}
	\phi_{(\bm\mu,\bm\Sigma)}\begin{pmatrix} \omega(\bm D,\bm V) \\ \bm D \end{pmatrix}
\end{equation}
with the restriction $\bm V\in\mathcal{V}$, assuming the Jacobian is constant. In the examples later, we will choose the randomization $\bm\omega$ to be of mean zero and asymptotically independent of $\bm D$, in which case the density in \eqref{eq:data:spliting:gaussian:denisty} simplifies further as
\begin{equation} \label{eq:data:spliting:independent:randomization}
	\phi_{(\bm\mu_{\bm D},\bm\Sigma_{\bm D})}(\bm D)\cdot\phi_{(\bm 0,\bm\Sigma_{\bm\omega})}(\omega(\bm D,\bm V)),
\end{equation}
where $\bm D\overset{d}{\rightarrow}\mathcal{N}(\bm\mu_{\bm D},\bm\Sigma_{\bm D})$ and $\bm\omega\overset{d}{\rightarrow}\mathcal{N}(\bm 0,\bm\Sigma_{\bm\omega})$ as $n\rightarrow\infty$. 
Now we can reuse any of the samplers above to do inference.


\subsection{GLM example}

To illustrate the idea through an example, consider $\ell_1$-penalized logistic regression. The data $\bm S$ above consists of the data matrix $\bm X\in\mathbb{R}^{n\times p}$ and a response vector $\bm y\in\mathbb{R}^n$. The first stage data is $\bm S_1=(\bm X_1, \bm y_1)$.
The model selection objective on the first stage data becomes
\begin{equation} \label{data:splitting:example:objective}
	\underset{\bm\beta}{\textnormal{minimize}}\:\ell(\bm\beta;(\bm X,\bm y))-\frac{1}{\rho}\delta(\bm\beta;(\bm X,\bm y))+\frac{\epsilon}{2}\|\bm\beta\|_2^2+\lambda\|\bm\beta\|_1.
\end{equation}
Since $\ell(\bm\beta;(\bm X,\bm y))-\frac{1}{\rho}\delta(\bm\beta;(\bm X,\bm y))=\frac{1}{\rho}\ell_1(\bm\beta;(\bm X_1,\bm y_1))$, the solution of the above optimization problem depends only on the split data. Denote the solution of \eqref{data:splitting:example:objective} as $\hat{\bm\beta}=\begin{pmatrix}	\bm \beta_E \\ \bm 0 \end{pmatrix}$. 
The KKT conditions of this problem are
\begin{equation*}
	\bm\omega=\frac{1}{\rho}\nabla\delta\begin{pmatrix}\bm\beta_E \\ \bm 0\end{pmatrix} = \nabla\ell\begin{pmatrix}\bm\beta_E \\ \bm 0\end{pmatrix}+\epsilon\begin{pmatrix}\bm\beta_E\\ \bm 0\end{pmatrix}+\lambda\begin{pmatrix}
\bm s_E \\ \bm u_{-E}\end{pmatrix}
\end{equation*}
with the constrains $\textnormal{sign}(\bm\beta_E)=\bm s_E$ and $\|\bm u_{-E}\|_{\infty}\leq 1$.
We want to sample the data vector (a function of $(\bm X,\bm y)$) and the optimization variables $(\bm\beta_E,\bm u_{-E})$ such that they satisfy the KKT conditions corresponding to \eqref{data:splitting:example:objective}.

In order to get more explicit form for the randomization $\bm\omega$, we do the Taylor expansion of both gradient of $\delta$ and the gradient of the loss $\ell$, implying 
\begin{equation*} 
	\bm \omega \approx \frac{1}{\rho}\nabla\delta\begin{pmatrix}\bm\beta_E^*\\ \bm 0 \end{pmatrix} =-\bm X^T(\bm y-\pi(\bm X_E\bm \beta_E^*))+\frac{1}{\rho}\bm X_1^T(\bm y_1- \pi(\bm X_{1,E}\bm\beta_E^*)).
\end{equation*}
A simple calculation shows the asymptotic covariance between $\bm\omega$ and $\bm D$ is zero, hence we sample from the density in \eqref{eq:data:spliting:independent:randomization}, where all the covariances are estimated using bootstrap.


\subsection{Multiple splits} \label{sec:multiple:splits}

Combining the ideas from multiple spits and data splitting, we now do inference after looking at the outcomes of the model selection procedures run on different splits of the data. 

Let us in total do $K$ splits of the data, each running a model selection procedure on the split data $(\bm X_k,\bm y_k)\in\mathbb{R}^{n_k\times p}\times\mathbb{R}^{n_k}$, of size $n_k$, selecting an active set $E_k$ at each of $k=1,\ldots, K$. Let an analyst choose a model $E$ upon seeing the outcomes $E_1,\ldots, E_K$. Each of the model selection procedures can be written as 
\begin{equation} \label{eq:multiple:splits:objective}
	\ell(\bm \beta;(\bm X,\bm y))+\delta_k(\bm\beta;(\bm X, \bm y))+\mathcal{P}(\bm\beta),
\end{equation}
where $\delta_k(\bm\beta;(\bm X;\bm y))=\ell(\bm\beta;(\bm X,\bm y))-\frac{1}{\rho_k}\ell_k(\bm\beta;\bm X_k,\bm y_k))$, $\rho_k=\frac{n_k}{n}$. Let us write the KKT conditions of \eqref{eq:multiple:splits:objective} as
$\bm\omega_k = \omega_k(\bm D,\bm V_k)$,
where $\bm\omega_k = \nabla\delta_k(\bm\beta;(\bm X_k,\bm y_k))$, $\bm D$ is the data vector and $\bm V_k\in\mathcal{V}_k$ are the optimization variables constrained to a set $\mathcal{V}_k$, $k=1,\ldots, K$. Assuming that pre-selection the asymptotic covariance between any two of $\bm D$, $\bm\omega_1, \ldots, \bm\omega_K$, is zero and $\bm D-\bm\mu_{\bm D}\overset{d}{\rightarrow}\mathcal{N}_{\textnormal{dim}(\bm D)}(\bm 0,\bm\Sigma_{\bm D})$, $\bm\omega_k\overset{d}{\rightarrow}\mathcal{N}_{\textnormal{dim}(\bm\omega_k)}(\bm 0,\bm\Sigma_{\bm\omega_k})$
as $n\rightarrow\infty$, we have that the selective density of $\bm D,\bm V_1,\ldots,\bm V_K$, is
\begin{equation*}
	\phi_{(\bm\mu_{\bm D},\bm\Sigma_{\bm D})}(\bm D)\cdot\prod_{k=1}^K \phi_{(\bm 0,\bm\Sigma_{\bm\omega_k})}(\omega_k(\bm D,\bm V_k))
\end{equation*}
with the restrictions $\bm V_k\in\mathcal{V}_k$. In the case of the $\ell_1$-penalized logistic regression, the randomization reconstruction map becomes $\bm\omega_k = -\bm X^T(\bm y-\pi(\bm X_E\bm\beta_E^*))+\frac{1}{\rho}\bm X_k^T(\bm y_k-\pi(\bm X_{k, E}\bm\beta_E^*))$, for $k=1,\ldots, K$ and in this example the cross covariances between $\bm D, \bm\omega_1,\ldots,\bm\omega_K$ are zero. As elaborated in detail in multiple views of the data for GLMs (Section \ref{sec:mv:glm}), we do inference in the same way. 


\section{Conclusion} 

Inference after selection is a challenging problem since its aim is to not only provide theoretically valid tests and confidence intervals but also make its procedures powerful and computationally efficient.
With added randomization, we gain in both power and computational simplicity compared to non-randomized setting. We have extended the randomized framework to construct a new valid pivot using bootstrap, addressing the challenges in its computation. Through novel examples, including multiple views/queries and data carving, we have illustrated the applicability of our approach. 

The code for all the examples mentioned is available at \url{https://github.com/jonathan-} \url{taylor/selective-inference}. The implementation results along with applications of our methods to a real dataset will be included in future work.


\section*{Acknowledgments}

J.M.~would like to thank Rajarshi Muhkerjee for his generous help in editing this paper.


\input{supplement}

\bibliography{cite}
\bibliographystyle{imsart-nameyear}


\end{document}

%% file: preamble.tex
\newcommand{\handout}[5]{
   \renewcommand{\thepage}{#1-\arabic{page}}
   \noindent
   \begin{center}
   \framebox{
      \vbox{
    \hbox to 5.78in { {\bf 6.893 Sub-linear Algorithms}
     	 \hfill #2 }
       \vspace{4mm}
       \hbox to 5.78in { {\Large \hfill #5  \hfill} }
       \vspace{2mm}
       \hbox to 5.78in { {\it #3 \hfill #4} }
      }
   }
   \end{center}
   \vspace*{4mm}
}

\newtheorem{theorem}{Theorem}
\newtheorem{lemma}[theorem]{Lemma}
\newtheorem{corollary}[theorem]{Corollary}

\newtheorem{remark1}[theorem]{Remark}

\newcommand{\qed}{\rule{7pt}{7pt}}

\newenvironment{proof}{\noindent{\bf Proof}\hspace*{1em}}{\qed\bigskip}
\newenvironment{proof-sketch}{\noindent{\bf Sketch of Proof}\hspace*{1em}}{\qed\bigskip}
\newenvironment{proof-idea}{\noindent{\bf Proof Idea}\hspace*{1em}}{\qed\bigskip}
\newenvironment{proof-of-lemma}[1]{\noindent{\bf Proof of Lemma #1}\hspace*{1em}}{\qed\bigskip}
\newenvironment{proof-attempt}{\noindent{\bf Proof Attempt}\hspace*{1em}}{\qed\bigskip}

\newenvironment{proof-of-theorem}[1]{\noindent{\bf Proof of Theorem #1}\hspace*{1em}}{\qed\bigskip}

\makeatletter
\def\fnum@figure{{\bf Figure \thefigure}}
\def\fnum@table{{\bf Table \thetable}}
\long\def\@mycaption#1[#2]#3{\addcontentsline{\csname
  ext@#1\endcsname}{#1}{\protect\numberline{\csname
  the#1\endcsname}{\ignorespaces #2}}\par
  \begingroup
    \@parboxrestore
    \small
    \@makecaption{\csname fnum@#1\endcsname}{\ignorespaces #3}\par
  \endgroup}
\def\mycaption{\refstepcounter\@captype \@dblarg{\@mycaption\@captype}}
\makeatother

\newcommand{\mathify}[1]{\ifmmode{#1}\else\mbox{$#1$}\fi}
\newcommand{\bigO}O



%% file: supplement.tex







\renewcommand\thesection{\Alph{section}}
\setcounter{section}{0}
\setcounter{theorem}{11}


\begin{center}\textbf{\large{SUPPLEMENT}}\end{center}

\section{Selective CLT - simple version (with parameter convergence assumption)}
\label{app:selective:clt:different}

\subsection{Asymptotic linearity implies uniform CLT}

\begin{lemma}[Asymptotic linearity implies uniform CLT]
 \label{lemma:al:implies:clt}
	\eqref{assumption:asymptotically:linear} assumption implies \eqref{assumption:clt}.
\end{lemma}
\begin{proof}
For any $\epsilon>0$, we decompose
\begin{align*}
	&\mathbb{F}_n\{\sqrt{n}(\bm D_n-\bm \mu_n)\leq \bm t\}
	= \mathbb{F}_n\left\{\frac{1}{\sqrt{n}}\sum_{i=1}^n(\bm\xi_i-\bm\mu)+\sqrt{n}\bm X_n\leq \bm t\right\} \\
	&= \mathbb{F}_n\left\{\frac{1}{\sqrt{n}}\sum_{i=1}^n(\bm\xi_i-\bm\mu)+\sqrt{n}\bm X_n\leq \bm t , \|\sqrt{n}\bm X_n\|_{\infty}\leq \epsilon\right\} \\
	&+\mathbb{F}_n\left\{\frac{1}{\sqrt{n}}\sum_{i=1}^n(\bm\xi_i-\bm\mu) +\sqrt{n}\bm X_n\leq \bm t , \|\sqrt{n}\bm X_n\|_{\infty}\geq \epsilon\right\}.
\end{align*}

Denoting with $\bm\epsilon=(\epsilon,\ldots,\epsilon)\in\mathbb{R}^p$, we have the upper bound
\begin{equation*}
	\mathbb{F}_n\{\sqrt{n}(\bm D_n-\bm \mu_n)\leq \bm t\} 
	\leq \mathbb{F}_n\left\{\frac{1}{\sqrt{n}}\sum_{i=1}^n(\bm\xi_i-\bm\mu)\leq\bm t +\bm \epsilon \right\}	
	+\mathbb{F}_n\{\sqrt{n}\|\bm X_n\|_{\infty}\geq \epsilon\}
\end{equation*}
and the lower bound
\begin{equation*}
	\mathbb{F}_n\{\sqrt{n}(\bm D_n-\bm\mu_n)\leq \bm t\} 
	\geq \mathbb{F}_n\left\{\frac{1}{\sqrt{n}}\sum_{i=1}^n(\bm\xi_i-\bm\mu)\leq\bm t -\bm \epsilon \right\} 
\end{equation*}
Thus,
\begin{align*}
	\left|\mathbb{F}_n\{\sqrt{n}(\bm D_n-\bm\mu_n)\leq \bm t\}	 -\mathbb{F}_n\left\{\frac{1}{\sqrt{n}}\sum_{i=1}^n(\bm\xi_i-\bm\mu)\leq\bm t\right\}\right| \leq \textnormal{max}\{I_1, I_2\},
\end{align*}
where
\begin{equation*}
	I_1=\mathbb{F}_n\left\{\frac{1}{\sqrt{n}}\sum_{i=1}^n(\bm\xi_i-\bm\mu)\in \bm (\bm t, \bm t+\bm\epsilon) \right\} + \mathbb{F}_n\{\sqrt{n}\|\bm X_n\|_{\infty}\geq\epsilon\}
\end{equation*}
and 
\begin{equation*}
	I_2=\mathbb{F}_n\left\{\frac{1}{\sqrt{n}}\sum_{i=1}^n(\bm\xi_i-\bm\mu)\in (\bm t-\bm\epsilon, \bm t)\right\}.
\end{equation*}

Using the assumption $\bm X_n=o_{\mathbb{F}_n}(1/\sqrt{n})$ uniformly across $\mathbb{F}_n\in\mathcal{F}_n$, for any $\epsilon>0$ we have that $\underset{n\rightarrow\infty}{\lim}\underset{\mathbb{F}_n\in\mathcal{F}_n}{\sup}\mathbb{F}_n\{\sqrt{n}\|\bm X_n\|_{\infty}>\epsilon\}=0$.

Letting $n\rightarrow\infty$ and then $\epsilon\rightarrow 0$, we have that $\underset{\mathbb{F}_n\in\mathcal{F}_n}{\sup}\max\{I_1,I_2\}\rightarrow 0$ as $n\rightarrow\infty$. Using \eqref{eq:uniform:clt:xi}, we have that \eqref{assumption:clt} holds.
\end{proof}

\subsection{Selective CLT using pre-selection \eqref{assumption:clt} assumption}

We prove a version of selective CLT that does not assume a asymptotic linearity condition on the statistic $\bm D$ but a weaker CLT statement; however, it requires the unknown parameter and the selection region to converge as the sample size $n$ increase, hence not accounting for the rare selection events.

Assume the selection region is $\sqrt{n}\bm A_M\bm D+\bm\omega\geq \bm b_M$, where $\bm A_M\in\mathbb{R}^{d\times p}$ and $\bm b_M\in\mathbb{R}^d$.
We denote the pivot as
\begin{align*}
	&\mathcal{P}^G\left(\sqrt{n}(\bm D-\bm\mu), \bb A^T\bm\mu, \sqrt{n}\bm A_M\bm\mu, \bm A_M, \bm b_M \right)\\
	&=(\mathbb{P}_{\bm G}\times\mathbb{G}) \bigl \{ \|\bm\Sigma_{\bm A}^{-1/2}\bb G\|_2 \geq \|\sqrt{n}\bm\Sigma_{\bm A}^{-1/2}\bb A^T(\bm D-\bm\mu)\|_2\:\big|\: \\
	&\hspace{8em} \sqrt{n}\bm{A}_M\bm D_{\bb A}+\bm A_M\bm C(\bb G+\sqrt{n}\bb A^T\bm\mu)+\bb\omega\geq \bm b_M, \bm D_{\bb A} \bigr\},\\
	&=(\mathbb{P}_{\bm G}\times\mathbb{G}) \bigl\{ \|\bm\Sigma_{\bm A}^{-1/2}\bb G\|_2 \geq \|\sqrt{n}\bm\Sigma_{\bm A}^{-1/2}\bb A^T(\bm D-\bm\mu)\|_2  \:\big|\\
	& \hspace{8em} \sqrt{n}\bm{A}_M(\bm D_{\bb A}-\bm\mu_{\bb A})+\bm A_M\bm C\bb G+\sqrt{n}\bm A_M\bm\mu+\bb\omega\geq\bm b_M, \bm D_{\bb A}\bigr\},
\end{align*}
where $\mathbb{P}_{\bm G}$ is under $\bb G\sim \mathcal{N}_a(\bb 0,\bb\Sigma_{\bb A})$. Note that here we assume the affine selection region, written before as $\sqrt{n}\bm A_M\bm D+\bm\omega\in\bm H_M$, is given as $\sqrt{n}\bm A_M\bm D+\bm\omega\geq \bm b_M$ for a sequence of vectors $\bm b_M\in\mathbb{R}^d$. We do this in order to simplify the notation for the convergence of the selection region.

To prove a version of selective CLT, let us introduce a new set of assumptions.
\begin{itemize}
\item Convergence of the selection regions:
\begin{equation*}
	\underset{n\rightarrow\infty}{\lim}\bm A_M=\widetilde{\bm A} \;\;\textnormal{ and }\;\; \underset{n\rightarrow\infty}{\lim}\bm b_M=\tilde{\bm b}
\end{equation*}
	for some $\widetilde{\bb A}\in\mathbb{R}^{d\times p}$ and $\tilde{\bm b}\in \mathbb{R}^d$.
\item Convergence of the parameter: 
\begin{equation*} \label{assumption:parameter:conv}
	\underset{n\rightarrow\infty}{\lim}\sqrt{n}\bm{A}_M\bm \mu=\tilde{\bm\mu}.
\end{equation*}
for some $\tilde{\bm\mu}\in\mathbb{R}^d$
\item Assumptions on the distribution of the randomization: $\mathbb{G}$ is continuous and has full support.
\end{itemize}

\begin{theorem}[Selective CLT, version 1] Assuming the two convergence conditions above and \eqref{assumption:clt} hold, we have the following
\begin{align*}
	\underset{n\rightarrow\infty}{\lim}\:\underset{\mathbb{F}_n^*\in\mathcal{F}_n^*}{\sup}\:\underset{t\in[0,1]}{\sup} &\left|\mathbb{F}^*_n\left\{\mathcal{P}^G\left(\sqrt{n}(\bm D-\bm\mu), \bb A^T\bm\mu, \sqrt{n}\bm A_M\bm\mu ,\bm A_M, \bm b_M\right) \leq t \right\} -t \right|=0
\end{align*}
\end{theorem}


\begin{proof}

Denote $\Phi^*_{(\bm 0,\bb\Sigma)}$ the distribution of $\bb G\sim\mathcal{N}_p(\bb 0,\bb \Sigma):=\Phi_{(\bm 0,\bb\Sigma)}$ conditional on the selection
\begin{equation*}
	\widetilde{\bb A}\bm G+\widetilde{\bb\mu}+\bb\omega\geq \tilde{\bm b}.
\end{equation*}
We know that 
\begin{equation*}
	\Phi^*_{(\bm 0,\bb\Sigma)} \left\{\mathcal{P}^G(\bm G,\bb A^T\bm\mu, \tilde{\bm\mu},\widetilde{\bb A}, \tilde{\bb b}) \leq t \right\}=t.
\end{equation*}
Suffices to show that for every continuous bounded function $g:\mathbb{R}\rightarrow\mathbb{R}$, we have
\begin{equation*}
	\mathbb{E}_{\mathbb{F}_n^*}\left[g\left(\mathcal{P}^G\left(\sqrt{n}(\bm D-\bm\mu), \bb A^T\bm\mu, \sqrt{n}\bm A_M\bm\mu,\bm A_M, \bm b_M\right)\right)\right] \rightarrow \mathbb{E}_{\Phi^*_{(\bm 0,\bb\Sigma)}}\left[g(\mathcal{P}^G\left(\bm G,\bb A^T\bm{\mu},\tilde{\bm \mu},\widetilde{\bm A}, \tilde{\bm b} \right)\right]
\end{equation*}
as $n\rightarrow\infty$. Denoting the selective likelihood ratios for the data originally distributed as $\bb D\sim\mathbb{F}_n$ as
\begin{equation*}
	\ell_{\mathbb{F}_n}(\sqrt{n}(\bm D-\bm\mu),\sqrt{n}\bm A_M\bm\mu,\bm A_M,\bm b_M)
		=\frac{\mathbb{G}\left\{\sqrt{n}\bm{A}_M(\bm D-\bm\mu)+\sqrt{n}\bm A_M\bm\mu+\bb\omega\geq\bm b_M\right\}}{\mathbb{E}_{\mathbb{F}_n}\left[\mathbb{G}\left\{\bm\omega\geq\bm b_M-\sqrt{n}\bm A_M(\bm D'-\bm\mu)-\sqrt{n}\bm A_M\bm\mu\right\}\right]}
\end{equation*}
and Gaussian one as
\begin{equation*}
	\ell_{\Phi_{(\bm 0,\bb\Sigma)}}(\bm G, \tilde{\bm\mu},\widetilde{\bb A},\tilde{\bb b})
	=\frac{\mathbb{G}\left\{\widetilde{\bb A}\bm G+\tilde{\bb \mu}+\bb\omega\geq\tilde{\bm b}\right\}}{\mathbb{E}_{\Phi_{(\bm 0,\bb\Sigma)}}\left[\mathbb{G}\left\{\bb\omega\geq\tilde{\bb b}-\widetilde{\bb A}\bm G'-\tilde{\bb\mu}\right\}\right]}
\end{equation*}
it suffices to show the following holds
\begin{align*}
	&\mathbb{E}_{\mathbb{F}_n}\left[g\left(\mathcal{P}^G\left(\sqrt{n}(\bm D-\bm\mu), \bb A^T\bm\mu, \sqrt{n}\bm A_M\bm\mu,\bm A_M, \bm b_M\right)\right) \ell_{\mathbb{F}_n}(\sqrt{n}(\bm D-\bm\mu),\sqrt{n}\bm A_M\bm\mu,\bm A_M,\bm b_M)\right] \\
	&\rightarrow \mathbb{E}_{\Phi_{(\bm 0,\bb\Sigma)}}\left[g\left(\mathcal{P}^G(\bm G,\bb A^T\bm\mu, \tilde{\bm \mu},\widetilde{\bb A}, \tilde{\bb b})\right)\ell_{\Phi_{(\bm 0,\bb\Sigma)}}(\bm G,\tilde{\bm\mu},\widetilde{\bb A},\tilde{\bb b})\right]
\end{align*}
as $n\rightarrow\infty$ for an arbitrary sequence $\mathbb{F}_n\in\mathcal{F}_n$. 
\eqref{assumption:clt} assumption and continuity of $\mathbb{G}$ implies for all $\bb t\in\mathbb{R}^p$
\begin{equation*}
	\underset{n\rightarrow\infty}{\lim}\:\underset{\bb t\in\mathbb{R}^p}{\sup} \left|\ell_{\mathbb{F}_n}(\bb t,\sqrt{n}\bm A_M\bm\mu,\bm A_M,\bm b_M)-\ell_{\Phi_{(\bm 0,\bb\Sigma)}}(\bb t,\widetilde{\bb A},\tilde{\bm\mu},\widetilde{\bb A},\tilde{\bb b})\right|=0.
\end{equation*}
This implies (by Lebesgue dominated convergence theorem)
\begin{equation*}
	\underset{n\rightarrow\infty}{\lim}\mathbb{E}_{\mathbb{F}_n}\left[\left|\ell_{\mathbb{F}_n}\left(\sqrt{n}(\bm D-\bm\mu),\sqrt{n}\bm A_M\bm\mu,\bm A_M,\bm b_M\right)-\ell_{\Phi_{(\bm 0,\bb\Sigma)}}(\sqrt{n}(\bm D-\bm\mu),\tilde{\bm\mu},\widetilde{\bb A},\tilde{\bb b})\right| \right]=0.
\end{equation*}
Since $g$ is bounded function, it suffices to show 
\begin{align*}
	&\mathbb{E}_{\mathbb{F}_n}\left[g \left(\mathcal{P}^G\left(\sqrt{n}(\bm D-\bm\mu),\bb A^T\bm\mu,\sqrt{n}\bm A_M\bm\mu,\bm A_M, \bm b_M\right) \right) \ell_{\Phi_{(\bm 0,\bb\Sigma)}}(\sqrt{n}(\bm D-\bm\mu),\tilde{\bm\mu},\widetilde{\bb A},\tilde{\bb b})\right] \\
	&\rightarrow \mathbb{E}_{\Phi_{(\bm 0,\bb\Sigma)}}\left[g\left(\mathcal{P}^G(\bm G,\bb A^T\bm\mu,\tilde{\bm\mu},\widetilde{\bb A}, \tilde{\bb b} )\right) \ell_{\Phi_{(\bm 0,\bb\Sigma)}}(\bm G,\tilde{\bm\mu},\widetilde{\bb A},\tilde{\bb b})\right]
\end{align*}
as $n\rightarrow\infty$.
By the uniform Continuous mapping theorem the convergence above holds in distribution. To go to expectation, it suffices to have $g$ and $\ell_{\Phi_{(\bm 0,\bb\Sigma)}}$ to be bounded. $g$ is bounded by assumption and $\ell_{\Phi_{(\bm 0,\bb\Sigma)}}$ is bounded above as long as $\mathbb{E}_{\Phi_{\bb \Sigma}}\left[\mathbb{G}\left\{\bb\omega:\bb\omega\geq\tilde{\bb b}-\widetilde{\bb A}\bm G-\widetilde{\bm\mu}\right\}\right]>0$, which holds since $\mathbb{G}$ has full support by assumption.

\end{proof}


\section{Proving selective CLT for Lipschitz randomization} \label{sec:proofs:selective:clt:lip:randomization}

\subsection{Smoothness of the Gaussian likelihood and the pivot for Lipschitz randomization}	

Recall that
\begin{equation*}
	Q_1(\bm z;\bm\Delta)=\mathbb{G}\{\bm A_M\bm z+\bm A_M\bm\Delta+\bm\omega\in\bm H_M\},
\end{equation*}
and the Gaussian selective likelihood ratio becomes
\begin{equation*}
	\ell_{\Phi_n}(\bm z)=\frac{\mathbb{G}\{\bm A_M\bm z+\bm A_M\bm\Delta+\bm\omega\in \bm H_M\}}{\Phi_n\times\mathbb{G}\{\bm A_M\bm Z+\bm A_M\bm\Delta+\bm\omega \in\bm H_M\}} =\frac{Q_1(\bm z;\bm\Delta)}{\mathbb{E}_{\Phi_n}\left[Q_1(\bm Z;\bm\Delta)\right]}.
\end{equation*}

\begin{lemma}[Lipschitz randomization: smoothness of the Gaussian likelihood] Assuming  \eqref{assumption:norm:A_M} and \eqref{assumption:g:smoothness} hold, we have
	\begin{equation*}
		\partial_{\bm z}^{\bm\alpha}\ell_{\Phi_n}(\bm z) \leq \frac{C_M\cdot K_g}{\mathbb{E}_{\Phi_n}\left[e^{-K_g\|\bm A_M\bm Z\|_h}\right]}e^{K_g\|\bm A_M\bm z\|_h}= O(e^{K_g\|\bm A_M\bm z\|_h}).
	\end{equation*} 
\end{lemma}	

\begin{proof}

Using the change of variables $\bm\omega'=\bm\omega+\bm A_M\bm z+\bm A_M\bm\Delta$ and the Lipschitz assumption on $\tilde{g}$, we have the lower bound
\begin{align} 
	Q_1(\bm z;\bm\Delta) 
	&=\int\limits_{\bm\omega\in\bm H_M-\bm A_M\bm z-\bm A_M\bm \Delta}\frac{1}{C_g}\exp(-\tilde{g}(\bm\omega))d\bm\omega\nonumber \\
	&=\frac{1}{C_g}\int\limits_{\bm\omega'\in \bm H_M}\exp\left(-\tilde{g}(\bm\omega'-\bm A_M\bm z-\bm A_M\bm\Delta)\right)d\bm\omega' \nonumber\\
	&\geq\frac{1}{C_g}\int\limits_{\bm\omega'\in\bm H_M} \exp\left(-\tilde{g}(\bm\omega'-\bm A_M\bm\Delta)-K_g\|\bm A_M\bm z\|_h\right)d\bm\omega'\nonumber\\
	&=e^{-K_g\|\bm A_M\bm z\|_h}\int\limits_{\bm\omega'\in\bm H_M} \frac{1}{C_g}\exp\left(-\tilde{g}(\bm\omega'-\bm A_M\bm\Delta)\right) d\bm\omega' \nonumber \\
	&=Q_1(\bm 0;\bm\Delta)\cdot e^{-K_g\|\bm A_M\bm z\|_h},\label{eq:Q1:lower:bound}
\end{align}
implying 
\begin{equation} \label{eq:sel:prob:lower:bound}
\mathbb{E}_{\Phi_n}\left[Q_1(\bm Z;\bm\Delta)\right]
\geq Q_1(\bm 0;\bm\Delta)\cdot\mathbb{E}_{\Phi_n}\left[e^{-K_g\|\bm A_M\bm Z\|_h} \right].
\end{equation}

Similarly, using the Lipschitz property on $\tilde{g}$ again we have an upper bound
\begin{align}
	Q_1(\bm z;\bm\Delta) 
	&\leq \frac{1}{C_g}\int\limits_{\bm\omega'\in\bm H_M} \exp\left(-\tilde{g}(\bm\omega'-\bm A_M\bm\Delta)+K_g\|\bm A_M\bm z\|_h\right)d\bm\omega' \nonumber \\
	&=e^{K_g\|\bm A_M\bm z\|_h}\int\limits_{\bm\omega'\in\bm H_M}\frac{1}{C_g}\exp\left(-\tilde{g}(\bm\omega'-\bm A_M\bm\Delta)\right) d\bm\omega' \nonumber \\
	&=Q_1(\bm 0;\bm\Delta)\cdot e^{K_g\|\bm A_M\bm z\|_h}.\label{eq:Q1:upper:bound}
\end{align}

By the smoothness assumption on $\tilde{g}$, we have 
\begin{align}
	\left|\partial_{\bm z}^{\bm\alpha}	Q_1(\bm z;\bm\Delta) \right|
	&=\frac{1}{C_g}\left|\int\limits_{\bm\omega'\in\bm H_M} \partial_{\bm z}^{\bm\alpha}\exp(-\tilde{g}(\bm\omega'-\bm A_M\bm z-\bm A_M\bm\Delta))d\bm\omega'\right| \nonumber\\
	&\leq C_M \cdot K_g\int\limits_{\bm\omega'\in\bm H_M} \frac{1}{C_g}\exp(-\tilde{g}(\bm\omega'-\bm A_M\bm z-\bm A_M\bm\Delta))d\bm\omega' \nonumber \\
	&=C_M\cdot K_g \cdot Q_1(\bm z;\bm\Delta)\nonumber \\
	&\leq C_M\cdot K_g \cdot Q_1(\bm 0;\bm\Delta)\cdot e^{K_g\|\bm A_M\bm z\|_h}, \label{eq:partial:Q1:upper:bound}
\end{align}
where $C_M$ comes from differentiating $\bm A_M\bm z$ with respect to $\bm z$ and the last inequality follows by using the upper bound on $Q_1(\bm z;\bm\Delta)$ from \eqref{eq:Q1:upper:bound}. Combining \eqref{eq:Q1:lower:bound} and \eqref{eq:partial:Q1:upper:bound} and we have
\begin{equation*}
	\partial_{\bm z}^{\bm\alpha}\ell_{\Phi_n}(\bm z)=\frac{\partial_{\bm z}^{\bm\alpha} Q_1(\bm z;\bm\Delta)}{\mathbb{E}_{\Phi_n}Q_1(\bm Z;\bm\Delta)} \leq \frac{C_M\cdot K_g}{\mathbb{E}_{\Phi_n}\left[e^{-K_g\|\bm A_M\bm Z\|_h}\right]}e^{K_g\|\bm A_M\bm z\|_h}.
\end{equation*}

\end{proof}




Denote 
\begin{equation*}
	Q_2(\bm \tau; \bm z_{\bm A},\bm\Delta) = \mathbb{G}\{\bm A_M\bm z_{\bm A}+\bm A_M\bm C\bm\tau+\bm A_M\bm\Delta+\bm\omega\in\bm H_M\}
\end{equation*}
and
\begin{align*}
	Q_3(\bm z_{\bm A};\bm\Delta) 
	&= \int\limits_{\mathbb{R}^a}\mathbb{G}\{\bm A_M\bm z_{\bm A}+\bm A_M\bm C\bm\tau+\bm A_M\bm\Delta+\bm\omega\in\bm H_M\}\phi_{(\bm 0,\bm\Sigma_{\bm A})}(\bm\tau)d\bm\tau \\
	&= \mathbb{E}_{\bm G\sim\mathcal{N}_a(\bm 0,\bm\Sigma_{\bm A})}\left[Q_2(\bm G;\bm z_{\bm A},\bm\Delta)\right].
\end{align*}	
Then the pivot equals
\begin{equation*}
\mathcal{P}^G(\bm Z,\bm A_M\bm\Delta, \bm A_M, \bm H_M)=\frac{\int\limits_{\|\bm\Sigma_{\bm A}^{-1/2}\bm\tau\|_2\geq \|\bm\Sigma_{\bm A}^{-1/2}\bm A^T\bm Z\|_2} Q_2(\bm\tau;\bm Z_{\bm A},\bm\Delta)\phi_{(\bm 0,\bm\Sigma_{\bm A})}(\bm\tau)d\bm\tau}{Q_3(\bm Z_{\bm A};\bm\Delta)},
\end{equation*}
where with the abuse of notation we write different the arguments of the pivot slightly differently in this section.

\begin{lemma}[Lipschitz randomization: smoothness of the pivot] \label{lemma:pivot:smoothness:lip:randomization} Assuming \eqref{assumption:norm:A_M} and \eqref{assumption:g:smoothness} hold, we have
\begin{align*}
	\left|\partial_{\bm z}^{\bm\alpha}\mathcal{P}^G(\bm z,\bm A_M\bm\Delta,\bm A_M,\bm H_M)\right|
	&\leq C(\bm A, \bm A_M, K_g)+C(\bm A)\frac{e^{K_g\|\bm A_M\bm C\bm\Sigma_{\bm A}\|_{h,2}\|\bm\Sigma_{\bm A}^{-1/2}\bm A^T\bm z\|_2}}{\mathbb{E}_{\bm G\sim\mathcal{N}_a(\bm 0,\bm\Sigma_{\bm A})}\left[e^{-K_g\|\bm A_M\bm C\bm G\|_h}\right]} \\
	&=O\left(1+e^{K_g\|\bm A_M\bm C\bm\Sigma_{\bm A}\|_{h,2}\|\bm\Sigma_{\bm A}^{-1/2}\bm A^T\bm z\|_2}\right),
	\end{align*}
	where the constants above depend only on their arguments.
\end{lemma}
	
\begin{proof}

By the change of variables $\bm\omega=\bm\omega'-\bm A_M\bm z_{\bm A}-\bm A_M\bm C\bm\tau-\bm A_M\bm\Delta$ and using the Lipschitz property of the randomization we get the lower bound on $Q_2(\bm\tau;\bm z_{\bm A},\bm\Delta)$:
\begin{equation*}
\begin{aligned}
		Q_2(\bm\tau;\bm z_{\bm A},\bm\Delta) 
		&= \int\limits_{\bm\omega\in\bm H_M-\bm A_M\bm z_{\bm A}-\bm A_M\bm C\bm\tau-\bm A_M\bm\Delta}\frac{1}{C_g}\exp(-\tilde{g}(\bm\omega))d\bm\omega \\
		&=\frac{1}{C_g}\int\limits_{\bm\omega'\in\bm H_M}\exp\left(-\tilde{g}(\bm\omega'-\bm A_M\bm z_{\bm A}-\bm A_M\bm C\bm\tau-\bm A_M\bm\Delta)\right)d\bm\omega' \\
		&\geq\frac{1}{C_g}\int\limits_{\bm\omega'\in\bm H_M} \exp\left(-\tilde{g}(\bm\omega'-\bm A_M\bm z_{\bm A}-\bm A_M\bm\Delta)-K_g\|\bm A_M\bm C\bm\tau\|_h\right)d\bm\omega' \\
		&=Q_2(\bm 0;\bm z_{\bm A},\bm\Delta)\cdot e^{-K_g\|\bm A_M\bm C\bm\tau\|_h}.
\end{aligned}
\end{equation*}
Hence
\begin{equation} \label{eq:lip:randomization:denom:pivot:lower:bound}
	Q_3(\bm z_{\bm A};\bm\Delta) \geq Q_2(\bm 0;\bm z_{\bm A},\bm\Delta)\cdot \mathbb{E}_{\bm G\sim\mathcal{N}_a(\bm 0,\bm\Sigma_{\bm A})}\left[e^{-K_g\|\bm A_M\bm C\bm G\|_h}\right].
\end{equation}

As for the upper bound on $Q_2(\bm\tau;\bm z_{\bm A},\bm\Delta)$, we have
\begin{equation*}
\begin{aligned}
	Q_2(\bm\tau;\bm z_{\bm A},\bm\Delta) 
	&\leq\frac{1}{C_g}\int\limits_{\bm\omega'\in\bm H_M} \exp\left(-\tilde{g}(\bm\omega'-\bm A_M\bm z_{\bm A}-\bm A_M\bm\Delta)+K_g\|\bm A_M\bm C\bm\tau\|\right)d\bm\omega'\\
	&=Q_2(\bm 0;\bm z_{\bm A};\bm\Delta)\cdot e^{K_g\|\bm A_M\bm C\bm\tau\|_h}.
\end{aligned}
\end{equation*}
Decomposing $\bm\tau=\bm\Sigma_{\bm A}^{1/2}\|\bm \Sigma_A^{-1/2}\bm\tau\|_2\frac{\bm\Sigma_{\bm A}^{-1/2}\bm\tau}{\|\bm \Sigma_A^{-1/2}\bm\tau\|_2}$ and denoting $r=\|\bm \Sigma_A^{-1/2}\bm\tau\|_2$ and $\bm u=\frac{\bm\Sigma_{\bm A}^{-1/2}\bm\tau}{\|\bm \Sigma_A^{-1/2}\bm\tau\|_2}$, we know that under $\bm\tau\sim\mathcal{N}(\bm 0,\bm\Sigma_{\bm A})$, $r$ and $\bm u$ are independent random variables. Denote their densities with $f_r(r)$ and $f_{\bm u}e(\bm u)$, respectively.

Then the derivative of the numerator with respect to $\bm A^T\bm z$ satisfies
\begin{align*}
	&\left|\partial_{\bm A^T\bm z}^{\bm\alpha}\int_{r\geq \|\bm\Sigma_{\bm A}^{-1/2}\bm A^T\bm z\|_2}  \int_{\bm u\in\mathcal{S}^{a}} Q_2\left(r\bm\Sigma_{\bm A}^{1/2}\bm u; \bm z_{\bm A},\bm\Delta\right)f_r(r)f_{\bm u}(\bm u) dr d\bm u\right| \\
	&\leq C(\bm A)\cdot  \int_{\bm u\in\mathcal{S}^{a}} Q_2\left(\|\bm\Sigma_{\bm A}^{-1/2}\bm A^T\bm z\|_2\bm\Sigma_{\bm A}^{1/2}\bm u; \bm z_{\bm A},\bm\Delta\right)f_r(\|\bm\Sigma_{\bm A}^{-1/2}\bm A^T\bm z\|_2)f_{\bm u}(\bm u) d\bm u \\
	&\leq C(\bm A)\cdot Q_2(\bm 0;\bm z_{\bm A},\bm\Delta)\cdot e^{K_g\|\bm A_M\bm C\bm\Sigma_{\bm A}^{1/2}\bm u\|_h\|\bm\Sigma_{\bm A}^{-1/2}\bm A^T\bm z\|_2} f_r(\|\bm\Sigma_{\bm A}^{-1/2}\bm A^T\bm z\|_2)\int_{\bm u\in\mathcal{S}^{a}}f_{\bm u}(\bm u) d\bm u,
\end{align*}
where $\mathcal{S}^a$ is a unit sphere in $\mathbb{R}^a$ and in the last inequality we used the upper bound on $Q_2$.
Since the last integral is 1 and the density $f_r(\cdot)$ is uniformly bounded by a constant, we get that the derivative above is upper bounded by 
\begin{equation*}
	C(\bm A)\cdot Q_2(\bm 0;\bm z_{\bm A};\bm\Delta)\cdot e^{K_g\|\bm A_M\bm C\bm\Sigma_{\bm A}\|_{h,2}\|\bm\Sigma^{-1/2}_{\bm A}\bm A^T\bm z\|_2}.
\end{equation*}

Combing this upper bound with the lower bound on the denominator from \eqref{eq:lip:randomization:denom:pivot:lower:bound}, we get that the derivative of the pivot with respect to $\bm A^T\bm z$ is upper bounded by
\begin{equation*}
	\left|\partial_{\bm A^T\bm z}^{\bm\alpha}\mathcal{P}^G(\bm z,\bm A_M\bm\Delta,\bm A_M,\bm H_M)\right|\leq C(\bm A)\frac{e^{K_g\|\bm A_M\bm C\bm\Sigma_{\bm A}\|_{h,2}\|\bm\Sigma_{\bm A}^{-1/2}\bm A^T\bm z\|_2}}{\mathbb{E}_{\bm G\sim\mathcal{N}_a(\bm 0,\bm\Sigma_{\bm A})}\left[e^{-K_g\|\bm A_M\bm C\bm G\|_h}\right]}.
\end{equation*}


Now we turn to the derivative of the pivot with respect to $\bm z_{\bm A}$. The derivative of the integrand of the numerator of the pivot is upper bounded by
\begin{equation*}
\begin{aligned}
		\left|\partial_{\bm z}^{\bm\alpha} Q_2(\bm\tau;\bm z_{\bm A},\bm\Delta) \right| 
		&=\frac{1}{C_g}\left|\int\limits_{\bm\omega'\in\bm H_M}\partial_{\bm z}^{\bm\alpha}\exp\left(-\tilde{g}(\bm\omega'-\bm A_M\bm z_{\bm A}-\bm A_M\bm C\bm\tau-\bm A_M\bm\Delta)\right)d\bm\omega'\right| \\
		&\leq C(\bm A_M,\bm A)\frac{K_g}{C_g}\int\limits_{\bm\omega'\in\bm H_M}\exp\left(-\tilde{g}(\bm\omega'-\bm A_M\bm z_{\bm A}-\bm A_M\bm C\bm\tau-\bm A_M\bm\Delta)\right)d\bm\omega'\\
		&= C(\bm A_M, \bm A)\cdot K_g \cdot Q_2(\bm\tau;\bm z_{\bm A},\bm\Delta) \\
\end{aligned}
\end{equation*}
and, similarly, the derivative of the denominator is upper bounded by 
\begin{equation}
	\left|\partial_{\bm z_{\bm A}}^{\bm\alpha}Q_3(\bm z_A;\bm\Delta)\right| \leq C(\bm A_M, \bm A)\cdot K_g \cdot Q_3(\bm z_{\bm A};\bm\Delta).
\end{equation}

Note that the derivative of the pivot with respect to $\bm z_{\bm A}$ will be the sum of the term of the form
\begin{equation*}
	\frac{\int\limits_{\|\bm\Sigma_{\bm A}^{-1/2}\bm\tau\|_2\geq\|\bm \Sigma_{\bm A}^{-1/2}\bm A^T\bm z\|_2}\partial_{\bm z}^{\bm\alpha}Q_2(\bm\tau;\bm z_{\bm A},\bm\Delta)\phi_{(\bm 0,\bm\Sigma_{\bm A})}(\bm\tau)d\bm\tau}{Q_3(\bm z_{\bm A};\bm\Delta)}
\end{equation*}
and
\begin{equation*}
	\frac{\left(\int\limits_{\|\bm\Sigma_{\bm A}^{-1/2}\bm\tau\|_2\geq\|\bm\Sigma_{\bm A}^{-1/2}\bm A^T\bm z\|_2}Q_2(\bm\tau;\bm z_{\bm A},\bm\Delta)\phi_{(\bm 0,\bm\Sigma_{\bm A})}(\bm\tau)d\bm\tau \right) \cdot \partial_{\bm z_{\bm A}}^{\bm \alpha}Q_3(\bm z_{\bm A};\bm\Delta)}{Q_3(\bm z_{\bm A};\bm\Delta)^2}
\end{equation*}
The first term above in the absolute value is upper bounded by
\begin{align*}
\frac{\int\limits_{\mathbb{R}^a}|\partial_{\bm z}^{\bm\alpha}Q_2(\bm\tau;\bm z_{\bm A},\bm\Delta)|\phi_{(\bm 0,\bm\Sigma_{\bm A})}(\bm\tau)d\bm\tau}{Q_3(\bm z_{\bm A};\bm\Delta)} 
&\leq \frac{C(\bm A_M,\bm A)\cdot K_g\cdot\mathbb{E}_{\bm G\sim\mathcal{N}_a(\bm 0,\bm \Sigma_{\bm A})}\left[Q_2(\bm\tau;\bm z_{\bm A},\bm\Delta)\right]}{Q_3(\bm z_{\bm A};\bm\Delta)}\\
&=C(\bm A_M,\bm A)\cdot K_g
\end{align*}
and the second term in absolute value is upper bounded by
\begin{equation*}
	\frac{\partial_{\bm z_{\bm A}}^{\bm\alpha}Q_3(\bm z_{\bm A};\bm\Delta)}{Q_3(\bm z_{\bm A};\bm\Delta)}\leq C(\bm A_M,\bm A)\cdot K_g.
\end{equation*}

\end{proof}	


\subsection{Applying Chatterjee's theorem for Lipschitz randomization}

Denote as $\bm y_i$ the centered versions of $\bm\xi_i$, $i=1,\ldots, n$. Let $\mathcal{M}:\mathbb{R}^{n\times p}\rightarrow \mathbb{R}^p$ denote the normalizing operator:
\begin{equation*} 
	\mathcal{M}(\bm y_1,\ldots,\bm y_n)=\frac{1}{\sqrt{n}}\sum_{i=1}^n \bm y_i.
\end{equation*}
Then our normalized test statistic becomes $\bm Z=\mathcal{M}(\bm y_1,\ldots,\bm y_n)$. Also, take $\bm g_1,\ldots,\bm g_n\overset{i.i.d.}{\sim}\mathcal{N}_p(\bm 0,\bm\Sigma_n)$:
\begin{equation*}
	\bm G=\left(\begin{matrix}
		\bm g_1^T \\ \vdots \\ \bm g_n^T	
		\end{matrix}\right)_{n\times p}.
\end{equation*}
Denote two sequences of matrices $\{\bm W_i:i=0,\ldots, n\}$, $\{\widetilde{\bm W}_i:i=0,\ldots,n\}$ as
	\begin{equation*}
		\bm W_i=\left(\begin{matrix}
\bm g_1^T \\ \vdots \\ \bm g_{i-1}^T \\ \bm 0 \\ \bm y_{i+1}^T \\ \vdots \\ \bm y_n^T	
\end{matrix}\right)_{n\times p},\;\;\;\; \widetilde{\bm W}_i=\left(\begin{matrix}
	\bm g_1^T \\ \vdots \\ \bm g_{i-1}^T \\ \bm g_i^T \\ \bm y_{i+1}^T \\ \vdots \\ \bm y_n^T
\end{matrix}\right)_{n\times p}.
	\end{equation*}

Here we present a generalization of Theorem 1 in \cite{chatterjee2005simple}, which provides a bound on the difference between smooth functions of Gaussian and averages of i.i.d.~random variables coming from a nonparametric distribution. 

Let $\Omega_n:\mathbb{R}^p\rightarrow\mathbb{R}$ (we will later take $\Omega(\bm z)=\|\bm A_M\bm z\|_h$) be a sequence of norm operators and let $\mathcal{W}_n:\mathbb{R}\rightarrow\mathbb{R}$ be a sequence of weights (we will later take $\mathcal{W}_n(x)=e^{K_gx}$).
For $f:\mathbb{R}^{n\times p}\rightarrow\mathbb{R}$, we define
\begin{equation*}
	\lambda_3^{(n)}(f) = \sup\left\{ \frac{\|\partial_{y_i}^{3} f(\bm y_1,\ldots,\bm y_n)\|_{\infty}}{\mathcal{W}_n\left(\Omega_n(\mathcal{M}(\bm y_1,\ldots,\bm y_n))\right)} , \;\; i=1,\ldots,n \right\}.
\end{equation*}

Now we have the notation necessary to state the theorem that bounds the difference between the expectations of a function applied to nonparametric and the Gaussian data in terms of the smoothness of the function.

\begin{theorem}[\cite{chatterjee2005simple}]\label{thm:chatterjee}
For $\mathcal{W}_n$ increasing and convex, we have
\begin{equation} \label{eq:chatterjee:bound}
	\begin{aligned}
	    \left|\mathbb{E}_{\mathbb{F}_n}\left[f(\bm Y)\right]-\mathbb{E}_{\Phi_n}\left[f(\bm G)\right]\right|
	    &\leq \frac{1}{2} \sum_{i=1}^n \lambda_3^{(n)}(f)\mathbb{E}\left[\mathcal{W}_n \left(2\cdot\Omega_n(\mathcal{M}(\bm W_i))\right)\right] \left(\mathbb{E}\left[\|\bm g_i\|_1^3\right] +\mathbb{E}\left[\|\bm y_i\|_1^3 \right]\right) \\
		&+\frac{1}{2} \sum_{i=1}^n \lambda_3^{(n)}(f)\mathbb{E}\left[\mathcal{W}_n\left(2\cdot\Omega_n(\mathcal{M}(\widetilde{\bm W}_i-\bm W_i))\right) \|\bm g_i\|_1^3\right]\\
		& + \frac{1}{2} \sum_{i=1}^n \lambda_3^{(n)}(f)\mathbb{E}\left[\mathcal{W}_n\left(2\cdot\Omega_n(\mathcal{M}(\widetilde{\bm W}_{i-1}-\bm W_i))\right) \|\bm y_i\|_1^3 \right].
	\end{aligned}
	\end{equation}
\end{theorem}



We now apply Theorem \ref{thm:chatterjee} to show the Gaussian and non-parametric likelihoods are close in expectation and to show that the constructed pivot under the non-parametric distribution asymptotically behaves as if the data was Gaussian. Since we know that the pivot is uniform under the Gaussian distribution, it will be asymptotically uniform under a non-parametric distribution.

\begin{lemma}[Lipschitz randomization: closeness of Gaussian and non-parametric LR] Assuming \eqref{assumption:asymptotically:linear},  \eqref{assumption:mgf}, \eqref{assumption:norm:A_M} and \eqref{assumption:g:smoothness} hold, we have
\begin{equation*}
	\mathbb{E}_{\mathbb{F}_n}\left[\left|\ell_{\mathbb{F}_n}(\bm Z)-\ell_{\Phi_n}(\bm Z)\right|\right]\leq \frac{1}{\sqrt{n}} C(C_{\tau}, C_M, K_g, K_G,\bm\Sigma),
	\end{equation*}
	where the constant on the RHS depends only on its arguments.
	\label{lemma:lip:randomization:lik:diff}
\end{lemma}

\begin{proof}

Since
\begin{equation*}
\begin{aligned}
	&\mathbb{E}_{\mathbb{F}_n}\left[\left|\ell_{\mathbb{F}_n}(\bm Z)-\ell_{\Phi_n}(\bm Z)\right|\right] = \mathbb{E}_{\mathbb{F}_n}\left[Q_1(\bm Z;\bm\Delta)\right]\left|\frac{1}{\mathbb{E}_{\mathbb{F}_n}\left[Q_1(\bm Z;\bm\Delta)\right]}-\frac{1}{\mathbb{E}_{\Phi_n}\left[Q_1(\bm Z;\bm\Delta)\right]}\right| \\
	& = \frac{\left|\mathbb{E}_{\mathbb{F}_n}\left[Q_1(\bm Z;\bm\Delta)\right]-\mathbb{E}_{\Phi_n}\left[Q_1(\bm Z;\bm\Delta)\right]\right|}{\mathbb{E}_{\Phi_n}\left[Q_1(\bm Z;\bm\Delta)\right]},
\end{aligned}
\end{equation*}
we have
\begin{equation*}
	\mathbb{E}_{\mathbb{F}_n}\left[\left|\ell_{\mathbb{F}_n}(\bm Z)-\ell_{\Phi_n}(\bm Z)\right|\right] \leq \frac{\left|\mathbb{E}_{\mathbb{F}_n}\left[\widetilde{Q}_1(\bm Z;\bm\Delta)\right]-\mathbb{E}_{\Phi_n}\left[\widetilde{Q}_1(\bm Z;\bm\Delta)\right]\right|}{\mathbb{E}_{\Phi_n}\left[e^{-K_g\|\bm A_M\bm Z\|_h}\right]},
\end{equation*}
where we used the lower bound on $\mathbb{E}_{\Phi_n}\left[Q_1(\bm Z;\bm\Delta)\right]$ from \eqref{eq:sel:prob:lower:bound} and $\widetilde{Q}_1(\bm z;\bm\Delta)$ denotes
\begin{equation*}
	\widetilde{Q}_1(\bm z;\bm\Delta) = \frac{Q_1(\bm z;\bm\Delta)}{Q_1(\bm 0,\bm\Delta)}.
\end{equation*}
	Hence it suffices to provide an upper bound on
\begin{equation*}
		\left|\mathbb{E}_{\mathbb{F}_n}\left[\widetilde{Q}_1(\bm Z;\bm\Delta)\right]-\mathbb{E}_{\Phi_n}\left[\widetilde{Q}_1(\bm Z;\bm\Delta)\right]\right|.
\end{equation*}
	We use Chatterjee technique here, providing an upper bound on the quantity above using the smoothness of $\widetilde{Q}_1(\bm z;\bm\Delta)$. From \eqref{eq:partial:Q1:upper:bound} we have
\begin{equation*}
	\left|\partial_{\bm z}^{\bm\alpha}\widetilde{Q}_1(\bm z;\bm\Delta)\right|\leq C_M\cdot K_g\cdot e^{K_g\|\bm A_M\bm z\|_h},
\end{equation*}
implying
\begin{equation*}
	\lambda_3^{(n)}(\widetilde{Q}_1) \leq \frac{1}{n^{3/2}} C_M\cdot K_g
\end{equation*}
for $\Omega_n(\bm z)=\|\bm A_M\bm z\|_h$, $\bm z\in\mathbb{R}^p$, and $\mathcal{W}_n(x)=e^{K_gx}$, $x\in\mathbb{R}$.
Using Theorem \ref{thm:chatterjee}, we have
\begin{equation} \label{eq:chaterjee:bound:sum}
\begin{aligned}
		&\left|\mathbb{E}_{\mathbb{F}_n}\left[\widetilde{Q}_1(\bm Z;\bm\Delta)\right]-\mathbb{E}_{\Phi_n}\left[\widetilde{Q}_1(\bm Z;\bm\Delta)\right]\right| \\
		&\leq \frac{1}{2} \sum_{i=1}^n \lambda_3^{(n)}(\widetilde{Q}_1)\mathbb{E}\left[e^{2K_g \|\bm A_M\mathcal{M}(\bm W_i)\|_h}\right] \left(\mathbb{E}\left[\|\bm g_i\|_1^3\right] +\mathbb{E}\left[\|\bm y_i\|_1^3 \right]\right) \\
		&+\frac{1}{2} \sum_{i=1}^n \lambda_3^{(n)}(\widetilde{Q}_1)\mathbb{E}\left[e^{2K_g\|\bm A_M\mathcal{M}(\widetilde{\bm W}_i-\bm W_i)\|_h} \|\bm g_i\|_1^3\right] \\
		&+ \frac{1}{2} \sum_{i=1}^n \lambda_3^{(n)}(\widetilde{Q}_1)\mathbb{E}\left[e^{2K_g\|\bm A_M\mathcal{M}(\widetilde{\bm W}_{i-1}-\bm W_i)\|_h}\|\bm y_i\|_1^3 \right].
\end{aligned}
\end{equation}
Note that we need \eqref{assumption:asymptotically:linear} assumption to be able to apply this theorem.
Following the proof of Theorem 9 in \cite{tian2015selective}, which uses \eqref{assumption:mgf} assumption, the sum in the RHS in \eqref{eq:chaterjee:bound:sum} is of order $O(n^{-1/2})$.

\end{proof}


Using the same proof as in Lemma \ref{lemma:lip:randomization:lik:diff}, we bound the difference between the expectations of a smooth function applied to the pivots constructed for non-parametric and Gaussian data. The result is also under the conditional distributions of the data ($\mathbb{F}_n^*$ and $\Phi_n^*$ for the non-parametric and Gaussian respectively).

\begin{theorem}[Lipschitz randomization: selective CLT] \label{thm:selective:clt:finite:n:lipschitz}
Assume \eqref{assumption:asymptotically:linear},  \eqref{assumption:mgf}, \eqref{assumption:norm:A_M} and \eqref{assumption:g:smoothness} hold.
Given a function $\mathcal{H}:\mathbb{R}\rightarrow\mathbb{R}$ with uniformly bounded derivatives up to the third order (call this bound $K_\mathcal{H}$), we have 
\begin{equation*}
	\left|\mathbb{E}_{\mathbb{F}_n^*}\mathcal{H}(\mathcal{P}^G(\bm Z)) - \mathbb{E}_{\Phi_n^*}\mathcal{H}(\mathcal{P}^G(\bm G))\right| \leq \frac{1}{\sqrt{n}} C(C_{\tau}, C_M, K_g, K_{\mathcal{H}}, \bm\Sigma).
\end{equation*}
\end{theorem}

\begin{proof}
	Since 
	\begin{equation*}
	\begin{aligned}
		&\left|\mathbb{E}_{\mathbb{F}_n^*}\mathcal{H}(\mathcal{P}^G(\bm Z)) - \mathbb{E}_{\Phi_n^*}\mathcal{H}(\mathcal{P}^G(\bm G))\right| \\
		&= \left|\mathbb{E}_{\mathbb{F}_n}\left[\mathcal{H}(\mathcal{P}^G(\bm Z)\ell_{\mathbb{F}_n}(\bm Z)\right] -\mathbb{E}_{\Phi_n}\left[\mathcal{H}(\mathcal{P}^G(\bm G)\ell_{\Phi_n}(\bm G)\right]\right|, \\
		&=  \left|\mathbb{E}_{\mathbb{F}_n}\left[\mathcal{H}(\mathcal{P}^G(\bm Z)\left(\ell_{\mathbb{F}_n}(\bm Z)-\ell_{\Phi_n}(\bm Z)\right)\right]\right|
		+  \left|\mathbb{E}_{\mathbb{F}_n}\left[\mathcal{H}(\mathcal{P}^G(\bm Z)\ell_{\Phi_n}(\bm Z)\right] -\mathbb{E}_{\mathbb{F}_n}\left[\mathcal{H}(\mathcal{P}^G(\bm G)\ell_{\Phi_n}(\bm G)\right]\right| \\
		&\leq K_\mathcal{H}\cdot \mathbb{E}_{\mathbb{F}_n}\left[\left|\ell_{\mathbb{F}_n}(\bm Z)-\ell_{\Phi_n}(\bm Z)\right|\right] +\left|\mathbb{E}_{\mathbb{F}_n}\left[\mathcal{H}(\mathcal{P}^G(\bm Z)\ell_{\Phi_n}(\bm Z)\right] -\mathbb{E}_{\mathbb{F}_n}\left[\mathcal{H}(\mathcal{P}^G(\bm G)\ell_{\Phi_n}(\bm G)\right]\right|,
	\end{aligned}
	\end{equation*}
	to bound the first term above we use the result of Lemma \ref{lemma:lip:randomization:lik:diff} and to bound the second term we apply Theorem \ref{thm:chatterjee} to
	\begin{equation*}
		f(\bm Y) = \mathcal{H}(\mathcal{P}^G(\mathcal{M}(\bm Y)))\cdot \ell_{\Phi_n}(\mathcal{M}(\bm Y)).
	\end{equation*}
	Applying Theorem \ref{thm:chatterjee} is then done in the same way as in Lemma \ref{lemma:lip:randomization:lik:diff}, where we additionally use Lemma \ref{lemma:pivot:smoothness:lip:randomization}.

\end{proof}


\section{Proving Selective CLT for Gaussian randomization} \label{sec:proofs:selective:clt:gaussian:randomization}

\subsection{Smoothness of the Gaussian likelihood and the pivot for Gaussian randomization} \label{sec:smoothness:lemmas:gaussian:randomization}

Denoting the probability of selection conditional on the data as
\begin{equation*}
	Q_1(\bm z;\bm\Delta)=\mathbb{G}\{\bm A_M\bm z+\bm A_M\bm\Delta+\bm\omega\in\bm H_M\},
\end{equation*}
we have
\begin{equation*}
	\ell_{\Phi_n}(\bm z)=\frac{\mathbb{G}\{\bm A_M\bm z+\bm A_M\bm\Delta+\bm\omega\in \bm H_M\}}{\Phi_n\times\mathbb{G}\{\bm A_M\bm Z+\bm A_M\bm\Delta+\bm\omega \in\bm H_M\}} =\frac{Q_1(\bm z;\bm\Delta)}{\mathbb{E}_{\Phi_n}\left[Q_1(\bm Z;\bm\Delta)\right]},
\end{equation*}
where we write the likelihood in terms of $\bm z=\sqrt{n}(\bm t-\bm\mu)$. In this section, we take the randomization distribution to be Gaussian
\begin{equation*}
	\mathbb{G}=\mathcal{N}_d(\bm 0, c\bm\Sigma_{\bm\omega}).
\end{equation*}

Assume that $\bm\Sigma_{\bm\omega}=c\bm I_d$ and that the selection region is rectangular $\bm H_M=\prod_{i=1}^n[b_i,\infty)$ for some vector $\bm b\in\mathbb{R}^d$.


\begin{lemma}[Gaussian randomization: smoothness of the Gaussian likelihood] Assume \eqref{assumption:gaussian:randomization} holds. Then
 for any $\kappa_2<\frac{1}{2c}$, there exists $\kappa_1$ such that
	\begin{equation*}
		e^{-\kappa_2\|\bm A_M\bm z\|_2^2}\cdot \left|\partial_{\bm z}^{\bm\alpha}\ell_{\Phi_n}(\bm z)\right|=O\left(e^{\kappa_1d(\bm 0, \bm H_M-\bm A_M\bm\Delta)^2}\right).
	\end{equation*} 
	\label{lemma:lik:smoothness:gaussian:randomization}
\end{lemma}	

\begin{proof}

The denominator of the likelihood satisfies
\begin{align*}
	\mathbb{E}_{\Phi_n}\left[Q_1(\bm Z;\bm\Delta)\right]
	&=\mathbb{P}_{\bm Z}\times\mathbb{P}_{\bm\omega}\{\bm A_M\bm Z+\bm\omega\in \bm H_M-\bm A_M\bm\Delta\} \\
	&= \mathbb{P}_{\bm \omega'\sim\mathcal{N}(\bm 0, \bm A_M\bm \Sigma \bm A_M^T+\bm\Sigma_{\bm\omega})}\{\bm\omega'\in \bm H_M-\bm A_M\bm\Delta\} \\
	&= const\cdot \int_{\bm H_M}\exp\left(-\frac{1}{2}(\bm \omega'-\bm A_M\bm\Delta)^T(\bm A_M\bm\Sigma\bm A_M^T+\bm\Sigma_{\bm\omega})^{-1}(\bm\omega'-\bm A_M\bm\Delta)\right)d\bm\omega' \\
	&= const\cdot \int_{\bm H_M} \exp\left(-\frac{1}{2}\left\|(\bm A_M\bm\Sigma\bm A_M^T+\bm\Sigma_{\bm\omega})^{-1/2}(\bm \omega'-\bm A_M\bm\Delta)\right\|_2^2\right)d\bm \omega'\\
	&\geq const \cdot \int_{\bm H_M} \exp\left(-\frac{1}{2}\|(\bm A_M\bm\Sigma\bm A_M^T+\bm\Sigma_{\bm\omega})^{-1/2}\|_2^2\left\|\bm \omega'-\bm A_M\bm\Delta)\right\|_2^2\right)d\bm \omega' \\
	&\geq const \cdot \exp\left(-\frac{\delta_1}{2}\|(\bm A_M\bm\Sigma\bm A_M^T+\bm\Sigma_{\bm\omega})^{-1/2}\|_2^2 \cdot d(\bm 0, \bm H_M-\bm A_M\bm\Delta)^2\right),
\end{align*}
for some $\delta_1>1$.
The numerator of the pivot satisfies 
\begin{align*}
	\left|\partial_{\bm z}^{\bm\alpha}Q_1(\bm z;\bm\Delta) \right|
	&=const\cdot\left|\int\limits_{\bm\omega'\in\bm H_M} \partial_{\bm z}^{\bm\alpha}\exp\left(-\frac{1}{2c}\|\bm\omega'-\bm A_M\bm z-\bm A_M\bm\Delta\|_2^2\right)d\bm\omega'\right| \nonumber\\
	&\leq const\cdot \int\limits_{\bm\omega'\in\bm H_M} \exp\left(-\frac{\delta_2}{2c}\|\bm\omega'-\bm A_M\bm z-\bm A_M\bm\Delta\|_2^2\right)d\bm\omega'\\
	&\leq const\cdot \exp\left(-\frac{\delta_2'}{2c} d\left(\bm 0, \bm H_M-\bm A_M\bm z-\bm A_M\bm\Delta\right)^2\right),
\end{align*}
for some $\delta_2$ and $\delta_2'$ satisfying $\delta_2'<\delta_2<1$.
By the triangle inequality
\begin{equation*}
	d(\bm 0,\bm H_M-\bm A_M\bm\Delta)\leq d(\bm 0,\bm H_M-\bm A_M\bm z-\bm A_M\bm\Delta)+\|\bm A_M\bm z\|_2,
\end{equation*}
hence 
\begin{equation*}
	d(\bm 0,\bm H_M-\bm A_M\bm\Delta)^2\leq 2d(\bm 0,\bm H_M-\bm A_M\bm z-\bm A_M\bm\Delta)^2+2\|\bm A_M\bm z\|_2^2.
\end{equation*}
Hence we get a further bound on the derivative of the numerator 
\begin{align*}
	\left|\partial_{\bm z}^{\bm\alpha}Q_1(\bm z;\bm\Delta) \right| \leq const\cdot \exp\left(-\frac{\delta_2'}{4c}d(\bm 0, \bm H_M-\bm A_M\bm\Delta)^2+\frac{\delta_2'}{2c}\|\bm A_M\bm z\|_2^2\right).
\end{align*}

Combining the upper bound on the derivative of the numerator together with the lower bound on the denominator we get that for any $\kappa_2<\frac{1}{2c}$, there exists $\kappa_1$ ($\kappa_1$ can be taken to be $-\frac{\delta_2'}{4c}+\frac{\delta_1}{2}\|(\bm A_M\bm \Sigma\bm A_M^T+\bm\Sigma_{\bm\omega})^{-1/2}\|_2^2$) such that
\begin{equation*}
	\left|\partial_{\bm z}^{\bm\alpha}\ell_{\Phi_n}(\bm z)\right|=\frac{\left|\partial_{\bm z}^{\bm\alpha} Q_1(\bm z;\bm\Delta)\right|}{\mathbb{E}_{\Phi_n}\left[Q_1(\bm Z;\bm\Delta)\right]} \leq const\cdot \exp\left(\kappa_1 d(\bm 0,\bm H_M-\bm A_M\bm\Delta)^2+\kappa_2\|\bm A_M\bm z\|_2^2\right).
\end{equation*}
	
\end{proof}


With Gaussian randomization, the pivot has a simper representation we now derive.
Recall the variance of the randomization is $c\bm I_d$. Denote $\bm v=\bm A_M\bm C$, $\bm w=\bm\omega'-\bm A_M\bm z_A-\bm A_M\bm\Delta$. Assuming the matrix $\bm A$ is of size $d\times 1$, and $\bm A^T\bm Z=:T$ has variance 1, the numerator of the pivot multiplied by $(2\pi)^{(d+1)/2}c^{d/2}$ becomes:
\begin{align*}
	&\int_{t\geq T} \exp\left(-\frac{t^2}{2}\right)\int_{\bm H_M} \exp\left(-\frac{1}{2c}(\bm\omega'-\bm A_M\bm Ct-\bm z_A'-\bm A_M\bm\Delta)^T(\bm\omega'-\bm A_M\bm Ct-\bm z_A'-\bm A_M\bm\Delta)\right)d\bm\omega' \\
	&=\int_{t\geq T} \int_{\bm H_M}\exp\left(-\frac{1}{2}\left(\frac{1}{\|\bm v\|_2^2}+\frac{1}{c}\right)(\bm vt)^T(\bm vt)+\sqrt{\frac{1}{\|\bm v\|_2^2}+\frac{1}{c}}(\bm vt)^T\frac{1}{c}\sqrt{\frac{c\|\bm v\|_2^2}{\|\bm v\|_2^2+c}}\bm w -\frac{1}{2c^2}\frac{\|\bm v\|_2^2c}{\|\bm v\|_2^2+c}\bm w^T\bm w\right)\\
	&\hspace{2em}\cdot\exp\left(-\frac{1}{2}\bm w^T\bm w\left(\frac{1}{c}-\frac{\|\bm v\|_2^2}{c(\|\bm v\|_2^2+c)}\right)\right)d\bm\omega'\\
	&=\int_{\bm H_M}\int_{t\geq T}\exp\left(-\frac{1}{2}\left(\sqrt{\frac{\|\bm v\|_2^2+c}{\|\bm v\|_2^2c}}\bm vt-\frac{1}{c}\sqrt{\frac{\|\bm v\|_2^2c}{\|\bm v\|_2^2+c}}\bm w\right)^T\left(\sqrt{\frac{\|\bm v\|_2^2+c}{\|\bm v\|_2^2c}}\bm vt-\frac{1}{c}\sqrt{\frac{\|\bm v\|_2^2c}{\|\bm v\|_2^2+c}}\bm w\right)\right)dt  \\
	&\hspace{2em}\cdot\exp\left(-\frac{1}{2}\bm w^T\bm w\frac{1}{\|\bm v\|_2^2+c}\right)d\bm \omega'.
\end{align*}	
	
Denote
\begin{equation*}
	\tilde{\bm v}=\sqrt{\frac{\|\bm v\|_2^2+c}{\|\bm v\|_2^2c}}\bm v,\;\;\; \tilde{\bm w}=\frac{1}{c}\sqrt{\frac{\|\bm v\|_2^2c}{\|\bm v\|_2^2+c}}\bm w.
\end{equation*}
Since
\begin{align*}
	&(\tilde{\bm v}t-\tilde{\bm w})^T(\tilde{\bm v}t-\tilde{\bm w})=\|\tilde{\bm v}\|_2^2\left(t^2-2t\frac{\tilde{\bm v}^T\tilde{\bm w}}{\|\tilde{\bm v}\|_2^2}+\frac{(\tilde{\bm v}^T\tilde{\bm w})^2}{\|\tilde{\bm v}\|_2^4}\right)-\frac{(\tilde{\bm v}^T\tilde{\bm w})2}{\|\tilde{\bm v}\|_2^2}+\|\tilde{\bm w}\|_2^2 \\
	&=\|\tilde{\bm v}\|_2^2\left(t-\frac{\tilde{\bm v}^T\tilde{\bm w}}{\|\tilde{\bm v}\|_2^2}\right)^2-\frac{(\tilde{\bm v}^T\tilde{\bm w})^2}{\|\tilde{\bm v}\|_2^2}+\|\tilde{\bm w}\|_2^2,
\end{align*}
we have the numerator (multiplied by $(2\pi)^{(d+1)/2}c^{d/2}$) to be
	\begin{align*}
		&\int_{\bm H_M}\int_{t\geq T}\exp\left(-\frac{\|\tilde{\bm v}\|_2^2}{2}\left(t-\frac{\tilde{\bm v}^T\tilde{\bm w}}{\|\tilde{\bm v}\|_2^2}\right)^2\right)dt\cdot\exp\left(\frac{(\tilde{\bm v}^T\tilde{\bm w})^2}{2\|\tilde{\bm v}\|_2^2}-\frac{1}{2}\|\tilde{\bm w}\|_2^2\right)\cdot\exp\left(-\frac{1}{2}\bm w^T\bm w\frac{1}{\|\bm v\|_2^2+c}\right)d\bm \omega'\\
		&=\frac{\sqrt{2\pi}}{\|\tilde{\bm v}\|_2}\int_{\bm H_M}\bar{\Phi}\left(\|\tilde{\bm v}\|_2\left(T-\frac{\tilde{\bm v}^T\tilde{\bm w}}{\|\tilde{\bm v}\|_2^2}\right)\right)\cdot\exp\left(\frac{(\tilde{\bm v}^T\tilde{\bm w})^2}{2\|\tilde{\bm v}\|_2^2}-\frac{1}{2}\|\tilde{\bm w}\|_2^2\right)\cdot\exp\left(-\frac{1}{2}\bm w^T\bm w\frac{1}{\|\bm v\|_2^2+c}\right)d\bm\omega'\\
		&=\frac{\sqrt{2\pi c}}{\sqrt{\|\bm v\|_2+c}}\int_{\bm H_M}\bar{\Phi}\left(\sqrt{\frac{\|\bm v\|_2^2+c}{c}}\left(T-\frac{\bm v^T\bm w}{\|\bm v\|_2^2+c}\right)\right)\cdot\exp\left(\frac{(\bm v^T\bm w)^2}{2c(\|\bm v\|_2^2+c)}-\frac{1}{2c}\|\bm w\|_2^2\right)d\bm\omega'\\
		&=\frac{\sqrt{2\pi c}}{\sqrt{\|\bm v\|_2+c}}\int_{\bm H_M}\bar{\Phi}\left(\sqrt{\frac{\|\bm v\|_2^2+c}{c}}\left(T-\frac{\bm v^T\bm w}{\|\bm v\|_2^2+c}\right)\right)\cdot\exp\left(-\frac{1}{2}\bm w^T(\bm v\bm v^T+c\bm I)^{-1}\bm w\right)d\bm\omega' \\
		&=\frac{\sqrt{2\pi c}}{\sqrt{\|\bm v\|_2+c}}\int_{\bm H_M}\bar{\Phi}\left(\sqrt{\frac{\|\bm v\|_2^2+c}{c}}T-\frac{\bm v^T\bm w}{\sqrt{c(\|\bm v\|_2^2+c)}}\right)\cdot\exp\left(-\frac{1}{2}\bm w^T(\bm v\bm v^T+c\bm I)^{-1}\bm w\right)d\bm\omega',
	\end{align*}
where $\bar{\Phi}(\cdot)=1-\Phi_{(0,1)}(\cdot)$ is the survival function of the standard Gaussian.
	
Hence the numerator becomes
\begin{align*}
	&\frac{1}{\sqrt{|2\pi(\bm v\bm v^T+c\bm I)|}}\int_{\bm H_M}\bar{\Phi}\left(\sqrt{\frac{\|\bm v\|_2^2+c}{c}}T-\frac{\bm v^T\bm w}{\sqrt{c(\|\bm v\|_2^2+c)}}\right)\cdot\exp\left(-\frac{1}{2}\bm w^T(\bm v\bm v^T+c\bm I)^{-1}\bm w\right)d\bm\omega'\\
	&=\int_{\bm H_M}\bar{\Phi}\left(\sqrt{\frac{\|\bm v\|_2^2+c}{c}}T-\frac{\bm v^T\bm w}{\sqrt{c(\|\bm v\|_2^2+c)}}\right)\cdot\phi_{(\bm 0,\bm v\bm v^T+\bm\Sigma)}\left(\bm \omega'-\bm A_M\bm z_{\bm A}-\bm A_M\bm\Delta\right)d\bm\omega'.
\end{align*}

The denominator of the pivot is
\begin{equation*}
	\int_{\bm H_M}\phi_{(\bm 0,\bm v\bm v^T+\bm\Sigma)}\left(\bm \omega'-\bm A_M\bm z_{\bm A}-\bm A_M\bm\Delta\right)d\bm\omega'.
\end{equation*}

The pivot becomes
\begin{equation*}
	\frac{\int_{\bm H_M}\bar{\Phi}\left(\sqrt{\frac{\|\bm v\|_2^2+c}{c}}T-\frac{\bm v^T\bm w}{\sqrt{c(\|\bm v\|_2^2+c)}}\right)\cdot\exp\left(-\frac{1}{2}\bm w^T(\bm v\bm v^T+c\bm I)^{-1}\bm w\right)d\bm\omega'}{\int_{\bm H_M}\exp\left(-\frac{1}{2}\bm w^T(\bm v\bm v^T+c\bm I)^{-1}\bm w\right)d\bm\omega'},
\end{equation*}
where $\bm w=\bm\omega'-\bm A_M\bm z_{\bm A}-\bm A_M\bm \Delta$.


\begin{lemma}[Gaussian randomization: smoothness of the pivot] \label{lemma:pivot:smoothness:gaussian:randomization}
Assume \eqref{assumption:gaussian:randomization}.
\begin{equation*}
	\left|\partial_{\bm z}^{\bm \alpha}\mathcal{P}^G(\bm z, \bm A_M\bm\Delta, \bm A_M,\bm H_M)\right|\leq O(1+d(\bm 0, \bm H_M-\bm A_M\bm\Delta)^3+\|\bm A_M\bm z_{\bm A}\|_2^3)
	\end{equation*}
\end{lemma}

\begin{proof}

The derivative of the pivot with respect to $\bm z$ is a linear combination of the terms:
\begin{align*}
	\bm J_1=\frac{\int_{\bm H_M}\left(\sqrt{\frac{\|\bm v\|_2^2+c}{c}}T-\frac{\bm v^T\bm w}{\sqrt{c(\|\bm v\|_2^2+c)}}\right)\phi_{(0,1)}\left(\sqrt{\frac{\|\bm v\|_2^2+c}{c}}T-\frac{\bm v^T\bm w}{\sqrt{c(\|\bm v\|_2^2+c)}}\right)\cdot\exp\left(-\frac{1}{2}\bm w^T(\bm v\bm v^T+c\bm I)^{-1}\bm w\right)d\bm\omega'}{\int_{\bm H_M}\exp\left(-\frac{1}{2}\bm w^T(\bm v\bm v^T+c\bm I)^{-1}\bm w\right)d\bm\omega'},
\end{align*}
\begin{align*}
	\bm J_2=\frac{\int_{\bm H_M}\bar{\Phi}\left(\sqrt{\frac{\|\bm v\|_2^2+c}{c}}T-\frac{\bm v^T\bm w}{\sqrt{c(\|\bm v\|_2^2+c)}}\right)\cdot(\bm v\bm v^T+c\bm I)^{-1}\bm w\cdot\exp\left(-\frac{1}{2}\bm w^T(\bm v\bm v^T+c\bm I)^{-1}\bm w\right)d\bm\omega'}{\int_{\bm H_M}\exp\left(-\frac{1}{2}\bm w^T(\bm v\bm v^T+c\bm I)^{-1}\bm w\right)d\bm\omega'}
\end{align*}
and
\begin{align*}
	\bm J_3 &=\frac{\int_{\bm H_M}\bar{\Phi}\left(\sqrt{\frac{\|\bm v\|_2^2+c}{c}}T-\frac{\bm v^T\bm w}{\sqrt{c(\|\bm v\|_2^2+c)}}\right)\cdot\exp\left(-\frac{1}{2}\bm w^T(\bm v\bm v^T+c\bm I)^{-1}\bm w\right)d\bm\omega'}{\int_{\bm H_M}\exp\left(-\frac{1}{2}\bm w^T(\bm v\bm v^T+c\bm I)^{-1}\bm w\right)d\bm\omega'} \\
	&\hspace{2em}\cdot \frac{\int_{\bm H_M}(\bm v\bm v^T+c\bm I)^{-1}\bm w\exp\left(-\frac{1}{2}\bm w^T(\bm v\bm v^T+c\bm I)^{-1}\bm w\right)d\bm\omega'}{\int_{\bm H_M}\exp\left(-\frac{1}{2}\bm w^T(\bm v\bm v^T+c\bm I)^{-1}\bm w\right)d\bm\omega'}.
\end{align*}

$\bm J_1$ is the expectation of a bounded function hence $\|\bm J_1\|_2$ is bounded. Since both $\|\bm J_2\|_2$ and $\|\bm J_3\|_2$ are upper bounded by 
\begin{equation*}
	 \|(\bm v\bm v^T+c\bm I)^{-1}\|_2\frac{\int_{\bm H_M}\|\bm w\|_2\cdot\exp\left(-\frac{1}{2}\bm w^T(\bm v\bm v^T+c\bm I)^{-1}\bm w\right)d\bm\omega'}{\int_{\bm H_M}\exp\left(-\frac{1}{2}\bm w^T(\bm v\bm v^T+c\bm I)^{-1}\bm w\right)d\bm\omega'},
\end{equation*} 
it suffices to analyze the growth of 
\begin{equation*}
	\frac{\int_{\bm H_M}\|\bm w\|_2\cdot\exp\left(-\frac{1}{2}\bm w^T(\bm v\bm v^T+c\bm I)^{-1}\bm w\right)d\bm\omega'}{\int_{\bm H_M}\exp\left(-\frac{1}{2}\bm w^T(\bm v\bm v^T+c\bm I)^{-1}\bm w\right)d\bm\omega'}
\end{equation*}
 in order to see how fast the derivative of the pivot grows. Similarly, in order to bound the second and the third derivative of the pivot suffices to bound
 \begin{align*}
	&\frac{\int_{\bm H_M}\|\bm w\|_2^{\alpha}\cdot\exp\left(-\frac{1}{2}\bm w^T(\bm v\bm v^T+c\bm I)^{-1}\bm w\right)d\bm\omega'}{\int_{\bm H_M}\exp\left(-\frac{1}{2}\bm w^T(\bm v\bm v^T+c\bm I)^{-1}\bm w\right)d\bm\omega'} \\
	&=\frac{\int_{\bm H_M-\bm A_M\bm z_{\bm A}-\bm A_M\bm\Delta}\|\bm w\|_2^{\alpha}\cdot\exp\left(-\frac{1}{2}\bm w^T(\bm v\bm v^T+c\bm I)^{-1}\bm w\right)d\bm w}{\int_{\bm H_M-\bm A_M\bm z_{\bm A}-\bm A_M\bm\Delta}\exp\left(-\frac{1}{2}\bm w^T(\bm v\bm v^T+c\bm I)^{-1}\bm w\right)d\bm w}\\
	&=\mathbb{E}_{\bm w\sim\mathcal{N}(\bm 0, \bm v\bm v^T+c\bm I)}\left[\|\bm w\|_2^{\alpha}\;\big|\;\bm w\in \bm H_M-\bm A_M\bm z_{\bm A}-\bm A_M\bm\Delta\right]
\end{align*}
for $\alpha=1,2,3$. Using Lemma \ref{lemma:gaussian:moments:sel}, we get 
\begin{equation}
	\left|\partial_{\bm z}^{\bm \alpha}\mathcal{P}^G(\bm z, \bm A_M\bm\Delta, \bm A_M,\bm H_M)\right|\leq O(1+d(\bm 0, \bm H_M-\bm A_M\bm z_{\bm A}-\bm A_M\bm\Delta)^3)
\end{equation}
By the triangle inequality we get the conclusion.

\end{proof}


The following lemma provides a bound on the growth of the moments of the Gaussian random variable after selection. Recall $\bm H_M=\prod_{i=1}^d[b_i,\infty)$ for some vector $\bm b\in\mathbb{R}^d$.

\begin{lemma} \label{lemma:gaussian:moments:sel}
For a given $\bm X\sim\mathcal{N}_d(\bm 0,\bm \Sigma)$ and a vector $\bm\mu\in\mathbb{R}^d$, we have
\begin{equation*}
	\mathbb{E}\left[\|\bm X\|_2^{\alpha} \:|\: \bm X\in\bm H_M-\bm\mu\right]=O(d(\bm 0, \bm H_M-\bm \mu)^{\alpha}+1)
\end{equation*}
for $\alpha\in\mathbb{N}$.
\end{lemma}


\begin{proof}

We want to show $\mathbb{E}\left[\|\bm X\|_2^{\alpha} \;|\; \bm X\in\bm H_M-\bm\mu\right]$ grows at most linearly in $d(\bm 0,\bm H_M-\bm\mu)^{\alpha}$ for $d:=d(\bm 0, \bm H_M-\bm \mu)>0$. The quantity above is  increasing in $d$.

Denoting $\bm Y\sim\mathcal{N}\left(\bm 0,\frac{1}{d^2}\bm\Sigma\right)$ we get the above equals
\begin{equation*}
	d^{\alpha}\cdot\mathbb{E}\left[\|\bm Y\|_2^{\alpha}\;\Big|\;\bm Y\in\frac{\bm H_M-\bm\mu}{d}\right].
\end{equation*}
Since $\bm H_M=\prod_{i=1}^d[b_i,\infty)$, $\frac{\bm H_M-\bm\mu}{d}$ is a rectangle $\prod_{i=1}^d[u_i,\infty)$, $\|\bm u\|_2= 1$. So $\mathbb{E}\left[\|\bm Y\|_2^{\alpha}\;\Big|\;\bm Y\in\frac{\bm H_M-\bm\mu}{d}\right]$ is decreasing in $d$ hence bounded above.

\end{proof}



\subsection{Applying Chatterjee's theorem for Gaussian randomization}

\begin{lemma}[Gaussian randomization: closeness of Gaussian and non-parametric LR] Assuming \eqref{assumption:asymptotically:linear}, \eqref{assumption:LA:weaker}, \eqref{assumption:norm:A_M} and \eqref{assumption:gaussian:randomization} and \eqref{assumption:subG} hold, we have
\begin{equation*}
	\mathbb{E}_{\mathbb{F}_n}\left[\left|\ell_{\mathbb{F}_n}(\bm Z)-\ell_{\Phi_n}(\bm Z)\right|\right]\leq \frac{1}{\sqrt{n}} C(C_M, K_g, K_G,\bm\Sigma),
	\end{equation*}
	where the constant on the RHS depends only on its arguments.
	\label{lemma:lik:diff:gaussian:randomization}
\end{lemma}

\begin{proof}

Since
\begin{equation*}
\begin{aligned}
	&\mathbb{E}_{\mathbb{F}_n}\left[\left|\ell_{\mathbb{F}_n}(\bm Z)-\ell_{\Phi_n}(\bm Z)\right|\right] = \mathbb{E}_{\mathbb{F}_n}\left[Q_1(\bm Z;\bm\Delta)\right]\left|\frac{1}{\mathbb{E}_{\mathbb{F}_n}\left[Q_1(\bm Z;\bm\Delta)\right]}-\frac{1}{\mathbb{E}_{\Phi_n}\left[Q_1(\bm Z;\bm\Delta)\right]}\right| \\
	&= \frac{\left|\mathbb{E}_{\mathbb{F}_n}\left[Q_1(\bm Z;\bm\Delta)\right]-\mathbb{E}_{\Phi_n}\left[Q_1(\bm Z;\bm\Delta)\right]\right|}{\mathbb{E}_{\Phi_n}\left[Q_1(\bm Z;\bm\Delta)\right]} =\left|\mathbb{E}_{\mathbb{F}_n}[\ell_{\Phi_n}(\bm Z)]-\mathbb{E}_{\Phi_n}[\ell_{\Phi_n}(\bm Z)]\right|.
\end{aligned}
\end{equation*}

We use Chatterjee technique here, providing an upper bound on the quantity above using the smoothness of $\ell_{\Phi_n}(\bm z)$. From Lemma \ref{lemma:lik:smoothness:gaussian:randomization}, we have
\begin{equation*}
	\left|\partial_{\bm z}^{\bm\alpha}\ell_{\Phi_n}(\bm z)\right|\leq const \cdot \exp\left(\kappa_1 d(\bm 0,\bm H_M-\bm A_M\bm\Delta)^2+\kappa_2\|\bm A_M\bm z\|_2^2\right),
\end{equation*}
implying
\begin{equation*}
	\lambda_3^{(n)}(\ell_{\Phi_n}) \leq \frac{1}{n^{3/2}} \exp(\kappa_1d(\bm 0,\bm H_M-\bm A_M\bm\Delta)^2).
\end{equation*}
for $\Omega_n(\bm z)=\|\bm A_M\bm z\|_2^2$, $\bm z\in\mathbb{R}^p$, and $\mathcal{W}_n(x)=e^{\kappa_2 x}$, $x\in\mathbb{R}$.
Using Theorem \ref{thm:chatterjee}, we have
\begin{equation*}
\begin{aligned}
		&\left|\mathbb{E}_{\mathbb{F}_n}\left[\ell_{\Phi_n}(\bm Z)\right]-\mathbb{E}_{\Phi_n}\left[\ell_{\Phi_n}(\bm Z)\right]\right| \\
		&\leq \frac{1}{2} \sum_{i=1}^n \lambda_3^{(n)}(\ell_{\Phi_n})\mathbb{E}\left[e^{2\kappa_2 \|\bm A_M\mathcal{M}(\bm W_i)\|_2^2}\right] \left(\mathbb{E}\left[\|\bm g_i\|_1^3\right] +\mathbb{E}\left[\|\bm y_i\|_1^3 \right]\right) \\
		&+\frac{1}{2} \sum_{i=1}^n \lambda_3^{(n)}(\ell_{\Phi_n})\mathbb{E}\left[e^{2\kappa_2\|\bm A_M\mathcal{M}(\widetilde{\bm W}_i-\bm W_i)\|_2^2} \|\bm g_i\|_1^3\right] \\
		&+ \frac{1}{2} \sum_{i=1}^n \lambda_3^{(n)}(\ell_{\Phi_n})\mathbb{E}\left[e^{2\kappa_2\|\bm A_M\mathcal{M}(\widetilde{\bm W}_{i-1}-\bm W_i)\|_2^2}\|\bm y_i\|_1^3 \right]
	\end{aligned}
\end{equation*}

By the sub-Gaussain assumption and the local alternatives the sum in the RHS above converges to zero as $n$ tends to infinity.
\end{proof}

\begin{theorem} \label{thm:selective:clt:gaussian:randomization}
	Assume \eqref{assumption:asymptotically:linear}, \eqref{assumption:LA:weaker}, \eqref{assumption:norm:A_M} and \eqref{assumption:gaussian:randomization} and \eqref{assumption:subG} hold. 
	Given a function $\mathcal{H}:\mathbb{R}\rightarrow\mathbb{R}$ with uniformly bounded derivatives up to the third order (call this bound $K_\mathcal{H}$), we have 
\begin{equation*}
	\left|\mathbb{E}_{\mathbb{F}_n^*}\mathcal{H}(\mathcal{P}^G(\bm Z)) - \mathbb{E}_{\Phi_n^*}\mathcal{H}(\mathcal{P}^G(\bm G))\right|=o_n(1).
\end{equation*}
\end{theorem}

\begin{proof}
	Similar to the proof of the above Lemma.
\end{proof}


\section{Proofs of Section \ref{sec:consistency}} \label{app:consistency}

Since $Q_1(\bm z;\bm\Delta)=\mathbb{G}\left\{\bm H_M-\bm A_M\bm\Delta-\bm A_M\bm z\right\}$, $\bm z\in\mathbb{R}^p,$
we write the likelihood ratio as
\begin{align*}
 \ell_{\mathbb{F}_n}(\bm t)=\frac{Q_1(\bm z;\bm\Delta)}{ \mathbb{E}_{\mathbb{F}_n}\left[Q_1(\bm Z;\bm\Delta)\right]}.
\end{align*}

\begin{proof-of-lemma}{\ref{lemma:upper:bound}}
	Using the change of variables $\bm\omega'=\bm\omega+\bm A_M\bm z+\bm A_M\bm\Delta$ and the Lipschitz assumption on $\tilde{g}$, we have the lower bound
\begin{equation*}
\begin{aligned}
	Q_1(\bm z;\bm\Delta) 
	&=\int\limits_{\bm\omega\in\bm H_M-\bm A_M\bm z-\bm A_M\bm \Delta}\frac{1}{C_g}\exp(-\tilde{g}(\bm\omega))d\bm\omega \\
	&=\frac{1}{C_g}\int\limits_{\bm\omega'\in \bm H_M}\exp\left(-\tilde{g}(\bm\omega'-\bm A_M\bm z-\bm A_M\bm\Delta)\right)d\bm\omega' \\
	&\geq\frac{1}{C_g}\int\limits_{\bm\omega'\in\bm H_M} \exp\left(-\tilde{g}(\bm\omega'-\bm A_M\bm\Delta)-K_g\|\bm A_M\bm z\|_h\right)d\bm\omega' \\
	&=e^{-K_g\|\bm A_M\bm z\|_h}\int\limits_{\bm\omega'\in\bm H_M} \frac{1}{C_g}\exp\left(-\tilde{g}(\bm\omega'-\bm A_M\bm\Delta)\right) d\bm\omega' =Q_1(\bm 0;\bm\Delta)\cdot e^{-K_g\|\bm A_M\bm z\|_h},
\end{aligned}
\end{equation*}
implying 
\begin{equation*}
\mathbb{E}_{\mathbb{F}_n}\left[Q_1(\bm Z;\bm\Delta)\right]
\geq Q_1(\bm 0;\bm\Delta)\cdot\mathbb{E}_{\mathbb{F}_n}\left[e^{-K_g\|\bm A_M\bm Z\|_h} \right].
\end{equation*}

Similarly, we have an upper bound
\begin{align*}
	Q_1(\bm z;\bm\Delta) 
	&\leq \frac{1}{C_g}\int\limits_{\bm\omega'\in\bm H_M} \exp\left(-\tilde{g}(\bm\omega'-\bm A_M\bm\Delta)+K_g\|\bm A_M\bm z\|_h\right)d\bm\omega' \\
	&=e^{K_g\|\bm A_M\bm z\|_h}\int\limits_{\bm\omega'\in\bm H_M}\frac{1}{C_g}\exp\left(-\tilde{g}(\bm\omega'-\bm A_M\bm\Delta)\right) d\bm\omega' =Q_1(\bm 0;\bm\Delta)\cdot e^{K_g\|\bm A_M\bm z\|_h}.
\end{align*}
\end{proof-of-lemma}


\begin{proof-of-lemma}{\ref{lemma:consistency}}
	Take any $\delta>0$. Using the upper bound on the selective likelihood ratio bound from Lemma \ref{lemma:upper:bound}, we get the following inequality
\begin{align*}
	\mathbb{F}_{n}^*\left\{f_n(\bm D)>\delta\right\}
	&=\mathbb{F}_n\left\{f_n(\bm D)>\delta\:\big|\:\textnormal{selection}(\bm D,\bb\omega) \right\}\\
	&=\int \mathbb{I}_{\left\{f_n(\bm t)>\delta\right\}}d\mathbb{F}_n^*(\bm t)=\int\mathbb{I}_{\{f_n(\bm t)>\delta\}}\ell_{\mathbb{F}_n}(\bm t)d\mathbb{F}_{n}(\bm t)\\
	&\leq \frac{1}{\mathbb{E}_{\mathbb{F}_n}\left[e^{-K_g\|\bm A_M\bm Z\|_h}\right]}\int \mathbb{I}_{\{f_n(\bm t)>\delta\}}e^{K_g\|\bm A_M\bm z\|_h}d\mathbb{F}_{n}(\bm t).
\end{align*}
Using \eqref{assumption:clt} and \eqref{assumption:norm:A_M} assumption we have $\|\bm A_M\bm z\|_h=O_{\mathbb{F}_n}(1)$ uniformly over $\mathbb{F}_n\in\mathcal{F}_n$, hence for every $\delta_1>0$, there exists $C(\delta_1)$ such that
\begin{equation*}
	\mathbb{F}_n\{\|\bm A_M\bm Z\|_g>C(\delta_1)\}<\delta_1
\end{equation*}
for $n$ sufficiently large.
Since
\begin{align*}
&\int \mathbb{I}_{\{f_n(\bm t)>\delta\}}e^{K_g\|\bm A_M\bm z\|_h}d\mathbb{F}_{n}(\bm t) \\
&=\int\mathbb{I}_{\{f_n(\bm t)>\delta\}} e^{K_g\|\bm A_M\bm z\|_h}\mathbb{I}_{\{\left\|\bm A_M\bm z\|_h\leq C(\delta_1)\right\}}d\mathbb{F}_{n}(\bm t)+\int \mathbb{I}_{\{f_n(\bm t)>\delta\}}e^{K_g\|\bm A_M\bm z\|_h}\mathbb{I}_{\left\{\|\bm A_M\bm z\|_h> C(\delta_1)\right\}}d\mathbb{F}_{n}(\bm t) \\
&\leq e^{K_g\cdot C(\delta_1)} \mathbb{F}_n\{f_n(\bm D)>\delta\} 
	+\mathbb{E}_{\mathbb{F}_n}\left[e^{K_g\|\bm A_M\bm Z\|_h}\mathbb{I}_{\left\{\|\bm A_M\bm Z\|_h>C(\delta_1)\right\}}\right] \\
&\leq e^{K_g\cdot C(\delta_1)} \mathbb{F}_n\{f_n(\bm D)>\delta\} + \left(\mathbb{E}_{\mathbb{F}_n}\left[e^{2K_g\|\bm A_M\bm Z\|_h}\right]\right)^{1/2}\left(\mathbb{P}_{\mathbb{F}_n}\left\{\|\bm A_M\bm Z\|_h> C(\delta_1)\right\}\right)^{1/2} \\
&\leq e^{K_g\cdot C(\delta_1)} \mathbb{F}_n\{f_n(\bm D)>\delta\} + \left(\mathbb{E}_{\mathbb{F}_n}\left[e^{2K_g\|\bm A_M\bm Z\|_h}\right]\right)^{1/2}\delta_1^{1/2}.
\end{align*}

Using $\underset{n\rightarrow\infty}{\lim}\:\underset{\mathbb{F}_n\in\mathcal{F}_n}{\sup}\mathbb{F}_{n}\left\{f_n(\bm D)>\delta\right\}=0$ (by the assumption of the lemma) and $\underset{\mathbb{F}_n\in\mathcal{F}_n}{\sup}\mathbb{E}_{\mathbb{F}_n}\left[e^{2K_g\|\bm A_M\bm Z\|_h}\right]=O(1)$, we get
\begin{equation*}
	\underset{n\rightarrow\infty}{\lim}\:\underset{\mathbb{F}_n^*\in\mathcal{F}_n^*}{\sup}\mathbb{F}_n^*\left\{f_n(\bm D)>\delta\right\}=0.
\end{equation*}

\end{proof-of-lemma}


\section{Proofs of Section \ref{sec:bootstrap}} \label{app:proofs:boot}

The goal here is to provide all the details showing that the $\mathcal{P}^B$ is asymptotically $\textnormal{Unif}[0,1]$ uniformly across $\mathbb{F}_n\in\mathcal{F}_n$. In order to show that, let us introduce some important notation. For $\bb\tau\in\mathbb{R}^a$ and $\bm{u}_{\bb A}\in\mathbb{R}^{p}$, define 
\begin{equation*}
\begin{aligned}
	Q_2\left(\bb\tau;\bm z_{\bb A},\bm\Delta\right)
	&:=\mathbb{G}\left\{\bb\omega: \bb\omega+\bm A_M \bm z_{\bb A}+\bm A_M\bm C\bb\tau +\bm A_M\bm\Delta\in\bm H_M\right\} \\
	&=\mathbb{G}\left\{\bm H_M-\bm A_M\bm z_{\bb A}-\bm A_M\bm C\bb\tau-\bm A_M\bm\Delta\right\} \\
	&=\mathbb{G}\left\{\bm H_M-\sqrt{n}\bm A_M\bm t_{\bb A}-\bm A_M(\bm C\bb\tau+\sqrt{n}\bm C\bb A^T\bm{\mu}) \right\},
\end{aligned}
\end{equation*}
where as before $\sqrt{n}\bm t_{\bb A}:=\bm z_{\bb A}+\bm\Delta-\sqrt{n}\bm C\bb A^T\bm\mu$, and
\begin{align*}
	&Q_3(\bm z_{\bb A};\bm\Delta)
	:=\int_{\mathbb{R}^a} Q_2\left(\bb\tau;\bm z_{\bb A},\bm\Delta\right)\frac{1}{(2\pi)^{-a/2}|\bb\Sigma_{\bb A}|^{-1/2}} e^{-\frac{1}{2}\bb\tau^T\bb\Sigma_{\bb A}^{-1}\bb\tau}d\bb\tau\\
	&=\int_{\mathbb{R}^a} \mathbb{G}\left\{\bm H_M-\sqrt{n}\bm A_M\bm t_{\bb A}-\bm A_M\left(\bm C\bb\tau+\sqrt{n}\bm C\bb A^T\bm\mu\right) \right\}\frac{e^{-\frac{1}{2}\bb \tau^T \bb\Sigma_{\bb A}^{-1}\bb\tau}}{(2\pi)^{-a/2}|\bb\Sigma_{\bb A}|^{-1/2}} d\bb\tau\\
	&=\int_{\mathbb{R}^a} \mathbb{G}\left\{\bm H_M-\sqrt{n}\bm A_M\bm t_{\bb A}-\sqrt{n}\bm A_M\bm C\bb s \right\}\frac{\sqrt{n}}{(2\pi)^{-a/2}|\bb\Sigma_{\bb A}|^{-1/2}}e^{-\frac{n}{2}(\bb s-\bb A^T\bm\mu)^T\bb\Sigma_{\bb A}^{-1}(\bb s-\bb A^T\bb\mu)}d\bb s,
\end{align*}
where $\bb\Sigma_{\bb A} = \bb A^T\bb\Sigma\bb A$ in the last equality we did the change of variables $\bb\tau+\sqrt{n}\bb A^T\bm\mu=\sqrt{n}\bb s$.

The following lemma is an important step in proving asymptotics of $\mathcal{P}^B$. It will be used to show that the denominator of the test statistic $\mathcal{P}^B$ is bounded below in probability uniformly over $\mathbb{F}_n\in\mathcal{F}_n$.

\begin{lemma}[Lower bound] \label{lemma:lower:bound} Assuming  \eqref{assumption:clt}, \eqref{assumption:local:alt}, \eqref{assumption:g:lip} and \eqref{assumption:norm:A_M}, we have that for every $\delta>0$, there exists $C>0$ such that
\begin{equation*}
	\underset{\mathbb{F}_n\in\mathcal{F}_n}{\sup} \mathbb{F}_n\left\{ Q_2(\bm Z_{\bb A};\bm\Delta)<C \right\}< \delta
\end{equation*}
for $n$ sufficiently large.
\end{lemma}



\begin{proof} 

Using the Lipschitz assumption on the randomization, we have
	\begin{align*}
		Q_2(\bb\tau;\bm z_{\bb A},\bm\Delta)
		&= \frac{1}{C_g}\int_{\bm H_M-\bm A_M\bm\Delta}\exp\left(-\tilde{g}(\bm \omega'-\bm A_M\bm z_{\bb A} -\bm A_M\bm C\bb \tau)\right)d\bm\omega'\\
		&=\frac{1}{C_g} \int_{\bm H_M-\bm A_M\bm\Delta}\exp\left(-\tilde{g}(\bm\omega'-\bm A_M\bm\Delta)-K_g\|\bm A_M\bm z_{\bb A}\|_h-K_g\|\bm A_M\bm C\bb\tau\|_h\right)d\bm\omega'	\\
		&= \exp\left(-K_g\|\bm A_M\bm z_{\bm A}\|_h\right)\exp\left(-K_g\|\bm A_M\bm C\bm\tau\|_h\right) \int_{\bm H_M-\bm A_M\bm\Delta}\exp(-\tilde{g}(\bm\omega'))d\bm\omega'
	\end{align*}
	for all $\bb\tau\in\mathbb{R}^a$ and $\bm z_{\bb A}\in\mathbb{R}^p$. Since $d_h\left(\bb 0,\bm H_M-\bm A_M\bm\Delta\right)\leq B$ (local alternatives), the last integral is O(1).
	\begin{align*}
		Q_2\left(\bb\tau;\bm z_{\bb A},\bm\Delta\right)
		&\geq e^{-K_g C_M\|\bm z_{\bm A}\|_2}e^{-K_g C_M\|\bm C\bm\tau\|_2} \int_{\bm H_M-\bm A_M\bm\Delta}\exp(-\tilde{g}(\bm\omega'))d\bm\omega',
	\end{align*}
	for all $\bb\tau\in\mathbb{R}^a$ and $\bb z_{\bb A}\in\mathbb{R}^p$, where the second inequality follows from (NA) assumption for $n$ large enough. Defining a function of $x\in\mathbb{R}$
	\begin{equation*}
		\int_{\mathbb{R}^a} e^{-\|\bb C\bb\tau\|_2x} \frac{1}{(2\pi)^{-a/2}|\bb\Sigma_{\bb A}|^{-1/2}} e^{-\frac{1}{2}\bb\tau\bb\Sigma_{\bb A}^{-1}\bb\tau}d\bb\tau=:f_{\bb\Sigma}(x), \footnote{Since under $\bb\tau\sim\mathcal{N}_a(\bb 0, \bb\Sigma_{\bb A})$, $\bb C\bb\tau \sim\mathcal{N}_p(\bb 0,\bb \Sigma)$, hence the RHS does not depend on $\bb A$.}
	\end{equation*}
	we have
	\begin{align*}	
		Q_3\left(\bm z_{\bb A};\bm\Delta\right)\geq e^{-K_g C_M\|\bm z_{\bm A}\|_2}f_{\bm\Sigma}(K_g C_M) \int_{\bm H_M-\bm A_M\bm\Delta}\exp(-\tilde{g}(\bm\omega'))d\bm\omega'
	\end{align*}
		for all $\bm z_{\bb A}\in\mathbb{R}^{d}$.
		
		For a given $\delta>0$, it suffices to find a constant $C_1>0$ such that 
		\begin{equation*} \label{eq:lb:lemma:sufficient}
			 \underset{\mathbb{F}_n\in\mathcal{F}}{\sup}\mathbb{F}_n\left\{\left\|\bm Z_{\bb A}\right\|_2>C_1\right\}<\delta
		\end{equation*}
		for $n$ large enough.
		By (CLT) assumption, $\bm Z_{\bb A}$ converges uniformly in distribution to a random variable $\bm{G}\sim\mathcal{N}_p\left(\bm 0,(\bb I_p-\bb C\bb A^T)\bb\Sigma(\bb I_p-\bb C\bb A^T)^T\right)$. Using the uniform continuous mapping theorem we get that $\|\bm Z_{\bb A}\|_2$ converges uniformly in distribution to $\|\bm G\|_2$, i.e.
		\begin{equation*}
			\underset{n\rightarrow\infty}{\lim}\:\underset{\mathbb{F}_n\in\mathcal{F}_n}{\sup}\underset{t\in\mathbb{R}}{\sup}\left|\mathbb{F}_n\left\{\|\bm Z_{\bb A}\|_2 \leq t\right\}-\mathbb{P}_{\bm G}\left\{ \|\bm G\|_2\leq t \right\}\right|=0.
		\end{equation*}
		This implies (\ref{eq:lb:lemma:sufficient}) holds.
	
 \end{proof}


The following lemma justifies using $\widehat{\bm D}_{\bb A}$ instead of $\bm D_{\bb A}$ in constructing the bootstrap pivot

\begin{lemma} \label{lemma:t_hat}
Assuming \eqref{assumption:clt}, \eqref{assumption:norm:A_M} and \eqref{assumption:variance} hold, for any $\delta>0$ the following holds 
	\begin{equation*}
		\underset{n\rightarrow\infty}{\lim}\:\underset{\mathbb{F}_n\in\mathcal{F}_n}{\sup}\mathbb{F}_n\left\{\left\|\sqrt{n}\bm{A}_M(\widehat{\bm D}_{\bb A}+\widehat{\bm C}\bb A^T\bm{\mu})-\sqrt{n}\bm A_M\left(\bm D_{\bb A}+\bm C\bb A^T\bm\mu\right)\right\|_2>\delta\right\}=0.
	\end{equation*}
\end{lemma}

\begin{proof}
Since
\begin{equation*}
	\sqrt{n}\bm A_M(\widehat{\bm D}_{\bb A}-\bm D_{\bb A})+\sqrt{n}\bm A_M(\widehat{\bm C}-\bm C)\bb A^T\bm\mu =\bm A_M(\bm C-\widehat{\bm C})\left(\sqrt{n}\left(\bb A^T\bm D-\bb A^T\bm\mu\right)\right),
\end{equation*}
we have
\begin{equation}\label{eq:t_hat_bound}
\begin{aligned}
	&\left\|\sqrt{n}\bm A_M(\widehat{\bm D}_{\bb A}+\widehat{\bm C}\bb A^T\bm\mu)-\sqrt{n}\bm A_M(\bm D_{\bb A}+\bm C\bb A^T\bm\mu)\right\|_2 \\
	&\leq \|\bm A_M\|_2 \|\bm C-\widehat{\bm C}\|_2\left\|\sqrt{n}\left(\bb A^T(\bm D-\bm\mu)\right)\right\|_2 \\
	&\leq a_M \|\bm{C}-\widehat{\bm C}\|_2\left\|\sqrt{n}(\bb A^T\bm D-\bb A^T\bm\mu)\right\|_2, 
\end{aligned}
\end{equation}
where the second inequality holds for $n$ large enough  by \eqref{assumption:norm:A_M} assumption.
Since $\|\widehat{\bm C}-\bm C\|_2=o_{\mathbb{F}_n}(1)$ uniformly across $\mathbb{F}_n$ by assumption \eqref{assumption:variance} and $\left\|\sqrt{n}(\bb A^T\bm D-\bb A^T\bm\mu)\right\|_2$ converges in distribution by (CLT) assumption to $\|\bb G\|_2$, where $\bb G\sim\mathcal{N}_a\left(\bb 0, \bb\Sigma_{\bb A}\right)$ uniformly across $\mathbb{F}_n\in\mathcal{F}_n$, we have that the term in (\ref{eq:t_hat_bound}) converges to zero in probability uniformly across $\mathbb{F}_n\in\mathcal{F}$, hence the conclusion follows.
\end{proof}

\begin{proof-of-theorem}{\ref{thm:boot}}

Since for any $t\in\mathbb{R}$
\begin{equation*}
\begin{aligned} 
		&\mathcal{P}^B\left(t,\bm D,\bb A^T\bm\mu\right)\\
		&=\resizebox{0.95\textwidth}{!} { $ \frac{\mathbb{E}_{\bb\omega\sim\mathbb{G}}\left[\widehat{\mathbb{F}}_n \left\{\sqrt{n}\left\|\bb A^T(\bm D^*-\bm D)\right\|_2 \geq t,\bm\omega\in \bm H_M-\sqrt{n}\bm A_M\widehat{\bb D}_{\bb A}-\bm A_M\widehat{\bm C}\bb A^T\bm Z^*-\sqrt{n}\bm A_M\widehat{\bm C}\bb A^T\bm\mu\:\Big|\: \bm D,\bm\omega\right\} \right] }{ \mathbb{E}_{\bb\omega\sim\mathbb{G}}\left[\widehat{\mathbb{F}}_n\left\{\bb\omega\in\bm H_M-\sqrt{n}\bm A_M\widehat{\bb D}_{\bb A}-\bm A_M\widehat{\bm C}\bb A^T\bm Z^*-\sqrt{n}\bm A_M\widehat{\bm C}\bb A^T\bm\mu\:\Big|\:\bm D,\bb\omega\right\}\right]}, $}
\end{aligned}
\end{equation*}
by the bootstrap consistency assumption \eqref{assumption:boot:consistency}, we can write	
\begin{equation*}
\begin{aligned}
	    &\mathcal{P}^B\left(t,\bm D,\bb A^T\bm \mu\right)\\	
		&= \resizebox{0.95\textwidth}{!} { $ \frac{\mathbb{E}_{\bb\omega\sim\mathbb{G}}\left[\mathbb{P}_{\bb G} \left\{\|\widehat{\bb \Sigma}_{\bb A}^{1/2}\bb G\|_2\geq t,\bb\omega\in \bm H_M-\sqrt{n}\bm A_M\widehat{\bm D}_{\bb A}-\bm A_M\widehat{\bm C}\widehat{\bb \Sigma}_{\bb A}^{1/2}\bb G-\sqrt{n}\bm A_M\widehat{\bm C}\bb A^T\bm\mu \:\Big|\: \bm D,\bb\omega\right\} \right]+E_n}{\mathbb{E}_{\bb\omega\sim\mathbb{G}}\left[\mathbb{P}_{\bb G}\left\{\bb\omega\in \bm H_M-\sqrt{n}\bm A_M\widehat{\bm D}_{\bb A}-\bm A_M\widehat{\bm C }\widehat{\bb\Sigma}_{\bb A}^{1/2}\bb G-\sqrt{n}\bm A_M\widehat{\bm C}\bb A^T\bm\mu\:\Big|\:\bm D,\bb\omega\right\}\right]+E'_n}, $}
\end{aligned}
\end{equation*}
where conditional on $\bm D$ we have $\bb G\sim\mathcal{N}_a(\bb 0,\bb I_a)$. Furthermore, for the random variables $E_n$ and $E_n'$ we have $E_n, E'_n=o_{\mathbb{F}_n}(1)$ uniformly over $t\in\mathbb{R}$ and over $\mathbb{F}_n\in\mathcal{F}_n$. 
%
%
By the uniform consistency of the variance, assumption \eqref{assumption:variance}, and Lemma \ref{lemma:t_hat}, the following holds for any $t\in\mathbb{R}$
\begin{equation*}
\begin{aligned}
	    &\mathcal{P}^B\left(t,\bm D,\bb A^T\bm\mu\right)\\	
		&=\resizebox{0.95\textwidth}{!} { $ \frac{\mathbb{E}_{\bb\omega\sim\mathbb{G}}\left[\mathbb{P}_{\bb G} \left\{\|\bb\Sigma_{\bb A}^{1/2} \bb G\|_2\geq t,\bb\omega\in \bm H_M-\sqrt{n}\bm A_M\bm D_{\bb A}-\bm A_M\bm C\bb\Sigma_{\bb A}^{1/2}\bb G-\sqrt{n}\bm A_M\bm C\bb A^T\bm\mu \:\big|\: \bm D,\bb\omega\right\} \right]+e_n}{\mathbb{E}_{\bb\omega\sim\mathbb{G}} \left[\mathbb{P}_{\bb G}\left\{\bb\omega\in\bm H_M-\sqrt{n}\bm A_M\bm D_{\bb A}-\bm A_M\bm C\bb\Sigma_{\bb A}^{1/2}\bb G-\sqrt{n}\bm A_M\bm C\bb A^T\bm\mu\:\big|\: \bm D,\bb\omega\right\}\right]+e'_n}, $}
\end{aligned}
\end{equation*}
	where $e_n, e_n'=o_{\mathbb{F}_n}(1)$ uniformly over $\mathbb{F}_n\in\mathcal{F}$ and over $t\in\mathbb{R}$. 
	By the law of iterated expectation and using the definitions of $Q_2$ and $Q_3$, we further have
	\begin{equation*}
	\begin{aligned}
	    &\mathcal{P}^B(t,\bm D,\bb A^T\bm\mu)\\	
		&=\frac{\mathbb{E}_{\bb G} \left[\mathbb{I}_{\left\{\|\bb\Sigma_{\bb A}^{1/2}\bb G\|_2\geq t\right\}}\mathbb{G}\left\{\bm H_M-\sqrt{n}\bm A_M\bm D_{\bb A}-\bm A_M\bm C\bb \Sigma_{\bb A}^{1/2}\bb G-\sqrt{n}\bm A_M\bm C\bb A^T\bm\mu\:\big|\: \bm D,\bb G\right\}\right]+e_n}{\mathbb{E}_{\bb G}\left[\mathbb{G}\left\{\bm H_M-\sqrt{n}\bm A_M\bm D_{\bb A}-\bm A_M\bm C\bb\Sigma_{\bb A}^{1/2}\bb G-\sqrt{n}\bm A_M\bm C\bb A^T\bm\mu\:\big|\:\bm D,\bb G\right\}\right]+e_n'}\\
		&=\frac{\mathbb{E}_{\bb G}\left[\mathbb{I}_{\left\{\|\bb\Sigma_{\bb A}^{1/2}\bb G\|_2\geq t\right\}}Q_2(\bb\Sigma_{\bb A}^{1/2}\bb G;\bm Z_{\bb A},\bm{\Delta})\:\big|\:\bm D\right]+e_n}{\mathbb{E}_{\bb Z}\left.\left[Q_2(\bb\Sigma_{\bb A}^{1/2} \bb G; \bm Z_{\bb A},\bm\Delta)\right|\bm D\right]+e_n'}\\
		&=\frac{\mathbb{E}_{\bb G}\left[\mathbb{I}_{\left\{\|\bb \Sigma_{\bb A}^{1/2}\bb G\|_2\geq t\right\}}Q_2(\bb\Sigma_{\bb A}^{1/2}\bb G;\bm Z_{\bb A},\bm\Delta)\:\big|\:\bm D\right]+e_n}{Q_3\left(\bm Z_{\bb A};\bm\Delta\right)+e_n'} \\
		&=\frac{\frac{\mathbb{E}_{\bb G}\left[\mathbb{I}_{\left\{\|\bb\Sigma_{\bb A}^{1/2}\bb G\|_2\geq t\right\}}Q_2(\bb\Sigma_{\bb A}^{1/2} \bb G;\bm Z_{\bb A},\bm\Delta)\:\big|\:\bm D\right]}{Q_3\left(\bm Z_{\bb A};\bm\Delta\right)}+\frac{e_n}{Q_3\left(\bm Z_{\bb A};\bm\Delta\right)}}{1+\frac{e_n'}{Q_3\left(\bm Z_{\bb A};\bm\Delta\right)}},
	\end{aligned}
	\end{equation*}
	where in the last equality we divided both numerator and denominator by $Q_3\left(\bm Z_{\bb A};\bm\Delta\right)$.
	
	For any $\delta>0$ and $C>0$ we have
	\begin{align*}
		\mathbb{F}_n\left\{\left|\frac{e_n}{Q_3\left(\bm Z_{\bb A};\bm\Delta\right)}\right|\leq \delta \right\}
		&\geq \mathbb{F}_n\left\{|e_n|\leq \delta C,\;\; Q_3\left(\bm Z_{\bb A};\bm\Delta\right)\geq C\right\} \\
		&\geq 1-\mathbb{F}_n\{|e_n|> \delta C\}-\mathbb{F}_n\left\{Q_3\left(\bm Z_{\bb A};\bm\Delta\right)< C\right\}.
	\end{align*}
	Since $e_n=o_{\mathbb{F}_n}(1)$ uniformly across $\mathbb{F}_n\in\mathcal{F}$, we have that $\underset{\mathbb{F}_n}{\sup}\:\mathbb{F}_n\{|e_n|>\delta C\}\rightarrow 0$ as $n\rightarrow\infty$. 
	This is the part where it becomes crucial to have the denominator $Q_3(\bm Z_{\bb A};\bm\Delta)$ lower-bounded in probability. By the Lower bound lemma (Lemma \ref{lemma:lower:bound}), there exists $C$ to make $\underset{\mathbb{F}_n\in\mathcal{F}_n}{\sup}\:\mathbb{F}_n\left\{Q_3\left(\bm Z_{\bb A};\bm\Delta\right)< C\right\}$ arbitrarily small for large enough $n$. This shows
	\begin{equation*}
		\underset{n\rightarrow\infty}{\lim}\underset{\mathbb{F}_n\in\mathcal{F}_n}{\sup}\:\mathbb{F}_n\left\{\left|\frac{e_n}{Q_3\left(\bm Z_{\bb A};\bm\Delta\right)}\right|> \delta \right\}=0.
	\end{equation*}
	Similarly, we have
	\begin{equation*}
		\underset{n\rightarrow\infty}{\lim}\underset{\mathbb{F}_n\in\mathcal{F}_n}{\sup}\:\mathbb{F}_n\left\{\left|\frac{e_n'}{Q_3\left(\bm Z_{\bb A};\bm\Delta\right)}\right|> \delta \right\}=0.
	\end{equation*}
	Thus, we have
	\begin{align*}	
	    &\mathcal{P}^B\left(t,\bm D,\bb A^T\bm\mu\right)
		=\frac{\mathbb{E}_{\bb G}\left[\mathbb{I}_{\left\{\|\bb\Sigma_{\bb A}^{1/2}\bb G\|_2\geq t\right\}}Q_2(\bb \Sigma_{\bb A}^{1/2} \bb G;\bm Z_{\bb A},\bm\Delta)\:\Big|\:\bm D\right]}{Q_3\left(\bm Z_{\bb A};\bm\Delta\right)}+g_n,
	\end{align*}
	where $g_n=o_{\mathbb{F}_n}(1)$ uniformly across $\mathbb{F}_n\in\mathcal{F}$ and over $t\in\mathbb{R}$.  By Lemma \ref{lemma:consistency} in Section \ref{sec:consistency} we also have $g_n=o_{\mathbb{F}_n^*}(1)$, which allows us to make conditional statements.
	Since
	\begin{align*}
		&\frac{\mathbb{E}_{\bb G}\left[\mathbb{I}_{\left\{\|\bb\Sigma_{\bb A}^{1/2}\bb G\|_2\geq t\right\}}Q_2(\bb \Sigma_{\bb A}^{1/2}\bb G;\bm Z_{\bb A},\bm\Delta)\:\Big|\:\bm D\right]}{Q_3\left(\bm Z_{\bb A};\bm\Delta\right)}\\
		&=\resizebox{0.95\textwidth}{!} { $\frac{\int_{\|\bb s-\bb A^T\bb\mu\|_2 \geq t} \mathbb{G}\left\{\bm H_M-\sqrt{n}\bm A_M\bm D_{\bb A}-\sqrt{n}\bm A_M\bm C\bb s \right\}\exp\left(-\frac{n}{2}(\bb s-\bb A^T\bm\mu)^T\bb\Sigma_{\bb A}^{-1}(\bb s-\bb A^T\bm{\mu})\right)d\bb s  }{ \int_{\mathbb{R}^a} \mathbb{G}\left\{\bm{H}_M-\sqrt{n}\bm A_M\bm D_{\bb A}-\sqrt{n}\bm A_M\bm C\bb s \right\} \exp\left(-\frac{n}{2}(\bb s-\bb A^T\bm\mu)^T\bb\Sigma_{\bb A}^{-1}(\bb s-\bb A^T\bm\mu)\right) d\bb s }  $}\\
		&=\mathcal{P}\left(t,\bb D_{\bb A},\bb A^T\bb \mu \right),
	\end{align*}
	the survival function of $\|\sqrt{n}\bb A^T(\bb D-\bb \mu)\|_2$ when the data is normally distributed,
	we get that the conclusion holds by the selective CLT assumption.

\end{proof-of-theorem}

\section{Model selection and the asymptotics for the LASSO} \label{sec:lasso:additional}

We characterize the model $(E, \bm s_E)$ chosen by the randomized Lasso objective in Section \ref{sec:lasso} and show that asymptotically the selection region is affine in terms of $\bm D$. We also show that $\bm D$ is asymptotically linear test statistic. 

The notation is as in Section \ref{sec:lasso}. Recall that our data comes from $(\bm X,\bm y)\sim\mathbb{F}_n^n$ with $\mathbb{F}_n\in\mathcal{F}_n$. We assume the following about the $\mathcal{F}_n$:
\begin{enumerate}[leftmargin=*, label=(\alph*)]
	\item \label{item:X:moments} Given $(\bm X_1,y_1)\sim\mathbb{F}_n$, we assume that $\|\bm X_1\|_{\infty}$ has uniformly bounded third moment across $\mathbb{F}_n\in\mathcal{F}_n$, i.e.~$\underset{\mathbb{F}_n\in\mathcal{F}_n}{\sup}\mathbb{E}_{\mathbb{F}_n}[n^{3/2}\|\bm X_1\|_{\infty}^3]<\infty,$ where recall $X_1$ are scaled with $1/\sqrt{n}$.
	\item \label{item:residuals:moments} Given $(\bm X_1, y_1)\sim\mathbb{F}_n$ and a fixed active set $E$, we assume that the residuals $\epsilon_1=y_1-\bm X_{1,E}^T\bm\beta_E^*$, where $\bm\beta_E^*$ are the population coefficients corresponding to set $E$, have bounded third moments, i.e.~$\underset{\mathbb{F}_n\in\mathcal{F}_n}{\sup}\mathbb{E}_{\mathbb{F}_n}[|\epsilon_1|^3]<\infty.$ 
\end{enumerate}
These assumption allow us to get uniform CLT results across $\mathbb{F}_n\in\mathcal{F}_n$ for the predictors and for the residuals. They can be weakened but for simplicity we keep them as above. Note that the assumptions are pre-selection, treating $E$ as fixed.

\begin{lemma}[Asymptotically affine LASSO selection event] \label{lemma:lasso:asymptotically:affine}
Assumming \ref{item:X:moments} holds,
	the selection event of randomized LASSO is asymptotically affine in $\bm D$.	
\end{lemma}

\begin{proof}
 The KKT conditions for the randomized LASSO are
\begin{align}
	\bm{X}^T_E(\bb y-\bm{X}_E\hat{\bm{\beta}}_E) &=\lambda\bm s_E-\bb\omega_E+\epsilon\hat{\bm{\beta}}_E \nonumber\\
    \bm{X}^T_{-E}(\bb y-\bm X_E\hat{\bm{\beta}}_E) &=\lambda\bm u_{-E}-\bb\omega_{-E}  \nonumber \\
     \textnormal{diag}(\bm s_E)\hat{\bm{\beta}}_E\geq 0, & \;\; \|\bm u_{-E}\|_{\infty}\leq 1, \label{eq:lasso:kkt:inequalities}
\end{align}
where $\bm u_{-E}$ is the sub-gradient of $\frac{\partial|\beta_j|}{\partial{\beta}_j}\big|_{\bb\beta=\widehat{\bb \beta}}$ for $j\not\in E$ (inactive variables) and $\textnormal{diag}(\bm{s}_E)$ is an $|E|\times |E|$ matrix having the entries of $\bm s_E$ on the diagonal and zeros elsewhere. The KKT conditions above become
\begin{align*}
	\hat{\bm\beta}_E &= \bar{\bm{\beta}}_E-(\bm{X}_E^T\bm{X}_E)^{-1}\left(\lambda\bm s_E-\bb\omega_E\right)+\epsilon(\bm X_E^T\bm X_E)^{-1}\hat{\bm\beta}_E \\
	\lambda\bm u_{-E} &=\bm X_{-E}^T\left(\bb y-\bm X_E\bar{\bm\beta}_E\right)+\bm X_{-E}^T\bm X_E(\bm X_E^T\bm X_E)^{-1}\left(\lambda\bm s_E-\bb\omega_E\right) \\
	&\qquad+\epsilon\bm X_{-E}^T\bm X_E(\bm X_E^T\bm X_E)^{-1}\hat{\bm\beta}_E+\bb\omega_{-E}.
\end{align*}

By the strong law of large numbers, pre-selection we have
\begin{equation*}
\begin{aligned}
	\bm X_E^T\bm X_E = \sum_{i=1}^n \bm x_{E,i}\bm x_{E,i}^T&\overset{\mathbb{F}_n}{\rightarrow} \mathbb{E}_{\mathbb{F}_n}\left[n\cdot \bm x_{E,1}\bm x_{E,1}^T\right]=:\bm E_1 \in\mathbb{R}^{|E|\times |E|} \\
	\bm X_{-E}^T\bm X_E = \sum_{i=1}^n\bm x_{-E,i} \bm x_{E,i}^T&\overset{\mathbb{F}_n}{\rightarrow} \mathbb{E}_{\mathbb{F}_n}\left[n\cdot\bm x_{-E,1}\bm x_{E,1}^T\right]=:\bm E_2 \in\mathbb{R}^{(p-|E|) \times |E| },
\end{aligned}
\end{equation*}
where the deviations from the mean are $O_{\mathbb{F}_n}\left(\frac{1}{\sqrt{n}}\right)$ (under CLT assumptions), we write the inequalities from (\ref{eq:lasso:kkt:inequalities}) using 
\begin{equation*}
\begin{aligned}
	&\textnormal{diag}(\bm s_E)\bar{\bm{\beta}}_E-\textnormal{diag}(\bm s_E) \bm E_1^{-1}(\lambda\bm s_E-\bb\omega_E)+O_{\mathbb{F}_n}\left(\frac{1}{\sqrt{n}}\right)\geq 0,\\
	&-\lambda \leq \left\|\bb X_{-E}^T(\bb y-\bm X_E\bar{\bm\beta}_E)+\bm E_2\bm E_1^{-1}(\lambda\bm s_E-\bb\omega_E)+\bb\omega_{-E} +O_{\mathbb{F}_n}\left(\frac{1}{\sqrt{n}}\right)\right\|_{\infty} \leq \lambda.
\end{aligned}
\end{equation*}
hence the selection event is asymptotically affine in terms of $\bm D$.
\end{proof}

The next lemma shows that $\frac{1}{\sqrt{n}}\bb D$ is asymptotically linear, i.e.~$\bb D = \frac{1}{\sqrt{n}}\sum_{i=1}^n\bb\xi_i+o_{\mathbb{F}_n}(1)$, where $\bb\xi_i$ are measurable with respect to $(\bb X_i, y_i)$ and $\mathbb{E}_{\mathbb{F}_n}\bm\xi_i$ is the same across $i=1,\ldots, n$. Note that we need $1/\sqrt{n}$ in front of $\bm D$ since having $\bm X$ scaled with $1/\sqrt{n}$ gives the CLT for $\bm D-\mathbb{E}_{\mathbb{F}_n}[\bm D]$.

\begin{lemma} [Asymptotic linearity] \label{lemma:lasso:asymptotic:linearity}
 Assuming \ref{item:X:moments} and \ref{item:residuals:moments},
	test statistic $\frac{1}{\sqrt{n}}\bb D$ is asymptotically linear.
\end{lemma}

\begin{proof}
We have
\begin{align*}
	&\bar{\bb \beta}_E-\bm{\beta}_E^*
	=(\bm X_E^T\bm X_E)^{-1}\bm X_E^T\left(\bb y-\bm X_E\bm\beta_E^*\right)\\
	&=\bm E_1^{-1}\bm X_E^T\left(\bb y-\bm X_E\bm\beta_E^*\right) +\left((\bm X_E^T\bm X_E)^{-1}-\bm E_1^{-1}\right)\bm X_E^T\left(\bb y-\bm X_E\bm\beta_E^*\right)\\
	&=\bm E_1^{-1}\sum_{i=1}^n\left(y_i-\bm x_{E,i}^T\bm\beta_E^*\right)\bm x_{E,i}+O_{\mathbb{F}_n}\left(\frac{1}{\sqrt{n}}\right)O_{\mathbb{F}_n}(1),
\end{align*}
and the last term is $O_{\mathbb{F}_n}\left(\frac{1}{\sqrt{n}}\right)$, which implies $o_{\mathbb{F}_n}(1)$. The second term similarly satisfies
\begin{align*}
	&\bm X_{-E}^T\left(\bb y-\bm X_E\bar{\bm\beta}_E\right)
	-\bm X_{-E}^T\left(\bb y-\bm X_E\bm{\beta}_E^*\right)=\bm X_{-E}^T\bm X_E(\bm\beta_E^*-\bar{\bb \beta}_E)\\
	&=-\bm X_{-E}^T\bm X_E(\bm X_E^T\bm X_E)^{-1}\bm X_E^T\left(\bb y-\bm X_E\bm\beta_E^*\right)\\
	&=-\bm E_2\bm E_1^{-1}\bm X_E^T\left(\bb y-\bm X_E\bm\beta_E^*\right)-\left(\bm X_{-E}^T\bm X_E(\bm X_E^T\bm X_E)^{-1}-\bm E_2\bm E_1^{-1}\right)\bm X_E^T\left(\bb y-\bm X_E\bm{\beta}_E^*\right).
\end{align*}
The first term after the last equality is linear and the second term is of order $O_{\mathbb{F}_n}\left(\frac{1}{\sqrt{n}}\right)O_{\mathbb{F}_n}(1)$, hence also $o_{\mathbb{F}_n}(1)$.

\end{proof}

Assuming \eqref{assumption:mgf} and \eqref{assumption:norm:A_M}, the two lemmas above allow us to get the uniform validity of the plugin Gaussian pivot. To have the asymptotic uniform validity of the bootstrap pivot we additionally need the pre-selection bootstrap consistency of the target. In the LASSO example, it means that treating $E$ as fixed asymptotically we have that $\bar{\bm\beta}^*_E-\bar{\bm\beta}_E$ and $\bar{\bm\beta}_E-\bm\beta_E^*$ have the same asymptotic distribution, where $\bar{\bm\beta}^*$ denotes the bootstrap version of $\bar{\bm\beta}_E$. Since these assumptions are  standard in the bootstrap literature we omit them here.

\section{Additional example: Marginal screening} \label{sec:marginal:screening}

Selective inference in the nonrandomized setting for the marginal screening problem was considered in \cite{lee_screening} and the randomized one in \cite{selective_sampler}.  We compute the bootstrapped test statistic in the randomized setting. Nonrandomized marginal screening computes marginal $t$-statistics $\bm S_j = \frac{\bb X_j^T \bb y}{\hat{\sigma}_j}$, where $\hat{\sigma}_j$ is the variance estimates of $\bm X_j^T\bm y$, $j=1,\ldots, p$, and thresholds their absolute value at some threshold $c$, perhaps
$z_{1 -\alpha/2}$ where $\alpha$ is some nominal $p$-value threshold. Randomized marginal screening solves the following randomized problem
\begin{equation*}
	\hat{\bb\eta}(\bb S,\bb\omega) = \underset{\bb\eta\in\mathbb{R}^p, \|\bb\eta\|_{\infty} \leq c}{\textnormal{argmin}}\:\frac{1}{2}\|\bb\eta-\bb S\|^2_2 - \bb\omega^T\bb\eta,
\end{equation*}
with $\bb\omega\sim\mathbb{G}$ independent of the data, to get the active set $E$. 
Conditioning on the set $E$ achieving the threshold $c$ and their signs to be $\bb s_E$, we see that this event is
\begin{equation*}
	\widehat{M}_{(E,\bb s_E)} =\left\{(\bb S,\bb\omega):\hat{\bb \eta}_E(\bb S,\bb\omega) =c\cdot \bb s_E, \|\hat{\bb\eta}_{-E}(\bb S,\bb\omega)\|<c \right\}
\end{equation*}
or in terms of the data vector $T$ and optimization variables the selection event becomes
\begin{equation} \label{eq:ms:support}
\left\{(\bb S,\bb\eta,\bb z): \bb\eta_E = c\cdot \bb s_E, \text{diag}(\bb s_E) \bb z_E \geq 0, \|\bb\eta_{-E}\|_{\infty} < c, \bb z_{-E}=0 \right\}.
\end{equation}
Here $\bb z$ is the subgradient of the characteristic function corresponding to the set $\{\bb\eta\in\mathbb{R}^p:\|\bb\eta\|_{\infty}\leq c\}$.


Let us describe the joint density of the data and optimization variables, $(\bb \eta_{-E},\bb z_E)$, assuming the goal is inference for a linear combination of the parameter $\bb \beta^*_E$. Since
\begin{equation*}
	\bb S = \textnormal{diag}(\bm 1/\hat{\bm\sigma}) \bb X^T\bb y = \textnormal{diag}(\bm 1/\hat{\bm\sigma}) \begin{pmatrix}
		\bb X_E^T\bb X_E & \bb 0_{|E|\times (p-|E|)} \\ \bb X_{-E}^T\bb X_E & \bb I_{p-|E|}
	\end{pmatrix}\bb D =:-\bb M\bb D,
\end{equation*}
where $\textnormal{diag}(\bm 1/\hat{\bm\sigma})$ is a $p\times p$ matrix with diagonal elements $\hat{\sigma}_j$, $j=1,\ldots, p$, and zeros elsewhere, the randomization reconstruction map becomes
\begin{equation*}
	\omega(\bb D, \bb \eta_{-E}, \bb z_E) = \begin{pmatrix}
		c \cdot \bb s_E \\ \bb\eta_{-E}
	\end{pmatrix} + \bb M \bb D+\begin{pmatrix}
		\bb z_E \\ \bb 0
	\end{pmatrix}.
\end{equation*}
If we are testing $H_0:\bm A^T\bm\beta_E^*=\bm\theta$, using the decomposition as in the LASSO example above we sample $(\bm T,\bm\eta_{-E}, \bm z_E)$ from
\begin{equation*}
\phi_{(\bm\theta,\widehat{\bb\Sigma}_{\bm T})}(\bm T) \cdot g\left(\begin{pmatrix}
	 c\cdot \bb s_E \\ \bb\eta_{-E} \end{pmatrix} + \bm M\bm F+\bm M\widehat{\bm\Sigma}_{\bm D,\bm T}\widehat{\bm\Sigma}_{\bm T}^{-1}\bm T+ \begin{pmatrix}
 	\bb z_E \\ \bb 0 \end{pmatrix} \right),
\end{equation*}
with the support $\mathbb{R}^a\times[-c, c]^{p-|E|}\times \mathbb{R}_{\bb s_E}^{|E|}$.

Having the density above, we can use any of the samplers mentioned: optimization weighted sampler, wild bootstrap sampler or directly sampling from the density above (selective sampler). Here we describe the wild bootstrap density since the others should be clear. Using the wild bootstrap approximation, the bootstrapped randomization reconstruction map denote the randomization reconstruction map as
\begin{equation*}
\begin{aligned}
	\omega^B(\bb\alpha, \bb\eta_{-E}, \bb z_E) &= \begin{pmatrix}
		c\cdot\bb s_E \\ \bb \eta_{-E} \end{pmatrix} +\bm M\bm F+\bm M\widehat{\bm\Sigma}_{\bm D,\bm T}\widehat{\bm\Sigma}_{\bm T}^{-1}\bm T(\bm\alpha)+\bm M\widehat{\bm\Sigma}_{\bm D,\bm T}\widehat{\bm\Sigma}_{\bm T}^{-1}\bm\theta+\begin{pmatrix}
			\bb z_E \\ \bb 0 \end{pmatrix}.
\end{aligned}
\end{equation*}
Now we sample $(\bb \alpha,\bb\eta_{-E},\bb z_E)$ from the density proportional to
\begin{equation*}
	\left(\prod_{i=1}^n h_{\alpha}(\alpha_i)\right) \cdot g(\omega^B(\bb\alpha, \bb\eta_{-E}, \bb z_E))
\end{equation*}
and supported on $\textnormal{supp}(\mathbb{H}_{\alpha})^n\times [-c,c]^{p-|E|}\times \mathbb{R}_{\bb s_E}^{|E|}.$


\section{Sampling details} \label{app:sampling}

To sample from either the plugin CLT or the bootstrap density, we use projected Langevin MC to sample from a log-concave density with constraints. Given a convex set $K\subset\mathbb{R}^n$ with a nonempty interior, let us denote with $\mathcal{P}_K$ the projection onto $K$. We are interested in sampling from a density $f(\bb x)$ on $\mathbb{R}^n$ given by
\begin{equation*}
	\frac{df(\bb x)}{d\bb x} \propto\: e^{-\tilde{f}(\bb x)}\mathbb{I}_{\{\bb x\in K\}},
\end{equation*}
where $\tilde{f}:K\rightarrow\mathbb{R}$ is a convex and  differentiable function (or at least sub-differentiable).

Based on the previous point $\bb X_k$ in the chain, projected Langevin computes the next point using the following update
\begin{equation} \label{eq:langevin:update}
	\bb X_{k+1}=\mathcal{P}_K\left(\bb X_k-\eta\nabla \tilde{f}(\bb X_k)+\sqrt{2\eta}\bb \xi_k\right),
\end{equation}
where $\bb\xi_1,\bb\xi_2,\ldots$ are i.i.d.~sequence of standard normal in $\mathbb{R}^n$ and $\eta$ is the step size. Since our constraint region is simple due to the change of variables of \cite{selective_sampler}, the projection step is easy. \cite{bubeck2015sampling} prove that this chain will converge to a true density under some conditions.

Taking the randomization density $g$ and the bootstrap weights density $h_{\alpha}$ to be log-concave, we use the updates above to sample from the bootstrap densities. Let us now write the Langevin updates for the sampling of density some of the examples mentioned so far.


\begin{Example}[LASSO, introduced in Section \ref{sec:lasso}] \label{ex:lasso}
We write the Langevin updates for the samplers mentioned in the main text.
\begin{itemize}[leftmargin=*]
\item (Selective sampler)	
Let us first write down the projected Langevin update for sampling from the plugin CLT density in \eqref{eq:lasso:density:plugin:clt}.
\begin{equation*}
	\begin{pmatrix}
		\bb T^{(i+1)} \\ \bb\beta_E^{(i+1)} \\ \bb u_{-E}^{(i+1)}
	\end{pmatrix} = \begin{pmatrix} \bm T^{(i)}-\eta\widehat{\bm\Sigma}_{\bm T}^{-1}(\bb T^{(i)}-\bm\theta)
		 + \eta\widetilde{\bm M}^T\nabla_{\bb\omega}\log g(\omega(\bm T^{(i)}, \bm \beta_E^{(i)},\bm u_{-E}^{(i)}))+\sqrt{2\eta}\bb \xi^{(i)} \\ 
		\mathcal{P}_1\left(\bb\beta_E^{(i)}+\eta \bm B^T\nabla_{\bb\omega}\log g(\omega(\bm T^{(i)}, \bm \beta_E^{(i)},\bm u_{-E}^{(i)})) +\sqrt{2\eta}\bb\xi_{1}^{(i)}\right) \\ 
		\mathcal{P}_2\left(\bb u_{-E}^{(i)} +\eta\bm U^T\nabla_{\bb\omega}\log g(\omega(\bm T^{(i)}, \bm \beta_E^{(i)},\bm u_{-E}^{(i)})) +\sqrt{2\eta}\bb\xi_2^{(i)}\right)
	\end{pmatrix},
\end{equation*}
where $\omega(\bm T, \bm \beta_E,\bm u_{-E})=\widetilde{\bm M}\bm T+\bm B\bm\beta_E+\bm U\bm u_{-E}+\widetilde{\bm L}$
 and $(\bb\xi^{(i)},\bb\xi_1^{(i)},\bb\xi_2^{(i)})\sim \mathcal{N}(\bm 0,\bb I_a)\times\mathcal{N}(\bb 0, \bb I_{|E|})\times\mathcal{N}(\bb 0, \bb I_{p-|E|})$ and independent of all the other random variables. Here, $\mathcal{P}_1$ is the projection onto the orthant $\mathbb{R}^{|E|}_{\bb s_E}$ and $\mathcal{P}_2$ is the projection onto the cube $[-1,1]^{p-|E|}$. 
 
\item (Wild bootstrap sampler)
Now let us write down the projected Langevin update to sample from the bootstrap density in \eqref{eq:lasso:density:bootstrap}. Denoting $\widebar{\bm M} = \widetilde{\bm M}\bm A^T(\bm X_E^T\bm X_E)^{-1}\bm X_E^T\hat{\bm\varepsilon}$, we have $\omega^B(\bm\alpha,\bm\beta_E,\bm u_{-E})=\widebar{\bm M}\bm \alpha+\bm B\bm \beta_E+\bm U\bm\beta u_{-E}+\widetilde{\bm L}+\widetilde{\bm M}\bm\theta$ from \eqref{eq:lasso:random:reconstruction:boot}.
The gradient of the selective log-density becomes
\begin{equation*}
G\begin{pmatrix} \bm\alpha \\ \bm\beta_E \\ \bm u_{-E}\end{pmatrix} =\begin{pmatrix} \nabla_{\bb\alpha} \left(\sum_{j=1}^n\log h_{\alpha}(\alpha_j)\right)  + \widebar{\bm M}^T\nabla_{\bm\omega}\log g(\omega^B(\bb\alpha,\bb\beta_E,\bb u_{-E})) \\
	\bm B^T \nabla_{\bb\omega} \log g(\omega^B(\bb\alpha,\bb\beta_E,\bb u_{-E})) \\
	\bm U^T\nabla_{\bb\omega}\log g(\omega^B(\bb\alpha,\bb\beta_E,\bb u_{-E})) \end{pmatrix}.
\end{equation*}
Based on the current point $(\bb\alpha^{(i)}, \bb\beta_E^{(i)}, \bb u_{-E}^{(i)})$, the update becomes
\begin{equation*} 
	\begin{pmatrix}
		\bb\alpha^{(i+1)} \\ \bb\beta_E^{(i+1)} \\ \bb u_{-E}^{(i+1)}
	\end{pmatrix}=\mathcal{P}\left(\begin{pmatrix} \bb\alpha^{(i)} \\ \bb\beta_E^{(i)} \\ \bb u_{-E}^{(i)} \end{pmatrix} + \eta \cdot G\begin{pmatrix} \bb\alpha^{(i)} \\ \bb\beta_E^{(i)} \\ \bb u_{-E}^{(i)}\end{pmatrix} +\sqrt{2\eta}\bm \xi^{(i)} \right),
\end{equation*}
where $\bb\xi^{(i)}\sim \mathcal{N}(\bm 0,\bb I_{n+p})$ and independent of all the other random variables. Here, $\mathcal{P}$ projects $\bm\alpha$ onto $(\textnormal{supp}(\mathbb{H}_{\alpha}))^n$ and the optimization variables as above.

\item (Weighted optimization sampler) This sampler fixes $\bm T$ at its observed value while moving the optimization variables only with the updates similar to the ones in the selective sampler above with $\bm T^{(i)}=\bm T^{obs}$ throughout.

\end{itemize}
\end{Example}


\begin{Example}[Forward stepwise, introduced in Section \ref{sec:forward:stepwise}]
$\:$
\begin{itemize}[leftmargin=*]
\item (Selective sampler)	
To sample from the plugin CLT density in \eqref{eq:fs:density:plugin:clt}, the Langevin update at $(i+1)$-th step based on $(\bb T^{(i)}, \bb z_1^{(i)}, \ldots, \bb z_K^{(i)})$, the point at the $i$-th  step, is the following
\begin{equation*}
	\begin{pmatrix}
		\bb T^{(i+1)} \\ \bb z_1^{(i+1)} \\ \vdots \\ \bb z_K^{(i+1)}
	\end{pmatrix} = \begin{pmatrix}
		\bb T^{(i)}-\eta\widehat{\bm\Sigma}_{\bm T}^{-1}(\bm T^{(i)}-\bm\theta)+\eta\sum_{k=1}^K \widetilde{\bm M}_k^T\nabla_{\bm\omega_k}\log g(\omega_k(\bb T^{(i)},\bb z_k^{(i)})) +\sqrt{2\eta} \bb \xi^{(i)} 
		  \\ \mathcal{P}_1\left(\bb z_1^{(i)} +\eta \nabla_{\bb\omega_1}\log g_1(\omega_1(\bb T^{(i)},\bb z_1^{(i)})) +\sqrt{2\eta}\bb \xi^{(i)}_1 \right) \\ \vdots \\ \mathcal{P}_K\left(\bb z_K^{(i)} +\eta \nabla_{\bb\omega_K} \log g_K(\omega_K(\bb T^{(i)}, \bb z_K^{(i)}) +\sqrt{2\eta}\bb \xi^{(i)}_K\right)
	\end{pmatrix},
\end{equation*}
where $\omega_k(\bm z_k, \bm T)=\bm z_k+\bm M_k\bm F+\widetilde{\bm M}_k\bm T$  and $(\bb\xi^{(i)}, \bb\xi^{(i)}_1, \ldots, \bb\xi^{(i)}_K)\sim\mathcal{N}(\bb 0, \bb I_a)\times\mathcal{N}(\bb 0, \bb I_{p})\times \ldots\times \mathcal{N}(\bb 0, \bb I_{p-K+1})$ and independent of all the other variables. Here, $\mathcal{P}_k$ denotes the projection onto $\partial I_{\mathcal{B}_k}(\hat{\bb\eta}_k)$ for $k=1,\ldots, K$. 

\item (Wild bootstrap sampler)
Using wild bootstrap, the randomization reconstruction map at $k$-th step becomes
\begin{equation*}
	\omega_k^B(\bb\alpha, \bb z_k) = \bb z_k+\bm M_k\bb F+\widetilde{\bb M}_k\bm\theta+\widetilde{\bb M}_k \bm T(\bm\alpha) = \bb z_k+\bm M_k\bb F+\widetilde{\bb M}_k\bm\theta+\widebar{\bb M}_k\bm\alpha,
\end{equation*}
where $\widebar{\bm M}_k=\widetilde{\bm M}_k\bm A^T(\bm X_E^T\bm X_E)^{-1}\bm X_E^T\textnormal{diag}(\hat{\bm\varepsilon})$, hence linear in the bootstrap weights $\bm\alpha$ and the sub-gradient $\bm z_k$. Taking into account all the randomization reconstruction maps for $k=1,\ldots, K$, the bootstrap density of $(\bb\alpha,\bb z_1,\ldots,\bb z_K)$ is proportional to
\begin{equation} \label{eq:fs:density:bootstrap}
	\left(\prod_{i=1}^n h_{\alpha}(\alpha_i)\right) \cdot \prod_{k=1}^K g_k\left(\omega_k^B(\bb\alpha, \bb z_k)\right)
\end{equation}
and supported on $(\textnormal{supp}(\mathbb{H}_{\alpha}))^n\times \prod_{i=1}^n\partial I_{\mathcal{B}_k}(\hat{\bb\eta}_k)$. 
The gradient of the selective log-density is
\begin{equation*}
	G\begin{pmatrix} \bb \alpha \\ \bb z_1 \\ \vdots \\ \bb z_K
	\end{pmatrix} = \begin{pmatrix}
		\nabla_{\bb\alpha}\left(\sum_{j=1}^n\log h_{\alpha}(\alpha_j)\right)+\sum_{k=1}^K \widebar{\bm M}_k^T\nabla_{\bm\omega_k}\log g(\omega_k^B(\bb \alpha,\bb z_k)) \\ 
		\nabla_{\bb\omega_1}\log g_1(\omega^B_1(\bb\alpha,\bb z_1)) \\ \vdots \\ \nabla_{\bb\omega_K} \log g_K(\omega^B_K(\bb\alpha, \bb z_K)
	\end{pmatrix}.
\end{equation*}
Given the gradient, the Langevin updates are straightforward.

\end{itemize}
\end{Example}


\begin{Example}\textnormal{\textbf{(Multiple views with GLMs and group LASSO penalty, introduced in Section \ref{sec:mv:glm})}}
$\:$
\begin{itemize}[leftmargin=*]
\item (Selective sampler)
We write the sampling updates for the density in \eqref{eq:mv:glm:target:density}. At each step of the Langevin MC, we move the target vector  $\bm T$ and the optimization variables $\left\{(\gamma_{k,g})_{g\in A_k}, (\bm z_{k,h})_{h\in -A_k}\right\}_{k=1}^K$. The gradient of the logarithm of the selective density in \eqref{eq:mv:glm:target:density} becomes
\begin{equation*}
G\begin{pmatrix}
		\bm T \\ (\gamma_{1,g})_{g\in A_1} \\ (\bm z_{1,h})_{h\in -A_1} \\ \vdots \\ (\gamma_{K,g})_{g\in A_K} \\ (\bm z_{K,h})_{h\in -A_K}
	\end{pmatrix} = \begin{pmatrix}
		-\widehat{\bm\Sigma}_{\bm T}^{-1}(\bm T-\bm\theta)+\sum_{k=1}^K \left(\widetilde{\bm M}_k^T\nabla\log g_k(\bm\omega_k)+\nabla_{\bm T}\log J_k\right) \\ 
		\bm\Gamma_1^T\nabla\log g_1(\bm\omega_1)+\nabla_{(\gamma_{1,g})_{g\in A_1}}\log J_1 \\ 
		\bm Z_1^T\nabla\log g_1(\bm\omega_1) +\nabla_{(\bm z_{1,h})_{h\in -A_1}}\log J_1 \\ 
		\vdots \\ 
		\bm\Gamma_K^T\nabla\log g_K(\bm\omega_K)+\nabla_{(\gamma_{K,g})_{g\in A_K}}\log J_K \\ 
		\bm Z_K^T\nabla\log g_K(\bm\omega_K)+\nabla_{(\bm z_{K,h})_{h\in -A_K}}\log J_K
	\end{pmatrix},
\end{equation*}
where $\bm\omega_k$ is the randomization reconstruction from \eqref{eq:randomization:reconstruction}, $k=1,\ldots, K$. Given the gradient, the updates follow easily from \eqref{eq:langevin:update}.

\item (Wild bootstrap sampler) 
Using the wild bootstrap, $Q_E(\bm\beta_E^*)^{-1}\bm X_E^T\textnormal{diag}(\hat{\bm\varepsilon})\bm\alpha,$ 
where $\bm\alpha=(\alpha_1,\ldots, \alpha_n)\in\mathbb{R}^n$ with $\alpha_i\overset{i.i.d.}{\sim}\mathbb{H}_\alpha$, approximates the distribution of $\bar{\bm\beta}_E-\bm\beta_E^*$. Denote $T(\bm\alpha) = \bm A^TQ_E(\bm\beta_E^*)^{-1}\bm X_E^T\textnormal{diag}(\hat{\bm\varepsilon})$.
Replacing $\bm T$ with $T(\bm\alpha)+\bm\theta$, the sampling density on the bootstrap weights $\bm\alpha\in\mathbb{R}^n$ and the optimization variables $(\gamma_{k,g})_{g\in A_k}$ and $(\bm z_{k,g})_{g\in -A_k}$, $k=1,\ldots, K$, becomes
\begin{equation} \label{eq:mv:glm:boot:density}
	\left(\prod_{i=1}^n h_{\alpha}(\alpha_i)\right)\cdot\left(\prod_{k=1}^K g_k\left(\omega_k(\widebar{\bm M}_k\bm\alpha+\bm F_k(\gamma_{k,g})_{g\in A_k}+\bm Z_k(\bm z_{k,g})_{h\in -A_k})+\widebar{\bm L}_k\right)\cdot J_k\right),
\end{equation}
where $\widebar{\bm M}_k=\widetilde{\bm M}_k\bm A^TQ_E(\bm\beta_E^*)^{-1}\bm X_E^T\textnormal{diag}(\hat{\bm\varepsilon}),\;\; \widebar{\bm L}_k=\widetilde{\bm L}_k+\widetilde{\bm M}\bm\theta$,
with the weights restricted to $\bm\alpha\in\textnormal{supp}(\mathbb{H}_{\alpha})^n$ and the optimization variables as in \eqref{eq:constrains:inactive:and:norm}.
To sample from the bootstrap density in \eqref{eq:mv:glm:boot:density}, we use Langevin MC where at each step we move the bootstrap weights, vector $\bm \alpha\in\mathbb{R}^n$, and the optimization variables 
$\left\{(\gamma_{k,g})_{g\in A_k}, (\bm z_{k,h})_{h\in -A_k}\right\}_{k=1}^K$. 
The gradient of the logarithm of the selective density from \eqref{eq:mv:glm:boot:density} becomes
\begin{equation*}
G^b\begin{pmatrix}
		\bm\alpha \\ (\gamma_{1,g})_{g\in A_1} \\ (\bm z_{1,h})_{h\in -A_1} \\ \vdots \\ (\gamma_{K,g})_{g\in A_K} \\ (\bm z_{K,h})_{h\in -A_K}
	\end{pmatrix} = \begin{pmatrix}
		\left(\nabla_{\bm\alpha}\sum_{i=1}^n h_{\alpha}(\alpha_i)\right)+\sum_{k=1}^K \left(\widebar{\bm M}_k^T\nabla\log g_k(\bm\omega_k)+\nabla_{\bm T}\log J_k\right) \\ 
		\bm\Gamma_1^T\nabla\log g_1(\bm\omega_1)+\nabla_{(\gamma_{1,g})_{g\in A_1}}\log J_1 \\ 
		\bm Z_1^T\nabla\log g_1(\bm\omega_1) +\nabla_{(\bm z_{1,h})_{h\in -A_1}}\log J_1 \\ 
		\vdots \\ 
		\bm\Gamma_K^T\nabla\log g_K(\bm\omega_K)+\nabla_{(\gamma_{K,g})_{g\in A_K}}\log J_K \\ 
		\bm Z_K^T\nabla\log g_K(\bm\omega_K)+\nabla_{(\bm z_{K,h})_{h\in -A_K}}\log J_K
	\end{pmatrix},
\end{equation*}
where $\bm\omega_k$ is the randomization reconstruction map from \eqref{eq:randomization:reconstruction}, $k=1,\ldots, K$. In the case of standard normal bootstrap weights $\nabla_{\bm\alpha}\sum_{i=1}^n h_{\alpha}(\alpha_i)=-\bm\alpha$.

\end{itemize}
\end{Example}


\begin{remark1}
There are both theoretical and practical considerations when choosing the density of the randomization, $g$, and the density of bootstrap weights, $h_{\alpha}$. We take $g$ to be Laplace, logistic or Gaussian and $h_{\alpha}$ to be the standard normal.
\end{remark1}